\documentclass[prodmode,acmtcps]{acmsmall} 
\def\LB{\left(}         
\def\RB{\right)}        

\newfont{\bbb}{msbm10 scaled 500}

\newcommand{\cv}{{\bf c}}

\newcommand{\iv}{{\bf i}}

\newcommand{\nv}{{\bf n}}

\newcommand{\sv}{{\bf s}}

\newcommand{\vv}{{\bf v}}
\newcommand{\xv}{{\bf x}}
\newcommand{\yv}{{\bf y}}
\newcommand{\zv}{{\bf z}}

\newcommand{\Gm}{{\bf G}}
\newcommand{\Hm}{{\bf H}}
\newcommand{\Id}{{\bf I}}

\newcommand{\Wm}{{\bf W}}
\newcommand{\Vm}{{\bf V}}

\newcommand{\Ym}{{\bf Y}}

\newcommand{\Gammam}{\hbox{\boldmath$\Gamma$}}

\newcommand{\Sigmam}{\hbox{\boldmath$\Sigma$}}

\usepackage{comment}
\usepackage{amsmath}
\usepackage{subcaption}
\newcommand{\beqa}{\begin{eqnarray}}
\newcommand{\eeqa}{\end{eqnarray}}
\newcommand{\dsp}{\displaystyle}
\usepackage{multirow}
\usepackage{hhline}
\usepackage{array}
\newcolumntype{M}[1]{>{\centering\arraybackslash}m{#1}}
\newcommand{\thickhline}{%
    \noalign {\ifnum 0=`}\fi \hrule height 1pt
    \futurelet \reserved@a \@xhline
}

\usepackage[ruled]{algorithm2e}

\SetAlFnt{\small}
\SetAlCapFnt{\small}
\SetAlCapNameFnt{\small}
\SetAlCapHSkip{0pt}
\IncMargin{-\parindent}

\acmVolume{1}
\acmNumber{1}
\acmArticle{1}
\acmYear{2017}
\acmMonth{1}

\doi{0000001.0000001}

\issn{1234-56789}

\begin{document}


\markboth{Lakshiminarayana et.al.}{Modeling and Detecting False Data Injection Attacks against Railway Traction Power Systems}

\title{Modeling and Detecting False Data Injection Attacks against Railway Traction Power Systems}
\author{Subhash Lakshminarayana
\affil{Advanced Digital Sciences Center, Illinois at Singapore}
Teo Zhan Teng
\affil{GovTech, Singapore}
Rui Tan
\affil{Nanyang Technological University, Singapore}
David K.Y. Yau
\affil{Singapore University of Technology and Design}
}

\begin{bottomstuff}
This work was supported in part by the National Research
Foundation (NRF), Prime Minister's Office, Singapore, under
its National Cybersecurity R\&D Programme (Award
No. NRF2014NCR-NCR001-31) and administered by the
National Cybersecurity R\&D Directorate and in part by a
Start-up Grant at NTU.

Author's addresses: S. Lakshminarayana, Advanced Digital Sciences Center, Illinois at Singapore (e-mail: subhash.l@adsc.com.sg); Z.T. Teo, GovTech Singapore (e-mail: teozt@hotmail.com); R. Tan,  School of Computer Science and Engineering, Nanyang Technological University, Singapore (e-mail: tanrui@ntu.edu.sg); D.K.Y. Yau, Singapore University of Technology and Design (e-mail: david\_yau@sutd.edu.sg). The work was conducted when Z.T. Teo was with the Advanced Digital Sciences Center, Illinois at Singapore.
\end{bottomstuff}

\begin{abstract}
Modern urban railways extensively use computerized sensing and control technologies to achieve safe, reliable, and well-timed operations. However, the use of these technologies may provide a convenient leverage to cyber-attackers who have bypassed the air gaps and aim at causing safety incidents and service disruptions. In this paper, we study false data injection (FDI) attacks against railways' traction power systems (TPSes). Specifically, we analyze two types of FDI attacks on the train-borne voltage, current, and position sensor measurements -- which we call {\em efficiency attack} and {\em safety attack} -- that (i) maximize the system's total power consumption and (ii) mislead trains' local voltages to exceed given safety-critical thresholds, respectively. To counteract, we develop a global attack detection (GAD) system that serializes a {\em bad data detector} and a novel {\em secondary attack detector} designed based on unique TPS characteristics. With intact position data of trains, our detection system can effectively detect the FDI attacks on trains' voltage and current measurements even if the attacker has full and accurate knowledge of the TPS, attack detection, and real-time system state. 
In particular, the GAD system features an adaptive mechanism that ensures low false positive and negative rates in detecting the attacks under noisy system measurements.
Extensive simulations driven by realistic running profiles of trains verify that a TPS setup is vulnerable to the FDI attacks, but these attacks can be detected effectively by the proposed GAD while ensuring a low false positive rate.
\end{abstract}

\maketitle

\section{Introduction}
\label{sec:intro}

In modern cities, safe, reliable, and well-timed operations of urban railways are critical. A modern railway is a highly complex cyber-physical system (CPS) consisting of diverse subsystems including train motion control, traction powering, signaling, etc, where deeply embedded information and communication technologies (ICTs) are used to operate each train and connect trains to an operation center.
The extensive use of ICT may provide a convenient leverage to attackers, however, who may aim to hurt passengers' safety or cause widespread service disruptions. To date, the cybersecurity of modern railways has relied on air gaps that isolate their ICT systems from public networks. However, recent high-profile intrusions such as Stuxnet \cite{karnouskos2011} and Dragonfly \cite{dragonfly2014} have successfully breached the air gaps of critical CPS infrastructures and resulted in physical damage.
For instance, the Stuxnet worm damaged nuclear centrifuges by injecting false control commands and forging normal system states. Its design and architecture are not domain-specific -- they can be readily customized against other types of CPS including transportation~\cite{karnouskos2011}. Insider attacks represent another major threat to air-gapped systems; their severe consequences have likewise been well documented~\cite{insider2011}. It is thus critical to understand the cybersecurity risks of modern railways as a mission-critical CPS, and develop effective security countermeasures in their ICT design.

In this paper, we study the cybersecurity of direct current (dc) traction power systems (TPSes) that are widely deployed in urban electrical railways. The criticality of TPS is evidenced by prior severe incidents caused by TPS malfunctions. The 2014 Moscow metro derailment that led to 24 dead and 160 injured was caused by sudden braking of the train in question, when its traction voltage dropped abruptly~\cite{moscowtimes}. In Singapore, a system-wide metro service disruption, triggered by TPS faults, affected almost half a million commuters during rush hours on July 7, 2015~\cite{SMRT15}. Moreover, the computerized sensing and control in an automated TPS could be prime targets for cyber-attackers, who can sabotage the control and steer the system into inefficient and unsafe states.

Motivated by Stuxnet worm-type attacks that forge physical system states, in this paper we study a general class of integrity attacks called {\em false data injection} (FDI), which tampers with train-borne sensor measurements required by TPS control. In a TPS, the electricity power supplied by substations is delivered by a network of overhead lines and third rails to the trains. According to its operation mode, a train's power consumption can be highly dynamic. In traction mode, it draws power from the TPS, causing a drop in the train's local voltage; in braking mode, it regenerates electricity from kinetic energy and injects this electricity back to the TPS,\footnote{In electrical railways, trains are often equipped with regenerative brakes that generate electricity in deceleration~\cite{Fletcher1991}.} causing a rise in the voltage. To prevent the voltage from exceeding safety-critical thresholds, trains apply {\em overcurrent control} and {\em squeeze control} \cite{OkadaKoseki2004} to throttle their power draw and injection, respectively. As these controls are performed based on train-borne voltage and current sensor measurements, FDI attacks on the measurements may mislead the train into erroneous power control decisions, which may in turn produce damaging and even catastrophic physical impacts on the train and the TPS. Recent results show that the measurements can be compromised in practice by precisely controlled electromagnetic interference in analog sensors \cite{kune2013ghost}, hardware trojans in chips \cite{HardwareTrojans2010}, and malware infections in sensor firmwares \cite{SmartMeterSecurity2009,Theft2011,davis2009}. Hence, FDI attack is a clear and present threat that requires immediate attention.

In this paper, we aim to answer the following two fundamental research questions:

{\em (1) How to characterize the impact of FDI attacks on TPS system efficiency and safety?}  Analysis of the impact based on an essential TPS model will provide basic understanding for developing countermeasures. However, the analysis is difficult, due to complex system dynamics arising from the trains' motion. In particular, a moving train does not only act as ``load'' and ``generation'' alternately over time, but it also alters the power network's topology and electrical parameters continually. Moreover, because different TPS components (trains, substations, etc.) become physically interconnected through a common underlying power network, effects of an erroneous power control on a train during attack may propagate to the neighboring TPS components. The analysis must address these intricate and unique characteristics of TPSes.

{\em (2) How to develop effective approaches for detecting the FDI attacks?} Our thesis is that, because measurements from different trains are inherently correlated through interconnection over the same power network, for attack resilience we can apply a global detection that cross-checks the measurements collected from all trains based on an {\em a priori} global TPS model. However, in contrast to alternating current (ac) power grids that have well-established centralized monitoring and sensor data cross-check safeguards for reliable holistic control \cite{RahmanFormalModel2014,LiuNingReiter2009}, TPS is mainly concerned with individual trains' local operation (i.e., the overcurrent and squeeze controls), and therefore it is not traditionally subject to any global sensor data checks across trains. Thus, an existing dc TPS operation center seldom scrutinizes the sensor measurements, beyond their display and presentation for human operators. In this paper, we demonstrate the importance of these global, but hitherto ignored, sensor data cross checks in the TPS domain against FDI attacks.

In answering the above two research questions, our main contributions in this paper are as follows:

First, based on essential models of power substations, power flows, and train overcurrent and squeeze controls in a TPS, we formulate two types of FDI attacks that we call {\em efficiency attack} and {\em safety attack}. These attacks (i) maximize the total instantaneous power consumption of the TPS and (ii) mislead victim trains' local voltages to exceed given safety-critical thresholds, respectively.
Efficiency attacks will increase the train's traction power consumption, resulting in an increase in railways' operation expenses.\footnote{Energy costs of running urban rail pose a significant financial burden to transport companies, constituting about $20 \%$ of their operational expenses \cite{Osiris2015}. Of this, about $80 \%$ of energy is consumed for traction (e.g. train's motion, braking, electric losses) \cite{GonzálezGil2014509}.} 
Efficiency attacks will also potentially increase the carbon footprint of the transportation sector, which is an important consideration for railway operators \cite{London_Metro2008}.
On the other hand, safety attacks may trip circuit breakers, causing dangerous power loss and brake malfunction.

The efficiency attack formulation models an aggressive attacker who aims at maximizing the attack impact and provides insights into understanding the performance degradation limit caused by FDI attacks. Numerical results for a TPS section with two substations and two trains show that the efficiency attack can result in an instantaneous efficiency loss of about 20\%, whereas the safety attack on a single train can indeed lead to significant safety breaches. These results substantiate the potency of FDI attacks on train-borne sensor measurements.

Second, we propose to apply a global {\em bad data detection} (BDD) method, similar to that widely used in ac power grids \cite{LiuNingReiter2009}, to detect FDI attacks in a dc TPS. Despite a known vulnerability of the BDD -- it can be bypassed by an attacker who knows enough details of its design -- our numerical results show that, in order for an FDI attack to be stealthy against the BDD, it will have to settle for a significantly reduced damage on the system efficiency. Moreover, we observe that, given intact position data of trains, solutions of the BDD bypass condition will become discrete. Based on this observation, we develop a novel {\em secondary attack detection} (SAD) algorithm that can effectively detect the onset of an FDI attack on trains' voltage and current measurements after it has bypassed the BDD.
Hence, the BDD and the SAD form in tandem a global attack detector (GAD) under the Kerckhoffs's assumption (i.e., the attacker has full and accurate knowledge of the system model, attack detection, and real-time system state), provided that the integrity of trains' position information can be verified. Building on this result, we design an approach to mitigating the impact of an attack after its detection.

Third, we report extensive simulations, driven by realistic profiles of trains in operation, to evaluate our solutions. For a TPS consisting of four trains each running over a distance of ten kilometers for $800$ seconds, our results show that, without the global BDD, FDI attacks can increase the total system energy consumption by $28.3\%$ and breach the system's safety
condition. After applying the BDD, the system's total energy consumption increases by no more than $6.2\%$ under the efficiency attack, and safety attacks become no longer successful. Moreover, the proposed SAD algorithm achieves a detection probability of 96\% in detecting the onsets of the FDI attacks that have successfully bypassed the BDD.

Finally, we investigate the false positives (FPs) and missed detections (MDs) of the proposed detectors in the presence of sensor measurement noises. Simulation results illustrate that although the GAD yields low FP and MD rates during most of the simulation time, it gives a relatively 
high FP rates for a few short time durations when one or more trains change their status of motion (e.g., from tractioning mode to braking mode). To maintain a low FP rate all the time,
we propose an adaptation mechanism based on an attack detection window for the GAD. We call the improved attack detection system GAD-W. Simulation results show that with appropriately chosen detector parameters, the GAD-W detector achieves an average FP rate of $9 \times 10^{-4}$ and an MD rate of $7 \times 10^{-4}$ over the entire simulation time.

This work focuses on attacks against urban metros (e.g., Tokyo, Singapore, and Berlin) that adopt dc systems. Thus, our analysis is based on the dc TPS model. On the other hand, long-distance railways usually adopt an alternating current (ac) TPS, due to higher efficiency in transmitting ac power over long distances \cite{RailTechnical,Overhead}.
Although a detailed investigation of cybersecurity issues in ac traction power systems is beyond the scope of our paper, we conjecture that the vulnerabilities of the two systems are similar. This is because ac and dc TPSes mainly differ in their design of electrical components (e.g., substation and train motor) \cite{Overhead}, while the  ICT infrastructures in these two kinds of systems are similar. Thus, the attack surfaces of the cyber infrastructures in both cases are the same. Nevertheless, the attack impact analysis and the detector design may differ in details, which are left for future work.

The balance of the paper is organized as follows. Section~\ref{sec:RelatedWork} reviews related work. Section~\ref{sec:Tract_Power} 
describes our TPS model. Section~\ref{sec:Cyberattacks} formulates the efficiency and safety attacks. Section~\ref{sec:BDD} analyzes the effectiveness of the BDD and presents the proposed SAD algorithm that complements the BDD. Section~\ref{sec:Noise_Analysis} analyzes the impact of sensor measurement noises. Section~\ref{sec:SimRes} presents simulation results. Section~\ref{sec:Conclusion} concludes.

\section{Related Work}
\label{sec:RelatedWork}

Power flow analysis and optimization for TPS have received increasing research interest.
Power flow analysis is a basic tool for TPS
planning and operation.
Prior work has analyzed dc power flows \cite{CaiIterative1995,ArboleyaBFS2015,PiresICCG2007} and addressed the 
interactions between the dc TPS and a supporting ac power grid \cite{AbrahamssonThesis2012}, \cite{ArboleyaCotoTVT2012}.
We adopt existing electrical models for different TPS components \cite{CaiIterative1995}, \cite{ArboleyaBFS2015}, \cite{PiresICCG2007} in this work. These models provide sufficient accuracy generally \cite{ArboleyaBFS2015}, and they are tractable for analysis.
Based on power flow analysis, recent research has tried to improve the energy efficiency 
of railways by leveraging trains' power regeneration \cite{BBC_Regen2015}.
Techniques such as synchronizing the trains' speed profiles \cite{Miyatake2010,SuTang2014,SuTangRoberts2015} and real-time substation voltage control \cite{RaghunathanCOMPRAIL2014} have been shown to provide efficient reuse of the regenerated power. 
To the best of our knowledge, none of the existing studies on TPS control have addressed it from a cybersecurity perspective.
The security problem is imperative, since TPS is a form of critical infrastructure that renders it an attractive target for attacks.

Different types of CPS can have vastly different properties and characteristics, and their security concerns and admissible detection and mitigation strategies can be totally different. Typically, their cybersecurity analysis must be carried out in a domain specific manner, with customized considerations given to main details and semantics of specific systems. 
C{\'a}rdenas et al.~\cite{cardenas2011attacks} investigate the impacts of integrity and denial-of-service attacks on the process control system, which has multiple sensors and control loops, of a chemical reactor. Amin {et al.}~\cite{amin2013cyber} perform security threat assessment of supervisory control and data acquisition systems for water supply. Other efforts~\cite{LiuNingReiter2009,KimTongTopology2013} have analyzed FDI attacks against ac utility power grids.
They show that an attacker capable of tampering with grid sensor measurements or topology information can carefully construct attacks to bypass detection by certain existing fault data detectors. Recent studies have investigated the impact of such stealthy attacks on
grid power flows \cite{RenLoadRedis2011,OPFClosingLoop2012,RahmanFormalModel2014}.
They show that maliciously biased estimates of the system state can cause grid operators to make erroneous decisions that will lead to degraded performance or safety breaches.
This paper is the first to analyze the efficiency and safety of TPS under FDI attacks. We provide new and non-trivial domain-specific modeling and analysis to capture the targeted application's unique features and key properties. 
In particular, TPS involves real-time and complex interactions between two highly dynamical physical systems, namely a mechanical system of the trains' motion and an electrical system that governs the trains' power consumption and regeneration during this motion. Attackers could exploit the interactions to strengthen their attacks.

\section{Traction Power System Model}
\label{sec:Tract_Power}

In this section, we present a model of a dc TPS at a certain time instant. The TPS is modeled as a power network
consisting of $N$ nodes. Denote by $\mathcal{N} = \{ 1,2,\dots,N\}$ the set of nodes
and $\mathcal{L}$ the set of resistive branches connecting the nodes. 
The substations and the trains are connected to different nodes.
The sets of nodes for the substations, the tractioning trains, and the regenerating trains are denoted by $\mathcal{N}_{\text{sub}},$ $\mathcal{N}_{\text{tra}},$  and $\mathcal{N}_{\text{reg}}$, respectively. 
We define $\mathcal{N}_{\text{trains}} = \mathcal{N}_{\text{tra}} \cup \mathcal{N}_{\text{reg}}.$
The positions of the nodes $1,\dots,N$ are denoted by a set $ \sv = \{ s_1, s_2, \dots , s_N\},$ where
$s_1$ is fixed at zero and $s_i$ is the distance from node $i$ to node $1.$
Fig.~\ref{fig:DCTractCkt} illustrates a TPS section with two substations at nodes $1$ and $4$, as well as two
trains at nodes $2$ and $3.$ In this example, the train at node $2$ is tractioning and the train at node $3$ is braking and regenerating. Therefore,
$\mathcal{N}  = \{ 1,2,3,4 \},  \mathcal{L}  = \{ (1,2), (2,3), (3,4) \}, \mathcal{N}_{\text{sub}}  = \{ 1,4 \}, \mathcal{N}_{\text{tra}}  = \{ 2 \}, \mathcal{N}_{\text{reg}}  = \{ 3 \} .$ The electrical models for the power network, substations, and trains are described as follows.

\begin{minipage}[!t]{\textwidth}
\begin{minipage}[!t]{0.42\textwidth}

\centering
\includegraphics[width=1\textwidth]{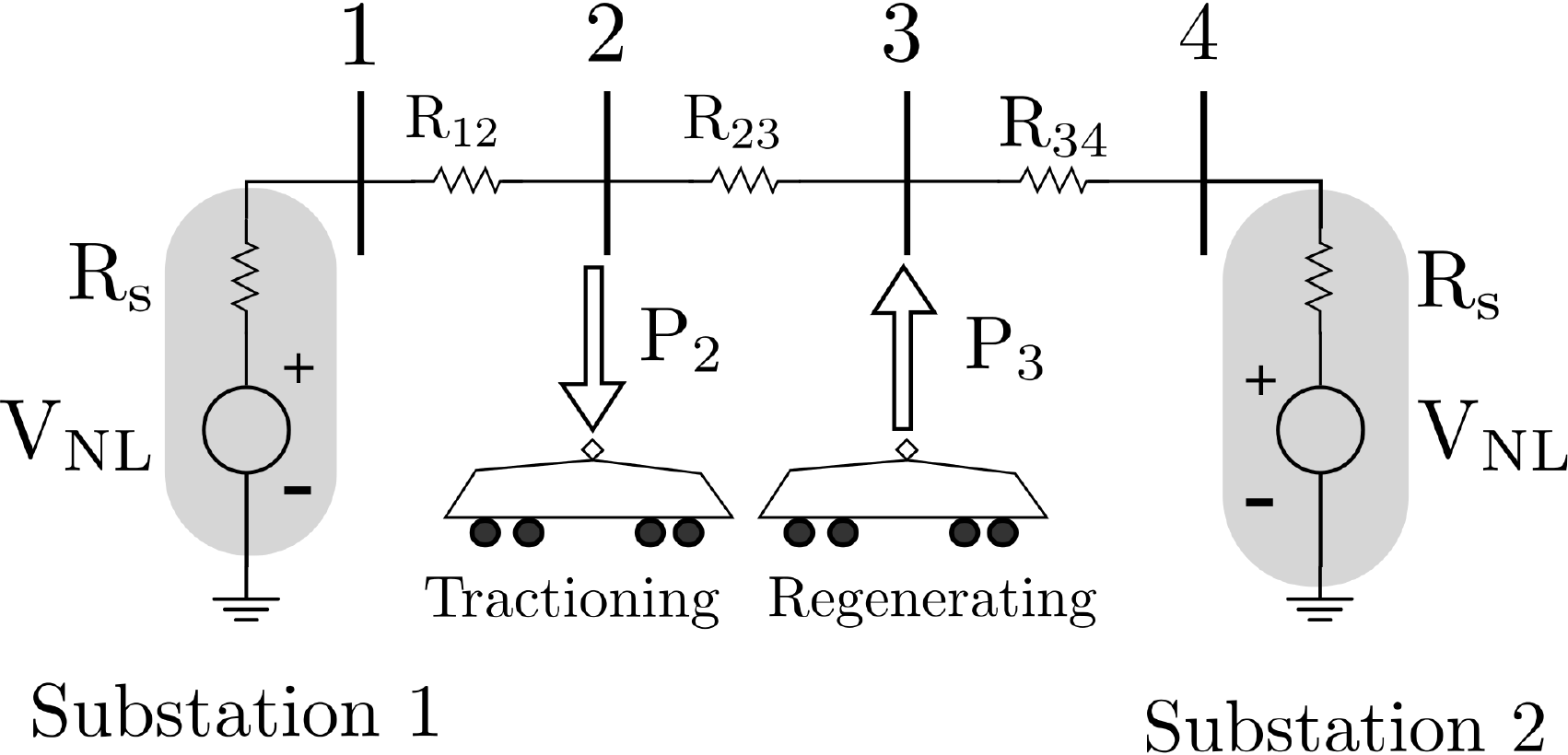}
\captionof{figure}{Illustration of a TPS section.}
\vspace{-1em}
\label{fig:DCTractCkt}

\end{minipage}
~
\begin{minipage}[!t]{0.54\textwidth}
\centering
\begin{minipage}{0.5\textwidth}
\includegraphics[width=1\textwidth,trim={0 2.5cm 0 0}]{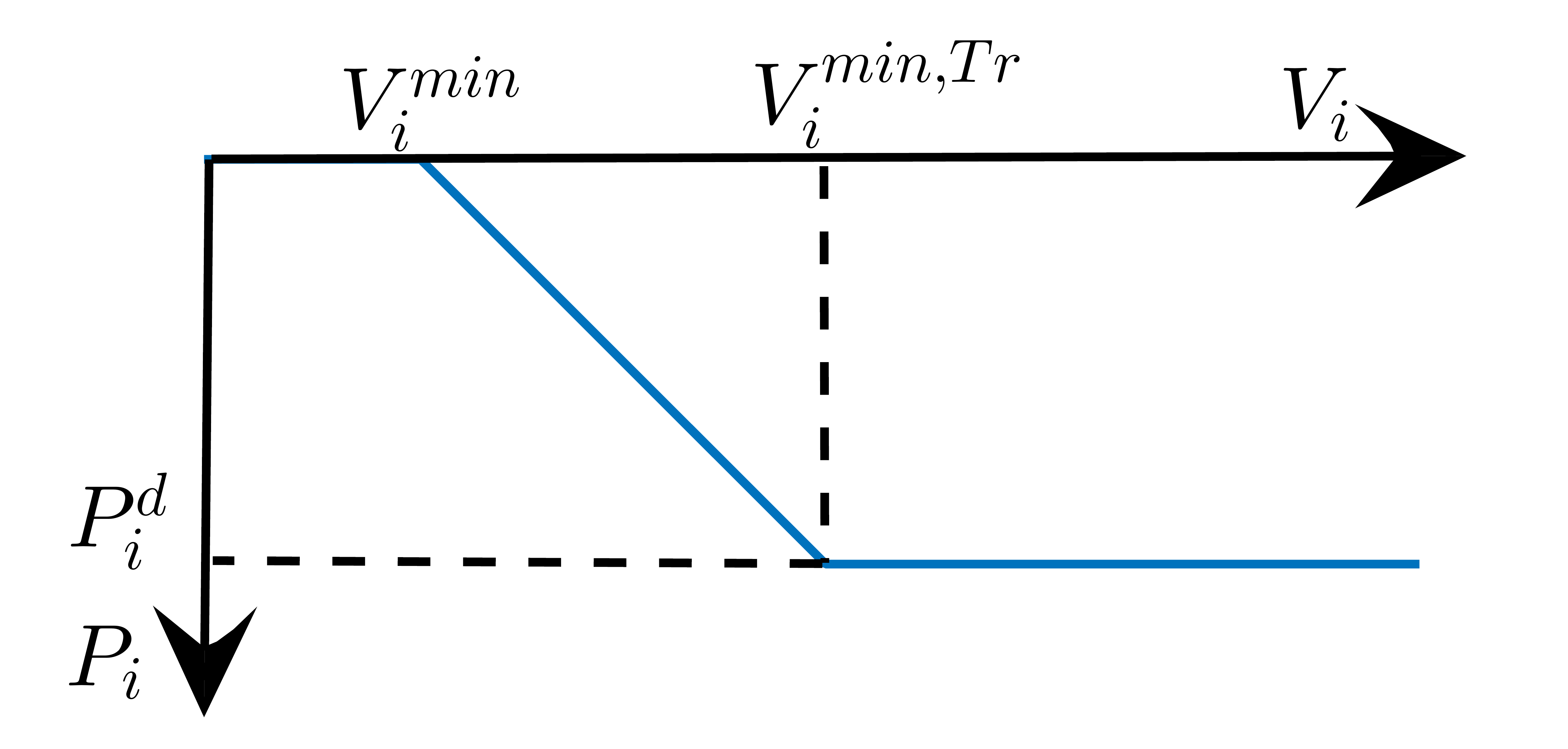}
\captionof{subfigure}{}
\end{minipage}
~
\begin{minipage}{0.5\textwidth}
\centerline{\includegraphics[width=1\textwidth,trim={0 2.5cm 0 0}]{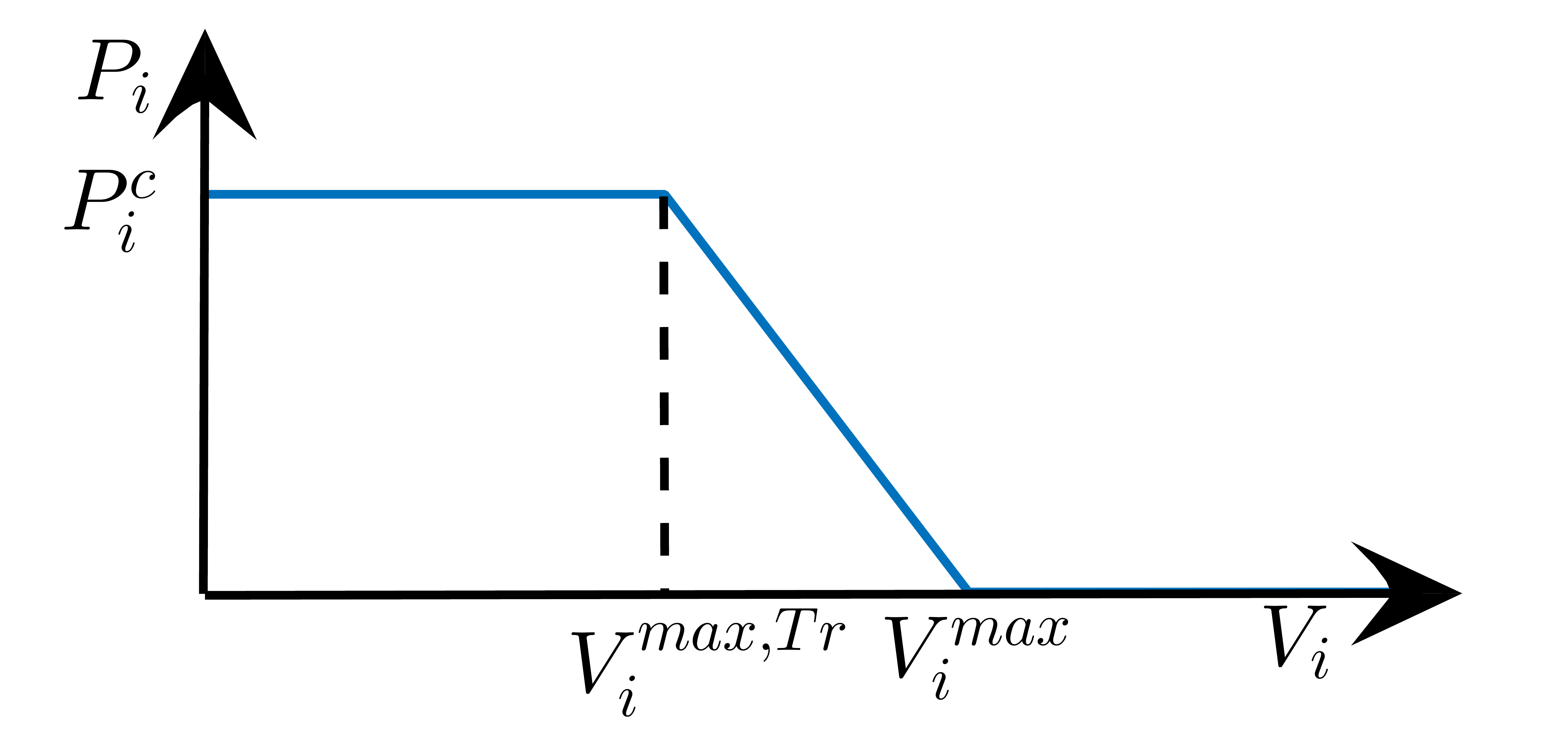}}
\captionof{subfigure}{}
\end{minipage}
\captionof{figure}{(a) Overcurrent control. (b) Squeeze control.}
\label{fig:Acceleration_Control}
\end{minipage}

\end{minipage}

\vspace{0.5em}\noindent
{\bf Power network:}
Let $V_i$ and $I_i$ denote the voltage and current at node $i$, respectively, and $\vv$ and $\iv$ the vectors of the nodal voltages and currents. For safe operations, all nodal voltages must be within a safety limit, i.e.,
\begin{align}
V^{\min}_i  \leq V_i \leq V^{\max}_i, \qquad \forall i \in \mathcal{N}, \label{eqn:Voltage_Limits}
\end{align}
where $V^{\min}_i$ and $V^{\max}_i$ are the safety thresholds for node $i$. By convention, we assume that the current injected into the TPS is positive. 
The resistance of the branch connecting the nodes $i$ and $j$ is denoted by $R_{i,j}(\sv)$ and its conductance by $G_{i,j}(\sv)$, where $G_{i,j}(\sv) = 1 / R_{i,j}(\sv)$. Note that branch resistance (and conductance) depends on the positions of the trains, i.e., $\sv$.
For instance, in Fig.~\ref{fig:DCTractCkt}, 
$R_{i,i+1}=\gamma(s_{i+1} - s_{i})$, 
where $\gamma$ is a constant depending on the electrical wire characteristics. 
From Kirchhoff's circuit laws, we have
\begin{align}
\Ym(\sv) \vv & = \iv \label{eqn:nodal_eq},
\end{align}
where $\Ym(\sv) \in \mathbb{R}^{N \times N}$ is the nodal conductance matrix and the $(i,j)$th element of $\Ym(\sv)$, denoted by $Y_{i,j}(\sv)$, is given by
\begin{align}
Y_{i,i}(\sv)  = \sum_{j : (i,j) \in \mathcal{L}} G_{i,j}(\sv), \qquad
Y_{i,j}(\sv)  = \begin{cases}
                        -G_{i,j}(\sv),  \ & \text{if} \ j \neq i \ \text{and} \ (i,j) \in \mathcal{L},  \\
                        0, \ & \text{if} \ j \neq i \ \text{and} \ (i,j) \notin \mathcal{L}.
                    \end{cases}   \nonumber
\end{align}

\vspace{0.5em}\noindent
{\bf Substations:} We consider inverting substations
capable of both supplying and absorbing power. They are modeled as DC voltage sources governed by
\begin{align}
\label{eqn:sub_vg}
        V_i &   = V_{\text{NL}} - R_s I_i, \qquad i \in \mathcal{N}_{\text{sub}}, 
\end{align}
where $V_{\text{NL}}$ and $R_s$ are the no-load voltage and the internal resistance of the substation. 
When a substation supplies power, $I_i>0$; when it absorbs power, $I_i<0$ and the absorbed power is injected back to the supporting ac power grid.
This dc substation model has been widely adopted in TPS analysis \cite{ArboleyaBFS2015}, \cite{RaghunathanCOMPRAIL2014}.

\vspace{0.5em}\noindent
{\bf Trains:} 
Let $P_i$ denote the power absorbed or injected by a tractioning train or a regenerating train at node $i$. 
We have 
\begin{align}
P_i = V_i I_i . \label{eqn:train_power_org}
\end{align}
For safety, the trains adopt the following two local power controls \cite{OkadaKoseki2004}.

\noindent \emph{Overcurrent control:} A tractioning train absorbs power from the power network, resulting in a drop in the train's nodal voltage. To prevent the nodal voltage from falling below the safety threshold $V^{\min}_i$, the overcurrent control is activated whenever the nodal voltage at the train node $i$ drops below a triggering voltage $V^{\min,\text{Tr}}_{i}$.
Let $P^{d}_{i}$ denote the power demand of a tractioning train at node $i$. 
The overcurrent control will command the train to absorb $P_i$ units of power, where $P_i$ is given by
\begin{align}
\label{eqn:acc_train}
        P_i & = \begin{cases}
        				  0, \ & \text{if} \ V_i \leq V^{\min}_{i} ;\\
        				  P^{d}_{i}  \LB \frac{V_i - V^{\min}_{i} }{V^{\min,\text{Tr}}_{i} - V^{\min}_{i} } \RB, \ & \text{if} \  V^{\min}_{i}   \leq V_i \leq  V^{\min,\text{Tr}}_{i} ;\\
                        P^{d}_{i},  \ & \text{if} \ V_i \geq V^{\min,\text{Tr}}_{i}.
                       \end{cases}   
\end{align}
This control law is illustrated in Fig.~\ref{fig:Acceleration_Control}~(a).
Specifically, if the nodal voltage at the train is greater than the triggering voltage, the train draws a power equal to its demand.
Otherwise, the train curtails its power consumption according to the $V_i$. If the voltage drops below $V^{\min}_{i}$, the train does not draw power to prevent safety incidents.

\noindent
\emph{Squeeze control:} A regenerating train injects power into the power network, resulting in an increase in the train's nodal voltage. To prevent the voltage from exceeding the safety limit $V^{\max}_i$, the squeeze control is activated whenever the nodal voltage at the train node $i$ exceeds a certain triggering voltage level $V^{\text{Tr}}_{i,\max}$.
Let $P^{c}_{i}$ denote the regeneration capacity of the train.
The squeeze control will command the train to inject $P_i$ units of power into the TPS, where $P_i$ is given by
\begin{align}
\label{eqn:reg_train}
 P_i  & = \begin{cases}
                        P^{c}_{i},  \ & \text{if} \ V_i \leq V^{\max,\text{Tr}}_{i}; \\
                        P^{c}_{i}   \LB \frac{V^{\max}_{i} - V_i}{V^{\max}_{i} - V^{\max,\text{Tr}}_{i}} \RB, \ & \text{if} \   V^{\max,\text{Tr}}_{i} \leq V_i \leq V^{\max}_{i}; \\
                        0, \ & \text{if} \ V_i \geq V^{\max}_{i}.
                    \end{cases} 
\end{align}
This control law is illustrated in Fig.~\ref{fig:Acceleration_Control}~(b).
Specifically, if the nodal voltage at the train is lower than the triggering voltage, the train injects all the regenerated power.
Otherwise, the train curtails the power injection according to the $V_i$ by burning the remaining power in a rheostatic braking system \cite{OkadaKoseki2004}.
If the voltage drops below $V^{\max}_{i}$, the train does not inject power into the TPS to prevent safety incidents. 

The train's power demand $P_i^d$ and regeneration capacity $P_i^c$ depend on the train's running profile and real-time state. They can be provided by the train's motion control system. We note that the electrical models described in this section address the steady-state
voltages and currents. They ignore the power transients of the trains due to their internal feedback control systems that implement the overcurrent/squeeze control decisions.
However, it is safe to ignore these transients because they can settle quickly, before the next overcurrent/squeeze control action \cite{Talukdar1977}.
This steady-state analysis approach has been widely adopted in TPS power flow analysis \cite{CaiIterative1995,PiresICCG2007,ArboleyaBFS2015}.

\section{False Data Injection Attacks Against TPS}
\label{sec:Cyberattacks}
In this section, we study how an attacker can mislead the TPS into an inefficient or unsafe operating state.
We focus on FDI attacks that tamper with the measurements of train-borne voltage and current sensors. Such an attack will cause the TPS to make wrong decisions of power absorption/injection, since a train's overcurrent and squeeze controls depend on the sensor measurements. We further consider attacks of two different objectives: (i) increase the system's total instantaneous power consumption, and (ii) cause breaches of the safety conditions in \eqref{eqn:Voltage_Limits}. We call these two types of attacks {\em efficiency attack} and {\em safety attack}, respectively. In this section, we first describe our threat model. Then, we analyze the attacker's approach of computing effective efficiency and safety attacks. Lastly, we present numerical results to illustrate the two kinds of attacks. 

\subsection{Threat Model}
\label{sec:ThreatModel}
Real-world attackers against critical CPSes are often smart, resourceful, and highly strategic. Their strategies can be guided by detailed knowledge of their targets, which can be obtained in practice by malicious insiders, long-term data exfiltration \cite{dragonfly2014}, or social engineering against employees, contractors, or vendors of the operators in question \cite{karnouskos2011}. In this paper, we follow Kerckhoffs's principle to consider an attacker who has accurate knowledge of the targeted system and read access to the system state. Knowledge of the system includes the electrical models and parameters given in Section~\ref{sec:Tract_Power}, as well as the system's method of attack detection. The system state includes present power demands, regeneration capacities, and voltage, current, and position measurements of all the trains. This information can be leaked through a compromised operation center, as in recent high-profile attacks~\cite{karnouskos2011,dragonfly2014}.
We assume that the attacker has write access to voltage, current, and position measurements of nodes in the set $\mathcal{N}_{\text{a}}$, where $\mathcal{N}_{\text{a}} \subseteq \mathcal{N}$, so that he can corrupt these measurements. 
Recent studies have demonstrated that such unauthorized write access can be obtained for analog sensors, traditional electro-mechanical meters, and modern smart meters \cite{SmartMeterSecurity2009,Theft2011,kune2013ghost}. Analog sensors are vulnerable to precisely controlled electromagnetic interference \cite{kune2013ghost}; measurement devices can be affected by {\em hardware trojans} \cite{HardwareTrojans2010} and infected with malwares \cite{SmartMeterSecurity2009,davis2009}.

Under the said Kerckhoffs's assumption on the attacker's knowledge, we will analyze his strategies of achieving successful efficiency and safety attacks. Conversely, we will also develop countermeasures by a defender to detect these attacks and mitigate their impacts. 
Our threat model is strong, but the conservative analysis is necessary because any underestimation of the attacker's capability may have catastrophic consequences, including extremely costly infrastructure damage and loss of human lives.

We note that, alternatively, the attacker can launch FDI attacks against the {\em decisions} of the local controls (i.e., the $P_i$ values for the trains). To detect such attacks, each train can compare the $P_i$ value in question with that computed based on the train's voltage and current measurements and the {\em a priori} overcurrent and squeeze control laws. In the rest of this paper, we focus on the analysis and detection of  FDI attacks on the voltage and current measurements only. This problem is comparatively much more challenging since information compromised right at the sources will preclude its use for any subsequent sanity checks.

Finally, we note that other potential attacks such as the denial-of-service (DoS) attacks 
that block sensor reading reporting can be easily detected, since in TPSes, sensors periodically report readings. Upon detection, the operator can initiate mitigation steps (e.g., stop the trains) to prevent any safety incidents. Thus, in this paper, we focus on the more challenging FDI attacks, as its detection generally needs a deep understanding on the power flows and train/substation operations.

\subsection{FDI Attack Construction} \label{subsec:att_constr}
In this section, we analyze how to compute an effective \emph{attack vector}, as a vector of
false voltage and current measurements to be injected into the sensing systems of the trains in $\mathcal{N}_a$. Note that, in this section we ignore position measurements in the attack vector, because they will not affect the trains' overcurrent and squeeze controls. In the rest of the paper, we will use $x'$ to denote the compromised version of a sensor measurement $x.$
In the following analysis, we first derive conditions for the attack vector to mislead the train into absorbing or injecting a certain amount of power. With the calculated power absorptions/injections of the trains, we can determine the system's total power consumption and hence its safety status. Thus, we can formulate the attacker's problem of finding an attack vector to achieve his goal of maximizing the total power consumption, under conditions that we will state presently for enforcing certain amounts of power absorption/injection.

The following conditions are sufficient to enforce that a train at node $i \in \mathcal{N}_a$ will absorb or inject $P_i$ units of power:
\begin{align}
& V^{\prime}_i  \ \begin{cases}
        				  \geq V^{\min,\text{Tr}}_{i},  & \text{if} \ P_i  = P^{d}_i,  \\
        				  = V^{\min}_{i}  + \frac{P_i (V^{\min,\text{Tr}}_{i} - V^{\min}_{i})}{ P^{\min}_{i} },  \ & \text{if} \ P^{d}_i \leq P_i  \leq 0, \\
                       \leq V^{\min}_{i} ,    \ & \text{if} \  P_i  = 0, 
                       \end{cases} \label{eqn:accvg_attack} \   \forall i \in \mathcal{N}_{\text{a}} \cap \mathcal{N}_{\text{tra}};
\end{align}
\begin{align}             
   & V^{\prime}_i \   \begin{cases}
     \geq V^{\max}_{i}, \ & \text{if} \ P_i  = 0, \\
      = V^{\max}_{i} - \frac{P_i (V^{\max}_{i} - V^{\max,\text{Tr}}_{i} )}{ P^{c}_{i} },  \ & \text{if} \ 0 \leq P_i  \leq P^{c}_i, \\
       \leq V^{\max,\text{Tr}}_{i}, \ & \text{if} \   P_i  = P^{c}_i ,
       \end{cases} \label{eqn:regvg_attack} \ \forall i \in \mathcal{N}_{\text{a}} \cap \mathcal{N}_{\text{reg}}   ;
       \end{align} 
 \begin{align}
 & V^{\prime}_i I^{\prime}_i = P_i, \qquad i \in \mathcal{N}_{\text{a}}; \label{eqn:pow_attack} \\
 & P_i  \geq P_i^d, \qquad i \in \mathcal{N}_{\text{tra}}; \label{eqn:train_acc_power} \\
& P_i \leq P_i^c, \qquad i \in \mathcal{N}_{\text{reg}} \label{eqn:train_reg_power}.
 \end{align}
The conditions in \eqref{eqn:accvg_attack} and \eqref{eqn:regvg_attack} are obtained by inverting the overcurrent and squeeze control laws given in Section~\ref{sec:Tract_Power}, and replacing the true voltage $V_i$ by the compromised measurement $V_i'$. As a result, based on $V_i'$, the train will follow the overcurrent/squeeze control law to regulate its power absorption/injection to the attacker's desired value $P_i$. This control process is often achieved in a closed loop, with the measurements $V_i'$ and $I_i'$ acting as feedback and the desired value $P_i$ as setpoint. Under the condition \eqref{eqn:pow_attack}, the actual power absorption/injection under the aforementioned closed-loop control will converge to $P_i$. Moreover, the condition \eqref{eqn:pow_attack} can hide the attack for trains that can directly measure the power consumption. 
The conditions \eqref{eqn:train_acc_power} and \eqref{eqn:train_reg_power} ensure the feasibility of inducing the train to absorb/inject $P_i$ units of power. Specifically, the attacker's desired $P_i$ should not exceed a regenerating train's regeneration capacity.
The condition \eqref{eqn:train_acc_power}, where both $P_i$ and $P^d_i$ are negative, prevents the mechanism from violating the overcurrent control. In summary, if the compromised measurements $V_i'$ and $I_i'$ satisfy the conditions in \eqref{eqn:accvg_attack} to \eqref{eqn:train_reg_power}, the train will control its power absorption/injection to $P_i$. With this understanding, the attacker can carefully plan the attack vector to achieve his goal. Without the conditions in \eqref{eqn:accvg_attack} to \eqref{eqn:train_reg_power}, the attacker cannot predict the impact of his attack and therefore cannot implement his desired strategy.

Each sensor in the TPS may apply data quality checks on its measurements. For instance, the measurements at the present time instant should not differ significantly from those predicted based on the measurements at the previous time instant.
Intuitively, if the compromised measurement is bounded around the true measurement, the data quality checks, designed to be insensitive to natural random noises of measurement, will not raise an alarm. Thus, we assume that the compromised measurements need to satisfy:
\begin{align}
& \vv -\Delta \vv \preceq \vv^{\prime}  \preceq \vv +\Delta \vv ,   \label{eqn:feas_org1d_nse} \\
                &    \iv-\Delta \iv  \preceq \iv^{\prime} \preceq \iv+\Delta \iv ,  \label{eqn:feas_org1e_nse} 
\end{align}
where $\Delta \vv = [\Delta V_1,\dots,\Delta V_N]^T $ and $\Delta \iv = [\Delta I_1,\dots,\Delta I_N]^T$ are 
the maximum errors allowed by the data quality checks (in Section~\ref{sec:SimRes}, we illustrate how to set the values of $\Delta \vv$ and $\Delta \iv$ based on practical considerations); ${\xv} \preceq {\yv}$ means that each element of ${\xv}$ is no greater than the corresponding element in ${\yv}$. We note that, if $i \notin \mathcal{N}_{\text{a}},$ $\Delta V_i = 0$. In practice, the attacker can obtain the settings of $\Delta \vv$ and $\Delta \iv$ by launching a 
data exfiltration attack \cite{dragonfly2014}. In the absence of such knowledge, the attacker must choose stringent values for these quantities such that the attack is not detected by the data-quality checks.

Based on the above conditions for the compromised measurements, we now formulate the efficiency and
safety attacks.

\subsubsection{Efficiency Attack}
\label{sec:eff_attacks_nse}
An efficiency attack causes an increase or decrease in the total instantaneous power injected or absorbed by the substations.
In particular, we consider an aggressive attacker who aims to maximize or minimize such
injected or absorbed power.
Formally, the attacker solves the following constrained optimization
problem to compute the attack vector $\{V_i', I_i' | \forall i \in \mathcal{N}_a\}$:
\beqa
&\dsp \max_{ \{V_i',I_i' | \ \forall i \in \mathcal{N}_a \}} &  \sum_{i \in \mathcal{N}_{\text{sub}}} V_i I_i  \label{eqn:attacker_economic_nse}  \\ 
& s.t. &  \text{constraints in \eqref{eqn:nodal_eq} to \eqref{eqn:feas_org1e_nse}}. \nonumber
\eeqa
The above formulation captures the physical laws governing the power network and the substations (i.e., \eqref{eqn:nodal_eq} to \eqref{eqn:train_power_org}), as well as how the attack vector induces the trains to make erroneous power control decisions (i.e., \eqref{eqn:accvg_attack} to \eqref{eqn:train_reg_power}). Specifically, for any $\{V_i', I_i' \ | \forall i \in \mathcal{N}_a \}$ satisfying \eqref{eqn:accvg_attack} to \eqref{eqn:train_reg_power}, the attacker can predict the trains' power absorptions/injections $\{P_i = V_i'I_i' \ | \forall i \in \mathcal{N}_{\text{trains}} \}$. He then uses the physical laws in \eqref{eqn:nodal_eq}, \eqref{eqn:sub_vg}, and \eqref{eqn:train_power_org} to determine the actual voltages and currents of the substations (i.e., $\{V_i, I_i \ | \forall i \in \mathcal{N}_\text{sub} \}$) and predict the system's total power consumption $\sum_{i \in \mathcal{N}_{\text{sub}}} V_i I_i$.

Solving the constrained optimization problem in \eqref{eqn:attacker_economic_nse} can be computationally expensive, mainly because the constraints in \eqref{eqn:accvg_attack} and \eqref{eqn:regvg_attack} are non-smooth and non-differentiable. 
Existing constrained non-linear optimization solvers (e.g., the \texttt{fmincon} function of MATLAB) often require the objective and constraint functions to be smooth. To use these existing solvers, the attacker can adopt a {\em divide-and-conquer} approach that splits the problem \eqref{eqn:attacker_economic_nse} into multiple subproblems in which a piece of \eqref{eqn:accvg_attack} or \eqref{eqn:regvg_attack} is selected as a constraint for a train. By comparing the optimization results of all the subproblems, the attacker can obtain a global optimal solution to the problem in \eqref{eqn:attacker_economic_nse}.
 Because each train has three choices in \eqref{eqn:accvg_attack} or \eqref{eqn:regvg_attack}, this approach will generate a total of $3^{|\mathcal{N}_a|}$ subproblems, where $|\mathcal{N}_a|$ is the number of trains under FDI attacks. As the subproblems are mutually independent, the attacker can solve the subproblems in parallel, to reduce computation time.
We note that the ability to solve the problem in \eqref{eqn:attacker_economic_nse} in real time can be important to the attacker. This is because, to accumulate large energy loss, the attacker needs to keep at the FDI attacks by solving \eqref{eqn:attacker_economic_nse} continually, based on the latest system state given by $\sv$, $P_i^d$, and $P_i^c$. The attacker will need to procure sufficient computing resources for achieving the real-time objective.
A resource-constrained attacker can inject a suboptimal attack that does not require extensive 
computations like solving \eqref{eqn:attacker_economic_nse}. In Section~\ref{sec:Illus_nse}, we present a numerical example to show that such an attack can still cause a considerable performance degradation.

\subsubsection{Safety Attack}
\label{sec:safe_attacks_nse}
For safety attacks, we model the space of attack vectors that can cause the voltages 
at a subset of the TPS nodes, denoted by $\mathcal{N}_{\text{unsafe}}$, to cross the safety limits in \eqref{eqn:Voltage_Limits}. The attack space is defined by all the constraint conditions in
the optimization problem \eqref{eqn:attacker_economic_nse}, and
$V_i \notin [V_{i,\min} ,V_{i,\max}], i \in \mathcal{N}_{\text{unsafe}}.$
As long as the attacker can find an attack vector satisfying the above constraints, he will be able to achieve the safety breaches. 

We now discuss a heuristic approach that the attacker can use to aggressively increase the extent of the safety breaches. Specifically, the attacker maximizes the total power injected into the TPS by the regenerating trains, i.e., $\sum_{i \in \mathcal{N}_{\text{reg}}} V_i I_i$, subject to all the constraints of the optimization problem in \eqref{eqn:attacker_economic_nse}. The intuition is that injecting more power into the TPS will result in higher catenary voltages. This constrained optimization problem can also be solved by the aforementioned divide-and-conquer approach. 

The TPS under FDI attacks can be analyzed using the same set of equations as in Section~\ref{sec:Tract_Power} (i.e., \eqref{eqn:Voltage_Limits}-\eqref{eqn:reg_train}), except that the train's overcurrent and squeeze control decisions are now computed based on the attacker's injections $V^\prime_i$ (in \eqref{eqn:acc_train} and \eqref{eqn:reg_train}).
Based on this analysis, in Section~\ref{sec:Illus_nse} we present numerical example  to show the impact of efficiency and safety attacks. We also present time-domain simulation results in Section~\ref{sec:SimRes} considering realistic running profiles of trains.

\subsection{Numerical Examples}
\label{sec:Illus_nse}
We now present numerical examples to illustrate the efficiency and safety attacks.
These examples are based on the TPS shown in Fig.~\ref{fig:DCTractCkt}, in which both trains are decelerating and regenerating. 
The system model parameters are given in Table~\ref{tbl:BFS_parameters}. We consider a time instant at which the system state in the absence of attack is given by the first part of Table~\ref{tbl:Attack_nse}, where the total instantaneous power absorbed by the substations and injected back into the supporting ac power grid is $3.601\, \text{MW}$.
In these examples, we assume that the attacker can 
only compromise the voltage and current measurements of the train at node $2.$ 

\begin{table} [!t]
\centering
\setlength\extrarowheight{2pt}
\begin{tabular}{|c|c|c|c|c|c|}
\hline
	\textbf{Parameters} &  $\text{V}_{\text{NL}}$ & $\gamma$ & $\text{R}_\text{s}$ & $\text{V}_i^{\max,\text{Tr}}$ & $\text{V}_i^{\max}$  \\ \hline
	\textbf{Value} & 750V & 30m$\Omega$/km & 29.56m$\Omega$ &  850V & 900V \\ \hline
\end{tabular}
\caption{TPS model parameters.}
\label{tbl:BFS_parameters}
\end{table}

\subsubsection{Efficiency Attack} 
The attacker solves the constrained optimization problem in \eqref{eqn:attacker_economic_nse} and 
tampers with $V_2$ and $I_2$ accordingly. 
We set $\Delta V_i = 50$V and $\Delta I_i = 200$A, $\forall i \in \mathcal{N}_a$.
The  compromised measurements and the true state of the system under attack are given in the second part of Table~\ref{tbl:Attack_nse}. We can see that the compromised voltage measurement at node $2$ 
is greater than the true value. Consequently, the train
injects less power into the TPS because of the squeeze control, resulting in less power absorption by the substations.
Specifically, the total power absorption is $2.888\,\text{MW}$,
a $20\%$ reduction compared with the case of no attack. 
Thus, the power efficiency of the system is degraded.

We also consider a suboptimal attack in which the attacker compromises the voltage of the train at node $2$ by $20$~V (hence $V'_i = V_i + 20V$).
Under this attack, the total power absorption is $3.25\,\text{MW}$, a $9.5\%,$ reduction compared with the case of no attack. This shows that the attacker can still cause a considerable degradation in the system efficiency by injecting a suboptimal attack. In practice, the attacker can tune his attack strategy to balance between the attack impact and the computational complexity of computing the attacks. 

\subsubsection{Safety Attack} 
The attacker uses the heuristic approach in Section~\ref{sec:safe_attacks_nse} to compute the safety attack.
The compromised measurements and the true system state 
are given in the third part of Table~\ref{tbl:Attack_nse}. The compromised voltage measurement at node $2$ is lower than its true value. Thus, the train
at node $2$ injects more power into the TPS because of the squeeze control, causing the actual voltage at node $2$ to exceed the safety limit.
We can see that it is possible for an attacker to tamper with the measurements of a single train and already achieve a safety attack. In this example, since both the trains are regenerating, the catenary voltages are closer to the safety limit. This makes it easier for the attacker to achieve the safety attack. Thus, for an attacker with limited write access to the trains' measurements (i.e., a small set $\mathcal{N}_a$), he can continuously monitor the system and wait for feasible moments for launching safety attacks. 

\begin{table} [!t]
\centering
\setlength\extrarowheight{2pt}
\begin{tabular}{|M{0.2\textwidth}|M{0.1\textwidth}|M{0.1\textwidth}|M{0.1\textwidth}|M{0.1\textwidth}|M{0.1\textwidth}||M{0.1\textwidth}|} \hline
	\multicolumn{2}{|c|}{Node} & 1 & 2 & 3 & 4 & Efficiency Loss \\ \hline
	\multirow{5}{*}{\parbox{0.2\textwidth}{\center TPS State \\(Without Attack)}} & $s_i$ & 0 & 0.9  & 1.2 & 2 & \\ \cline{2-6}
	& $P_i^c$ & - & 5.5 & 1.8 & - & \\ \cline{2-6}
	& $V_i$ & 815.6 & 875.5 & 867.7 & 815 & \textemdash \\ \cline{2-6}
	& $I_i$ & -2218.8 & 3079.2 & 1338 & -2198.4 & \\ \cline{2-6}
	& $P_i$ & -1.81 & 2.696 & 1.161  & -1.792 & \\ \hhline{|=|=|=|=|=|=|=|}
	\multirow{5}{*}{\parbox{0.2\textwidth}{\center Efficiency Attack (Optimal)}} & $V'_i$ & - & 888.6 & - & - & \\ \cline{2-6}
	& $I'_i$ & - & 1409.6 & - & - & \\ \cline{2-6}
	& $V_i$ & 801.1 & 847.7 & 850 & 805.2 & 20 \% \\ \cline{2-6} 
	& $I_i$ &-1728.2 & 1477.6 & 2117.6 & -1867.1 & \\ \cline{2-6}
	& $P_i$ & -1.384  & 1.253  & 1.8  & -1.503 &  \\ \hhline{|=|=|=|=|=|=|=|}
	\multirow{5}{*}{\parbox{0.2\textwidth}{\center Efficiency Attack (Suboptimal)}} & $V'_i$ & - & 881.9 & - & - &  \\ \cline{2-6}
	& $I'_i$ & - & 1409.6 & - & - & \\ \cline{2-6}
	& $V_i$ & 808.5 & 861.9 & 859.1 & 810.1 & 9.5 \% \\ \cline{2-6} 
	& $I_i$ &-1979.2 & 2301.1 & 1715 & -2036.5 & \\ \cline{2-6}
	& $P_i$ & -1.606 & 1.98  & 1.47  & -1.65 &  \\ \hhline{|=|=|=|=|=|=|=|}
	\multirow{5}{*}{\parbox{0.2\textwidth}{\center Safety Attack}}& $V'_i$ & - & 862.9 & - & - & \\ \cline{2-6}
	& $I'_i$ & - & 4731.6 & - & - & \\ \cline{2-6}
	& $V_i$  & 828.9 & {\bf 901} & 884.2 & 824.1 & \textemdash \\ \cline{2-6}
	& $I_i$ & -2669.1 & 4531.6 & 643.2 & -2505.7 & \\  \cline{2-6}
	& $P_i$ & -2.212  & 4.083  & 0.569  & -2.065 & \\ \hline
\end{tabular}
\caption{System state and compromised measurements under efficiency and safety attacks. Distance is measured in kilometers, voltage in volts, current in amperes, and power in megawatts.}
\label{tbl:Attack_nse}
\end{table}

\section{Global Attack Detection}
\label{sec:BDD}

As discussed in Section~\ref{sec:intro}, dc TPSes mainly rely on trains' local controls (i.e., overcurrent and squeeze controls) to avoid unsafe states.
The TPS does not otherwise cross-check sensor data from different trains to ensure the data's global consistency.
However, such global monitoring is clearly advantageous,
because anomalies in the data relationships can help flag the occurrence of an FDI attack. 
Furthermore, not only can we cross-check sensor measurements from different trains, we can also check these measurements 
against an  
{\em a priori} global TPS model to ensure agreement. An attacker that wishes to remain stealthy under global monitoring thus becomes more constrained, and his actions may become less effective.  
In this section, we present the design of a global monitor for detecting FDI attacks under the Kerckhoffs's assumption, which we will subsequently refer to as the global attack detector (GAD).

\begin{figure}
\centering
\includegraphics[width=0.6\textwidth]{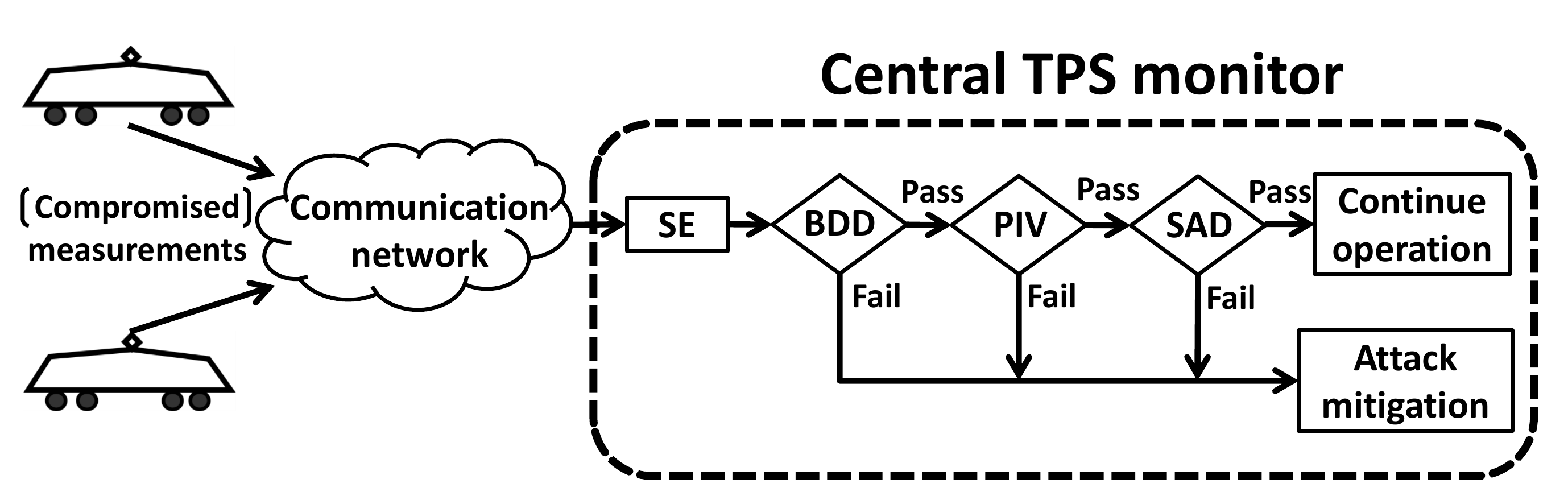}
\caption{Global attack detection. SE: State estimation; BDD: Bad data detection; PIV: Position integrity verification; SAD: Secondary attack detection.}
\label{fig:TPS_SE_BDD}
\end{figure}

Fig.~\ref{fig:TPS_SE_BDD} overviews our global attack detection approach. In the approach, the trains' voltage, current, and position measurements are sent to a central {\em TPS monitor} periodically. As Fig.~\ref{fig:TPS_SE_BDD} illustrates, the TPS monitor applies state estimation (SE), bad data detection (BDD), position integrity verification (PIV), and secondary attack detection (SAD) in sequence to detect attacks. In ac utility power grids, similar SE and related BDD are widely used for detecting faulty data or reducing the impact of noisy sensor measurements~\cite{wood1996power}. 
In Section~\ref{sec:SE}, we propose a new BDD design that is specific to the application domain of dc TPS.
By checking the consistency among measurements based on prior knowledge of the TPS, the BDD can detect a range of FDI attacks. However, the detection is not complete -- an attacker under the Kerckhoffs's assumption will be able to bypass it using his knowledge of the system. In Section~\ref{sec:SimRes}, we provide numerical results to illustrate the impacts of these stealthy efficiency and safety attacks.
To counter the stealthy attacks, in Section~\ref{subsec:SAD} we further propose a novel SAD algorithm to supplement the BDD, under an additional assumption that the trains' position data is intact, which is ensured by the PIV.

\subsection{TPS Bad Data Detection and Its Vulnerability} \label{sec:SE}
Recall that in Section~\ref{subsec:att_constr}, the trains apply local controls based on their own voltage and current measurements only. Hence, the trains' position information does not matter.
Under global detection, however, compromise of the trains' position information becomes relevant, since it may enable the attacker to mislead the TPS monitor into deriving an incorrect TPS model that is consistent with the compromised voltage and current measurements. Tampering with the position data can thus help the attacker evade detection. Although in practice it is extremely difficult for the attacker to hide the compromise of train position data because multiple sources of this data are often available (see Section~\ref{subsec:SAD} for the details), in this section, for generality, we account for possible compromise of the position data.

We use $\tilde{x}$ to represent a possibly compromised measurement $x$, i.e., $\tilde{x}=x$ in the absence of attack and $\tilde{x}=x'$ in the presence of attack.
The state of the TPS is the vector of the nodal voltages, i.e., $\vv$.
The set of measurements includes nodal positions $\tilde{\sv} = [\tilde{s}_1,\dots,\tilde{s}_N]^T \in \mathbb{R}^{N \times 1}$,  
and nodal voltage and current readings $\tilde{\zv} = [\tilde{\vv},\tilde{\iv}]^T \in \mathbb{R}^{2N \times 1}$.
In the absence of attack, the measurement vector $\zv$ 
is related to the system state $\vv$ as
$\zv = \Hm (\sv) \vv + \nv,$
where $\Hm (\sv)     
 = [\mathbb{I}_{N};\Ym(\sv)]$ is a {\em measurement matrix} depending on the positions $\sv$, $\mathbb{I}_{N}$ is an $N$-dimensional identity matrix, and
$\nv \in \mathbb{R}^{2N \times 1}$ is a random measurement noise vector. We assume that $\nv$ follows a multivariate Gaussian distribution.
The maximum likelihood (ML) estimate of $\vv$, denoted by $\hat{\vv}$, is given by \cite[Chap.~12]{wood1996power}
\begin{equation*}
\hat{\vv} = (\Hm (\tilde{\sv})^T \Sigmam^{-1} \Hm (\tilde{\sv}))^{-1} \Hm (\tilde{\sv})^T \Sigmam^{-1} \tilde{\zv},
\end{equation*}
where $\Sigmam$ is the covariance matrix of $\nv$. 
The SE's BDD raises an alarm if 
\begin{equation*}
(\tilde{\zv} - \Hm (\tilde{\sv}) \hat{\vv})^T \Sigmam^{-1} (\tilde{\zv} - \Hm (\tilde{\sv})\hat{\vv}) > \tau,
\end{equation*}
where $\tau$ is a constant threshold that can be determined to meet a given false alarm rate under random measurement noise\footnote{A detailed description of how to set the BDD threshold is given in Appendix~A.}. The BDD is originally designed to detect faulty sensor data caused by natural malfunction of sensors. Thus, it is effective in detecting a range of FDI attacks that are not specifically designed to bypass it. 
However, the attacker that we consider in this paper, following the Kerckhoffs's principle, will be able to design FDI attacks with the objective of bypassing the BDD.
In the following,
we formulate these stealthy safety and efficiency attacks. 

From an existing result~\cite{LiuNingReiter2009}, if the compromised measurement vector ${\zv}^{\prime}$ is in the column space of the compromised measurement matrix $\Hm(\sv^{\prime})$, 
${\zv}^{\prime}$ can bypass the BDD. Applying this result to the TPS, we have the following lemma.
\begin{lemma}
\label{cor:Bypass}
Any compromised measurements that satisfy 
\begin{align}
& \Ym({\sv}^{\prime})  {\vv}^{\prime} = {\iv}^{\prime}  \label{eqn:BDD_constraint1}
\end{align}
can bypass the BDD.
\end{lemma}
\begin{proof}
Lemma \ref{cor:Bypass} holds since
any compromised measurement vector ${\zv}^{\prime}$ that satisfies
\eqref{eqn:BDD_constraint1} lies in the column space of $\Hm (s')$, i.e.,
${\zv}^{\prime}  = [{\vv}^{\prime},{\iv}^{\prime}]^T  =  [\mathbb{I}_{N};\Ym(\sv^{\prime})] \vv^{\prime} =  \Hm(\sv^{\prime}) \vv^{\prime}$.
\end{proof}

In addition to \eqref{eqn:BDD_constraint1}, the TPS monitor may use two other sensor data checks. First, to meet the constraint in \eqref{eqn:BDD_constraint1}, the attacker may need to compromise the voltage and current measurements at the substations. The TPS monitor may check the substation measurements, i.e., $V_i$ and $I_i$, $\forall i \in \mathcal{N}_{\text{sub}}$, against the substation model in \eqref{eqn:sub_vg}. To be stealthy to this check, the attacker can impose an additional constraint of 
\begin{align}
V_i' = V_{\text{NL}}-R_s I_i', \quad \forall i \in \mathcal{N}_{\text{sub}}. \label{eqn:BDD_constraint2}
\end{align}
Second, the TPS monitor can also apply data quality checks similar to those in \eqref{eqn:feas_org1d_nse} and \eqref{eqn:feas_org1e_nse} to check the trains' position measurements. Thus, if the attacker can compromise the position measurements, he needs to satisfy
\begin{align}
\sv-\Delta \sv  \preceq \sv^{\prime}  \preceq \sv+\Delta \sv, \label{eqn:pos_attack}
\end{align}
where $\Delta \sv = [\Delta s_1,\dots,\Delta s_N]^T$ are the maximum allowed errors for position measurements and $\Delta s_i = 0$ if $i \notin \mathcal{N}_{\text{a}}$.

Therefore, the efficiency attacks that are stealthy to the BDD can be computed by solving 
the constrained optimization problem \eqref{eqn:attacker_economic_nse} with
the additional constraints \eqref{eqn:BDD_constraint1}, \eqref{eqn:BDD_constraint2}, and \eqref{eqn:pos_attack}. 
Similarly, the attack space for BDD-stealthy safety attacks is characterized by the constraints of the optimization problem \eqref{eqn:attacker_economic_nse}, $V_i \notin [V_{i,\min} ,V_{i,\max}], i \in \mathcal{N}_{\text{unsafe}}$, and the additional constraints \eqref{eqn:BDD_constraint1}, \eqref{eqn:BDD_constraint2}, and \eqref{eqn:pos_attack}. Naturally, BDD reduces the attack space since the attacker now needs
to satisfy additional constraints to remain undetected. In the simulation results presented in Section~\ref{sec:SimRes}, we show that, under a realistic TPS setting, the BDD significantly reduces the impact of attacks.

\subsection{Numerical Examples}
\label{sec:illus_se}
We now present numerical examples to illustrate the efficiency and safety attacks that can bypass the BDD as analyzed in Section~\ref{sec:SE}. 
The TPS model and parameters are identical 
to those in Section~\ref{sec:Illus_nse}. The true system state and the compromised measurements are given in Table~\ref{tbl:Attack_se}. 
We set $\Delta s_i = 0.6\, \text{km}$, $\forall i \in \mathcal{N}_a$.
To illustrate a powerful attacker, we assume that the attacker 
can corrupt the voltage and current measurements of all the four nodes in Fig.~\ref{fig:DCTractCkt}, as well as the positions of both the trains.

\subsubsection{Efficiency Attack}
Under the efficiency attack, the total power injected back to the supporting power grid by the substations is $3.431$ MW, which is
a reduction of about $4.7 \%$ compared with no attacks. 
This reduction 
is much less than the $20\%$ caused by the efficiency attack in Section~\ref{sec:Illus_nse}, which was achieved by compromising the voltage and current measurements of node $2$ only in the absence of BDD. This result illustrates the ability of the BDD in limiting the impact of efficiency attacks.

\subsubsection{Safety Attack}
We observe that by compromising the nodal measurements and the trains' position information, the attacker can increase the voltage at node $2$ to $901.4\,\text{V}$ while bypassing the BDD. Furthermore, if the attacker can gain write access to any one train  (i.e., $|\mathcal{N}_a|=1$), he cannot launch a successful safety attack. This is in contrast to the example in Section~\ref{sec:Illus_nse}, where the attacker could launch a successful safety attack by compromising the measurements of a single train only.

In summary, the above examples suggest that the global monitoring and BDD can significantly limit the impact of stealthy FDI attacks on the TPS even if the attacker can compromise the measurements of multiple trains. To accomplish a safety attack, the attacker needs to compromise more trains compared with no BDD.

\begin{table} [!t]
\centering
\setlength\extrarowheight{2pt}
\begin{tabular}{|M{0.2\textwidth}|M{0.1\textwidth}|M{0.1\textwidth}|M{0.1\textwidth}|M{0.1\textwidth}|M{0.1\textwidth}|} \hline
	\multicolumn{2}{|c|}{Node} & 1 & 2 & 3 & 4 \\ \hline
	\multirow{6}{*}{\parbox{0.2\textwidth}{\center Efficiency Attack}} & $s'_i$ & 0 & 1 & 1 & 2 \\ \cline{2-6}
	& $V'_i$ & 812 & 874.9 & 874.9 & 812 \\ \cline{2-6}
	& $I'_i$ &-2096.7 & 3159 & 1034.5 & -2096.8 \\ \cline{2-6}
	& $V_i$ & 813.2 & 871 & 861.7 & 811.6  \\ \cline{2-6}
	& $I_i$ & -2138.8 & 3173.2 & 1050.4 & -2084.8 \\ \cline{2-6}
	& $P_i$ & -1.739 & 2.764 & 0.905 & -1.692  \\ \hhline{|=|=|=|=|=|=|}
	\multirow{6}{*}{\parbox{0.2\textwidth}{\center Safety Attack}} & $s'_i$ & 0 & 0.43 & 1.8 & 2 \\ \cline{2-6}
	& $V'_i$ & 835 & 872.3 & 847.3 & 830.9 \\ \cline{2-6}
	& $I'_i$ & -2876.8 & 3487.8 & 2124.5 & -2735.4 \\ \cline{2-6}
	& $V_i$ & 829.1 & \textbf{901.4} & 895.1 & 830.1 \\ \cline{2-6}
	& $I_i$ & -2676.8 & 3375.3 & 2010.9 & -2709.4 \\ \cline{2-6}
	& $P_i$ & -2.219 & 3.043 & 1.8 & -2.249  \\ \hline
\end{tabular}
\caption{System state and compromised measurements under efficiency and safety attacks that have bypassed the BDD. Distance is measured in kilometers, voltage in volts, current in amperes, power in megawatts.}
\label{tbl:Attack_se}
\end{table}

\subsection{Secondary Attack Detection (SAD)} \label{subsec:SAD}
In this section, we propose a novel secondary attack detection (SAD) algorithm that can effectively detect the {\em onset} of an FDI attack that has bypassed the BDD. A requirement for the SAD is that the trains' position data communicated to the TPS monitor is intact. It is feasible for the TPS monitor to verify the integrity of the position data.
For example, real-world railway systems invariably provide multiple sources of train position information including train-borne wheel sensors and GPS, track-side Balise \cite{ALSTOM}, etc. 
By cross-checking position measurements from the multiple sources, we can readily identify FDI attacks on the position data unless the attacker succeeds in compromising all the data sources, which is highly challenging since these sensors use technologies that
are significantly different from each other. For example, GPS is a satellite-based system, Balise uses electronic beacon or transponder placed between the rails, etc. Such cross checks constitute the PIV illustrated in Fig.~\ref{fig:TPS_SE_BDD}. Given that TPS is a safety-critical system, the operator should enforce the highest consistency requirement on the position measurements from different sensors, i.e., if any inconsistency is found among different position sensors' readings, the PIV should raise a fault/attack alarm. If FDI attacks on the position data are identified, the TPS should immediately apply attack mitigation such as the approach discussed in Section~\ref{sec:mitigation}. 

Note that the analysis in the previous sections is for a particular time instant, and the attacker can use the techniques in Sections~\ref{sec:Cyberattacks} and \ref{sec:SE} to launch attacks continually over time. Once the SAD detects an attack's onset, the system can activate the attack mitigation approach in Section~\ref{sec:mitigation} to render subsequent FDIs ineffective. Thus, in this section we focus on analyzing the property of the system and designing the SAD accordingly for the onset time instant only of an attack.

\subsubsection{A Discrete Solution Property} \label{subsec:example}

The requirement of intact position data and the design of the SAD algorithm are based on a key observation as follows. If the attacker can compromise the trains' position data, the three equality conditions \eqref{eqn:pow_attack}, \eqref{eqn:BDD_constraint1}, and \eqref{eqn:BDD_constraint2} that the attacker must obey form an underdetermined problem with $3N$ variables and $2N$ equations. Since the other conditions that the attacker needs to follow (i.e., \eqref{eqn:accvg_attack}, \eqref{eqn:regvg_attack}, \eqref{eqn:train_acc_power} to \eqref{eqn:feas_org1e_nse}, and \eqref{eqn:pos_attack}) are inequalities, the attacker's problem of finding stealthy FDI attack vectors most likely has infinitely many solutions that are continuous. However, if the trains' position data is intact, the three equality conditions \eqref{eqn:pow_attack}, \eqref{eqn:BDD_constraint1}, and \eqref{eqn:BDD_constraint2} with $\sv^{\prime}$ replaced by the known $\sv$, will form a determined problem with $2N$ variables and $2N$ equations. 
As a result, the attacker's problem has a finite number of
discrete solutions (which we will prove shortly) and the attacker must choose one of them
that is different from the true measurement vector. In what follows, we first show that with intact position data, there are only a finite number
of discrete solutions that satisfy the BDD-passing conditions. We then describe the SAD algorithm.

The BDD bypass conditions given by \eqref{eqn:BDD_constraint1} and \eqref{eqn:BDD_constraint2} can be 
compactly represented as 
\begin{align}
\Vm \LB \Ym(\sv) +\Gm \RB \vv = \cv, \label{eqn:BDD_bypass_Bez}
\end{align}
where the $(i,j)^{\text{th}}$ elements of the matrices $\Vm \in \mathbb{R}^{N \times N}$ and $\Gm \in \mathbb{R}^{N \times N}$ are given by
\begin{align*}
\Vm(i,j) = \begin{cases} V_i, \text{if} \  i=j \ \text{and} \ i \in \mathcal{N}_{\text{trains}}, \\
1, \text{otherwise}, 
\end{cases}, \qquad \Gm(i,j) = \begin{cases} R_s^{-1}, \text{if} \ i=j \ \text{and}\  i \in \mathcal{N}_{\text{sub}}, \\
0, \text{otherwise},
\end{cases},
\end{align*}
and the $i^{\text{th}}$ element of the vector $\cv \in \mathbb{R}^{N \times 1}$ is given by
\begin{align*}
\cv(i) = \begin{cases} P(i), \ \text{if} \ i \in \mathcal{N}_{\text{trains}}, \\
\frac{V_{\text{NL}}}{R_s}, \ \text{if} \ i \in \mathcal{N}_{\text{sub}}.
\end{cases}
\end{align*}
Equation \eqref{eqn:BDD_bypass_Bez} is a system consisting of $N$ polynomial equations with $N$ variables. 
Such a system of equations is referred to as a \emph{square polynomial} system, and the Bezout's theorem provides an upper bound on the number of solutions for such systems. The Bezout's theorem is as follows.
\begin{theorem}
(Bezout's Theorem) \cite{frank2011}
For a square polynomial system, the bound on the number of complex solutions is at most the product of the degrees of the polynomials.
\end{theorem}
The existence of the upper bound proves that the system of polynomials in \eqref{eqn:BDD_bypass_Bez}
has a finite number of discrete solutions. For the BDD bypass condition in \eqref{eqn:BDD_bypass_Bez}, we have a polynomial constraint corresponding to each train in the system ($\mathcal{N}_{\text{tra}} \cup \mathcal{N}_{\text{reg}}$), and each one is a second degree polynomial. Thus, the upper bound according to Bezout's theorem would be $2^{|\mathcal{N}_{\text{tra}} \cup \mathcal{N}_{\text{reg}}|}.$ However, in practice, we found that several solutions to the square polynomial system were complex, which we can discard (since the voltages in a dc system cannot be complex).

We now provide a numerical example to illustrate this property. 
In this example, we use the TPS shown in Fig.~\ref{fig:DCTractCkt} with the settings 
listed in Table~\ref{tbl:BFS_parameters} and $P_2 = P_3 = -0.3\,\text{MW}$.
The two curves in Fig.~\ref{fig:Discrete_Sols} correspond to the two equality conditions that $V_2'$ and $V_3'$ need to satisfy to bypass the BDD. Their intersections are the solutions to the attacker's problem of finding stealthy attack vectors. We can see that the solutions are discrete.

\vspace{1em}

\begin{minipage}[!t]{\textwidth}
\begin{minipage}[!t]{0.33\textwidth}
\centering

\setlength\extrarowheight{4pt}
\begin{tabular}{|c|c|}
\hline

	$\text{V}_{\text{NL}}$ & 750 V \\ \hline
	$\gamma$ & 0.03 $\Omega$\ km \\ \hline
	$\text{R}_\text{s}$ & 0.02956 $\Omega$ \\ \hline
	$\text{V}_i^{\max,\text{Tr}}$ & 850 V \\ \hline
	$\text{V}_i^{\max}$ & 900 V \\ \hline

\end{tabular}
\captionof{table}{TPS parameters.}\label{tbl:BFS_parameters}
\end{minipage}
~
\begin{minipage}[!t]{0.65\textwidth}
	\centering
	\includegraphics[width=0.8\textwidth]{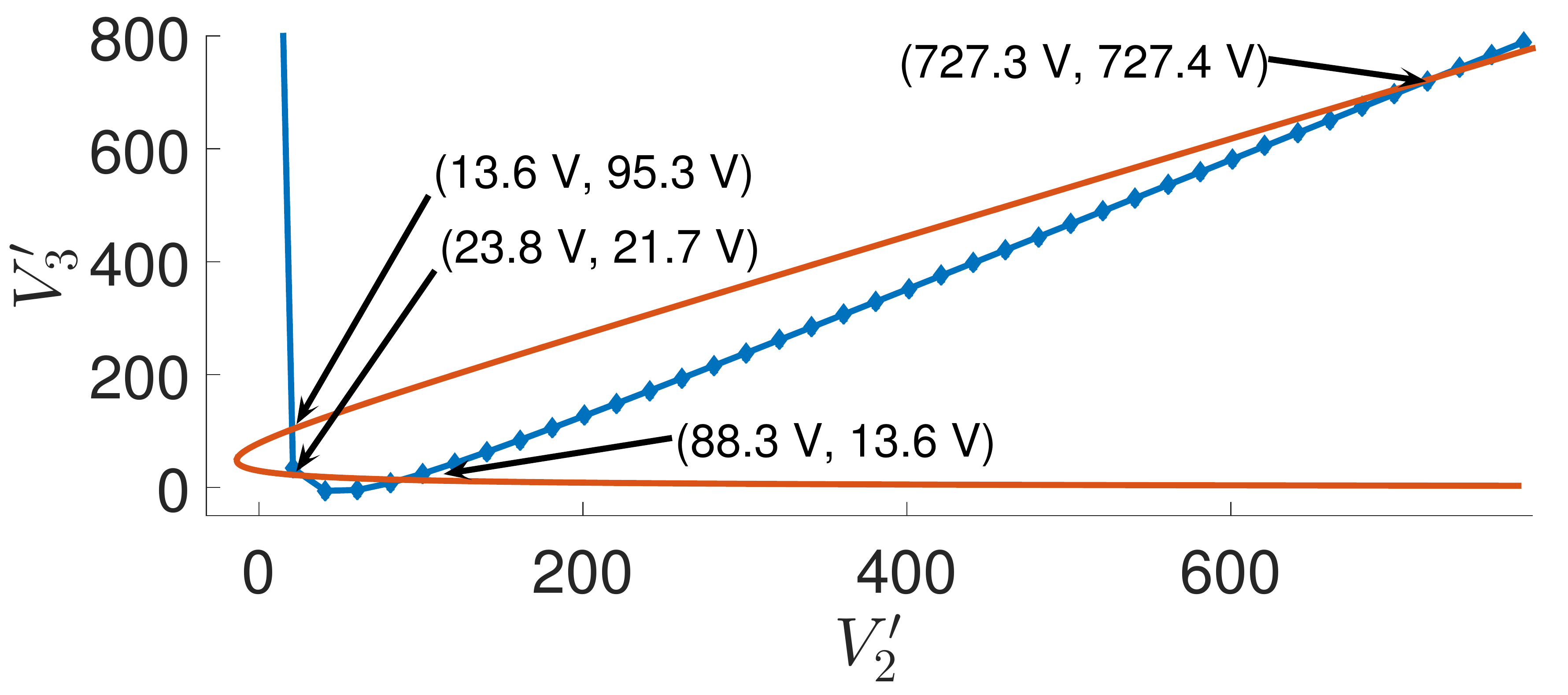}	
	\captionof{figure}{A numerical example illustrating discrete solution property.}
	\label{fig:Discrete_Sols}

\end{minipage}

\end{minipage}

\subsubsection{SAD Algorithm}

Based on the discrete solution property, we design the SAD algorithm as follows.
\begin{algorithm}
\label{alg:AttackMitigation}
\vspace{0.05in}
{\bf Inputs:} Trains' true positions $\sv$, possibly compromised measurement vector $\tilde{\zv}$, intact nodal voltage vector $\vv_{\text{pr}}$ at the previous time instant\\
{\bf Output:} Attack onset detection result
\begin{itemize}
 \item[1.] Using $\tilde{\zv},$
compute $P_i = \tilde{V}_i \tilde{I}_i, \ i \in \mathcal{N}_{\text{trains}}.$

\item[2.] Solve the following constrained optimization problem
\begin{subequations}
\label{eqn:second_check}
\beqa
J^* = &\dsp \min_{ \substack{\vv_{a} , \vv_{b} \\ \vv_{a} \neq \vv_{b}}} & ||\vv_{a} - \vv_{b}||_p \\ 
& s.t. &    \Vm_{a} \LB \Ym(\sv) +\Gm \RB \vv_{a} = \cv, \label{eqn:compact_nodaleq}\\
& & \Vm_{b} \LB \Ym(\sv) +\Gm \RB \vv_{b} = \cv, \label{eqn:compact_nodaleq1}
\eeqa
\end{subequations}
where $||\xv||_p$ represents the $p$-norm of a vector $\xv$ 
\item[3.] Extract $\widetilde{\vv}$ from $\widetilde{\zv}$. If $ || \widetilde{\vv} - {\vv}_{\text{pr}} ||_p \leq  J^*$, report no attack; Otherwise, report onset of attack.
\item[4.] If there only exists two discrete points in the solution set, $J^* = \alpha J^*$
\end{itemize}
\end{algorithm}

In Step 1 of the algorithm, given the possibly compromised measurement vector $\tilde{\zv}$, the TPS monitor computes the actual power absorption or injection of each train. We note that this follows from \eqref{eqn:train_power_org}. Based on the trains' true positions $\sv$ and powers, in Step 2, the TPS monitor solves the constrained optimization problem \eqref{eqn:second_check}. The constraints in \eqref{eqn:compact_nodaleq} and \eqref{eqn:compact_nodaleq1} are compact representations of the BDD bypass condition as explained in \eqref{eqn:BDD_bypass_Bez}, for two distinct solutions $\vv_{a}$ and $\vv_{b}$. By the observation that the BDD bypass condition given the trains' true positions has discrete solutions, $\vv_a$ and $\vv_b$ that solve the optimization problem  \eqref{eqn:second_check} are two distinct solutions that are closest to each other (among all such pairs of solution vectors) and the $J^*$ given by \eqref{eqn:second_check} is the minimum distance.

In Step 3, the TPS monitor compares the $J^*$ with the $p$-norm distance between the possibly compromised voltage measurement vector and the intact nodal voltage vector $\vv_{\text{pr}}$ at the previous time instant, to determine the possible onset of an attack. This step is based on that if the attacker launches a BDD-stealthy attack without tampering with the trains' position information, the $p$-norm distance between the compromised voltage vector and the voltage vector in the absence of attack must be no less than $J^*$. As the voltage vector in the absence of attack is unknown to the TPS monitor, a {\em practical approach} is to use the $\vv_{\text{pr}}$ that is not compromised before the onset of the attack. Since the TPS monitor can run the SAD periodically and frequently (e.g., every second), the TPS state will not change significantly over one monitoring time internal. In Section~\ref{sec:SimRes}, extensive simulations demonstrate the effectiveness of this practical approach by comparing it with an {\em oracle approach} that uses the voltage vector at the present time instant in the absence of attack in Step 3. If and when the onset of an attack is detected, the TPS switches to an attack mitigation mode, as discussed in Section~\ref{sec:mitigation}, to prevent safety breaches.
In step 4, we scale the value of $J^*$ by a parameter $\alpha \in [0,1]$ in case there are only two discrete solutions that satisfy the BDD-passing condition. Step 4 is introduced to 
reduce the MDs of the SAD in the presence of sensor measurement noises (the rationale
behind the introduction of this parameter will be explained in Section~\ref{sec:Noise_Analysis}). 

\subsection{Attack Mitigation}
\label{sec:mitigation}

We outline an approach to mitigating the impact of an attack that has been detected by the TPS monitor by the BDD, PIV, or SAD. On detecting the onset of the attack, the system switches to an {\em attack mitigation mode} in which the TPS monitor issues power absorption/injection commands to the trains to replace their local overcurrent/squeeze controls. Specifically, the TPS monitor computes the $P_i$ for each train based on the trains' power demands, regeneration capacities, and positions by solving the electrical models and trains' local control laws presented in Section~\ref{sec:Tract_Power}. Note that the trains can report their power demands and regeneration capacities to the TPS monitor. The TPS monitor can also estimate them based on trains' running profiles that are often fixed during the planning phase. If FDI attacks on trains' position information have been detected, the system can estimate the trains' positions based on their running profiles. Each train applies the $P_i$ received from the TPS monitor. The core idea of this mitigation approach is to run the TPS temporarily based on models rather than compromised sensor measurements. Emergent running profiles that stop the trains safely should be applied immediately once the system enters the attack mitigation mode.

\section{Impact of Sensor Measurement Noise}
\label{sec:Noise_Analysis}
In this section, we examine the performance of the GAD in the presence of sensor measurement noises. 
In this section, we consider the additive Gaussian noise model described in Section~\ref{sec:SE}. Sensor measurement noise provides the attacker an opportunity 
to hide its attack by masquerading false measurements as legitimate noisy measurements, leading
to MDs and FPs.
MDs may result in the loss of system efficiency or safety breaches. On the other hand, FPs result in the system operator initiating unnecessary mitigation steps that may degrade performance. 
We now formally define FPs and MDs for the BDD, SAD, and GAD respectively.

We consider two hypotheses: $H_0$ denotes that the system is not under attack, and $H_1$ denotes that the system is under attack. We let $Z_{\text{BDD}}$ and $\Xi_{\text{BDD}}$ represent the indicator variables for the occurrences of FP and MD, respectively, in the BDD, and $Z_{\text{SAD}}$ and $\Xi_{\text{SAD}}$ represent the corresponding 
quantities for the SAD.
They can be mathematically stated as
\begin{align*}
Z_{\text{BDD}} & = {1}_{\{ (\tilde{\zv} - \Hm (\tilde{\sv}) \hat{\vv})^T \Sigmam^{-1} (\tilde{\zv} - \Hm (\tilde{\sv})\hat{\vv}) > \tau \ \big{|} H_0  \}}, \qquad Z_{\text{SAD}}  = 1_{\{ || \widetilde{\vv} - {\vv}_{\text{pr}} ||_p >  J^* \ \big{|} H_0 \}} , \\
\Xi_{\text{BDD}} & = 1_{\{ (\tilde{\zv} - \Hm (\tilde{\sv}) \hat{\vv})^T \Sigmam^{-1} (\tilde{\zv} - \Hm (\tilde{\sv})\hat{\vv}) \leq \tau \ \big{|} H_1  \}}, \qquad \Xi_{\text{SAD}} = 1_{\{ || \widetilde{\vv} - {\vv}_{\text{pr}} ||_p \leq  J^* \ \big{|} H_1  \} },
\end{align*}
where ${1}_{\mathcal{A}}$ is an indicator function given by ${1}_{ \mathcal{A} }  = 1 \ \text{if } \ \mathcal{A} \text{ is true},$ or $0 \ \text{otherwise}.$
Similarly, we define $Z_{\text{GAD}}$ and $\Xi_{\text{GAD}}$ for the GAD.
Since the GAD serializes the BDD and SAD, it will raise an alarm if one of the following two events occurs: i) the BDD raises an alarm; or ii) if the measurements pass the BDD but the SAD raises an alarm. Thus $Z_{\text{GAD}}$ can be expressed in terms of 
$Z_{\text{BDD}}$ and $Z_{\text{SAD}}$ as 
\begin{align*}
Z_{\text{GAD}} & = Z_{\text{BDD}} \ \lor ( \neg Z_{\text{BDD}} \land \ Z_{\text{SAD}}).
\end{align*}
Similarly,  $\Xi_{\text{GAD}}$ can be expressed in terms of
$ \Xi_{\text{BDD}}$ and  $\Xi_{\text{BDD}}$ as
\begin{align}
\Xi_{\text{GAD}} & = \Xi_{\text{BDD}} \ \land \ \Xi_{\text{SAD}}.
\end{align}
Next, we use a numeric example to illustrate FPs and MDs in the cases of BDD and SAD, respectively, for a representative TPS network.
\begin{itemize}
\item {\bf BDD FPs and MDs:} The BDD's FPs and MDs are caused by fluctuations 
in the residual value $(\tilde{\zv} - \Hm (\tilde{\sv}) \hat{\vv})^T \Sigmam^{-1} (\tilde{\zv} - \Hm (\tilde{\sv})\hat{\vv}),$ which are in turn caused by the measurement noises. 
This is  illustrated in Fig.~\ref{fig:ResidualNoise} for a TPS model identical to that in Section~\ref{subsec:SAD}, considering $1000$ realizations of measurement noise sampled from an i.i.d. zero-mean Gaussian distribution with a standard deviation set to $0.3\%$ of the full-scale
voltage \cite{sensor_accuracy} and current sensor readings, respectively. (The full-scale voltage and current readings are $900$~V and  $2,500$~A, respectively.) To generate TPS measurements under $H_1,$ we inject an additive attack of $20$ V to the voltage measurement of node $2.$
It can be seen that in the absence of measurement noise, the value of the residual is $0$
under $H_0$, and non-zero under $H_1.$  Therefore, any occurrence of a non-zero residual indicates the presence of an attack in a noiseless environment. However, in the presence of sensor measurement noise, the value of the residual fluctuates under different noise instantiations. Thus, differentiating measurements under attack from those under natural measurement noise becomes challenging.
\item {\bf FPs and MDs of SAD under the Oracle Approach:} 
The SAD's FPs and MDs are due to fluctuations in the value of  $J^*$ and $||\vv-\vv_{\text{pr}}||$ under different realizations of the measurement noise, as illustrated in Fig.~\ref{fig:SADNoise}, for hypothesis $H_0.$ (Recall from Algorithm~1 that in this case, noisy voltage and
current measurements are used as inputs to the SAD algorithm.) Note from Fig.~\ref{fig:SADNoise} that in the absence of measurement noise, both these quantities have a fixed value, in contrast
to the case of noisy measurements. Thus, in the presence of noise, an FP is declared whenever there is no attack on the system and $||\vv-\vv_{\text{pr}}|| \geq J^*,$ and MD is declared whenever the system is under attack and $||\vv-\vv_{\text{pr}}|| < J^*.$
\item {\bf FPs and MDs of SAD under the Practical Approach:} 
Another factor that contributes to the occurrence of FPs and MDs in the practical approach is the following. Since $\vv_{\text{pr}}$ is estimated based on the historical measurements, whenever there is a sudden change in the system state between successive time slots, the difference $||\vv-\vv_{\text{pr}}||$ can become large and result in FPs. In the simulations presented in Section~\ref{sec:SimRes_Noise}, we observe that  when one or more trains in the TPS change status from tractioning mode to breaking mode, there is a large change in the TPS system state. 

\end{itemize}

\begin{figure}[!t]

\begin{minipage}[!t]{\textwidth}
\begin{minipage}[!t]{0.45\textwidth}
\vspace{0pt}
\includegraphics[width=1\textwidth,trim={0 5cm 0 0}]{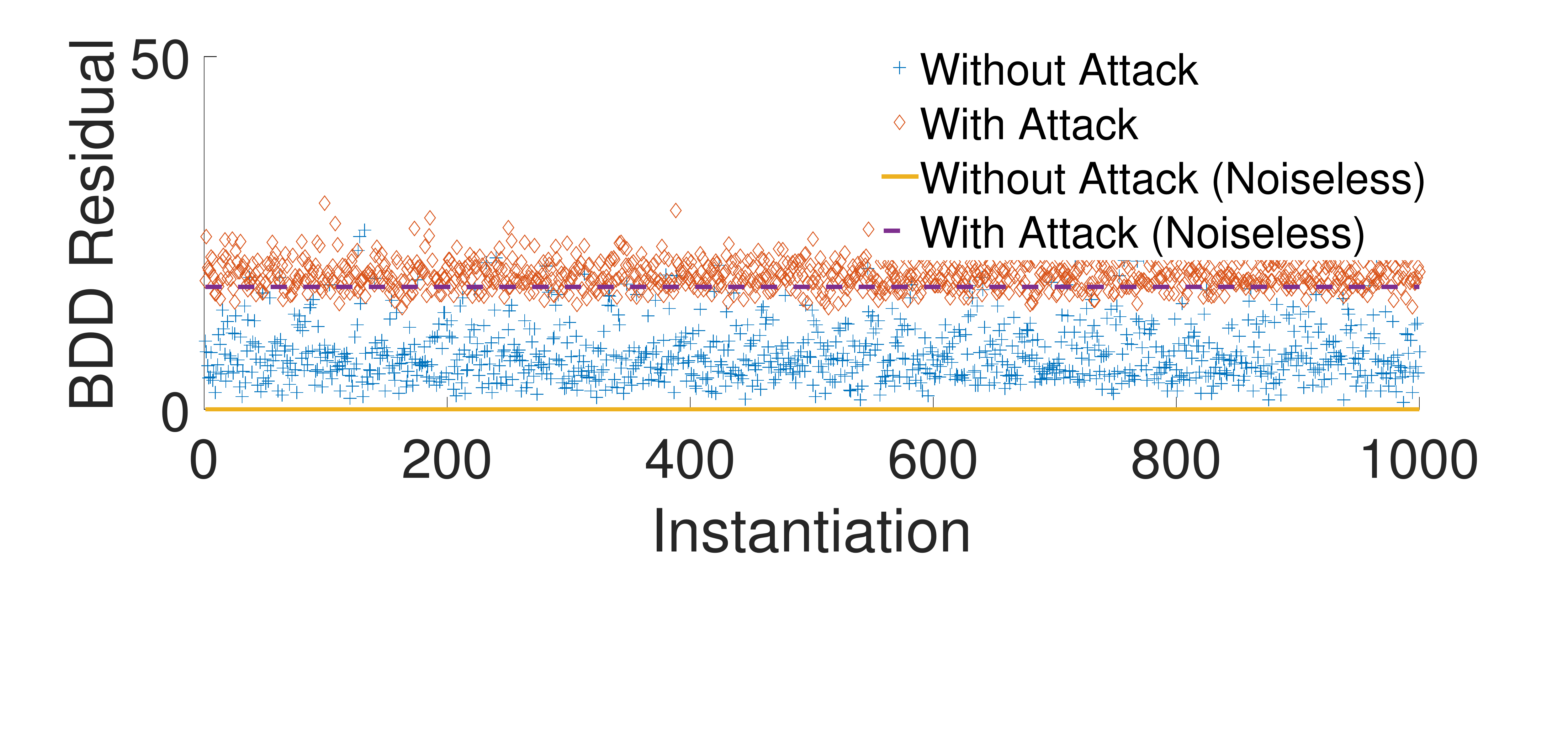}
\captionof{figure}{BDD residual under hypothesis $H_0$ and $H_1$ with and without sensor measurement noise.}
\label{fig:ResidualNoise}
\end{minipage}
\hspace{1em} 
\begin{minipage}[!t]{0.45\textwidth}
\vspace{0pt}
\includegraphics[width=1\textwidth,trim={0 5cm 0 0}]{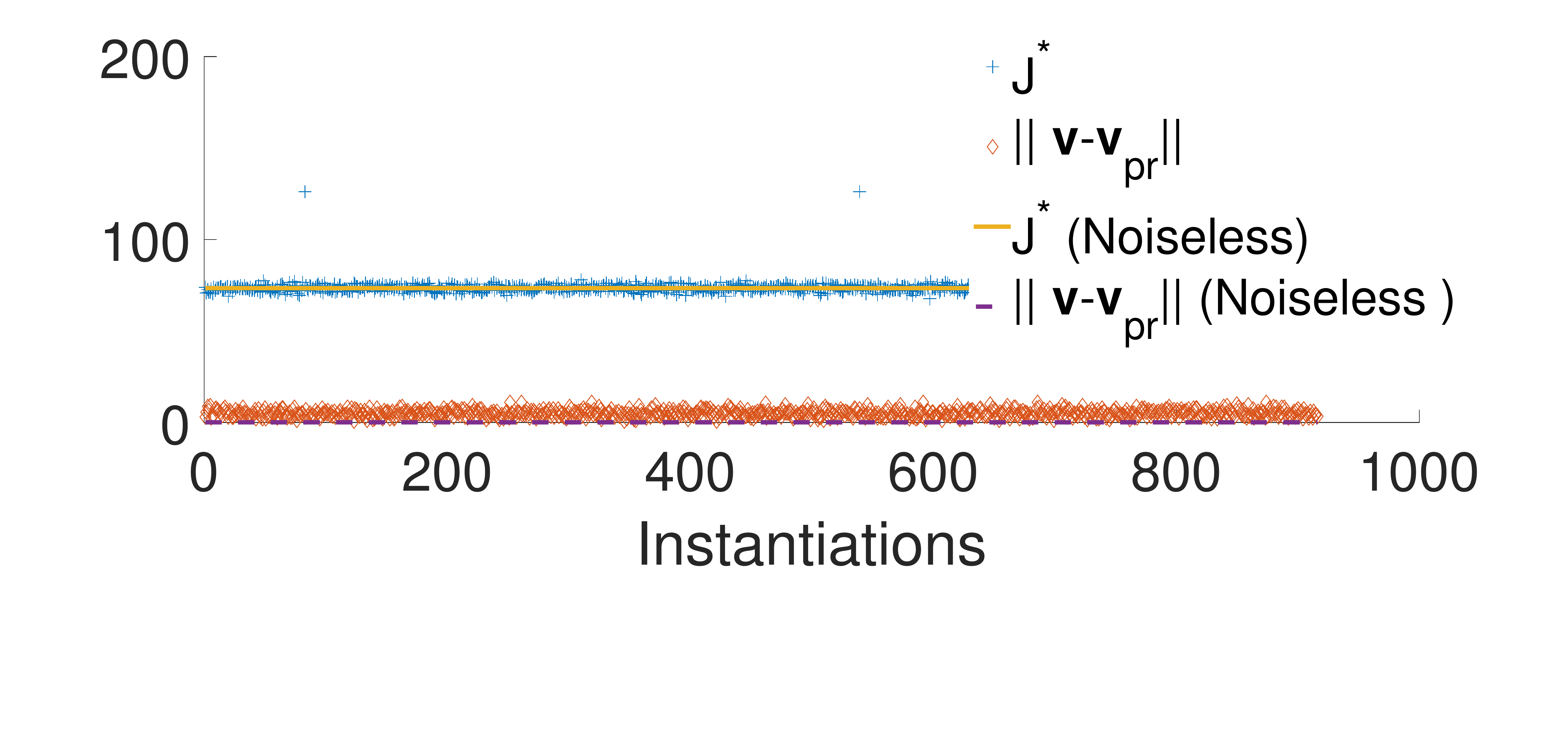}
\captionof{figure}{$J^*$ and $||\vv-\vv_{\text{pr}}||$ under hypothesis $H_0$ with and without sensor measurement noise.}
\label{fig:SADNoise}
\end{minipage}

\end{minipage}
\end{figure}

To assess the performance of the proposed detectors, we 
examine their receiver operating characteristics (ROC)  
curve, obtained by varying the BDD's detection threshold $\tau,$ and the adaptive parameter of the SAD algorithm $\alpha.$ 
Each value of the parameter $\tau$ (and $\alpha$) yields certain FP and MD rates, which
are the x and y-axes of the ROC curve, respectively.
We consider three levels of the measurement noise by varying its standard deviation 
from $0.1\%$ to $0.3\%$ of the full-scale current and voltage sensor readings. 
We consider two different attacks: (i) In Fig.~\ref{fig:ROC_static_FP_MD_BDD_Random} and Fig.~\ref{fig:ROC_static_FP_MD_GAD_Random}, we plot ROC curves for attacks designed without imposing the 
BDD-passing condition (we refer to it subsequently as \emph{random attack}). In particular, we inject an additive attack of $20$ V to the voltage measurement of node $2.$ (ii) In Fig.~\ref{fig:ROC_static_FP_MD_BDD_Bypass} and Fig.~\ref{fig:ROC_static_FP_MD_GAD_Bypass}, we plot the ROC curves for BDD-stealthy attacks.

\begin{figure}[!t]
\centering
\begin{subfigure}{0.48\textwidth}
\includegraphics[width=1\textwidth,trim={0 5cm 0 0}]{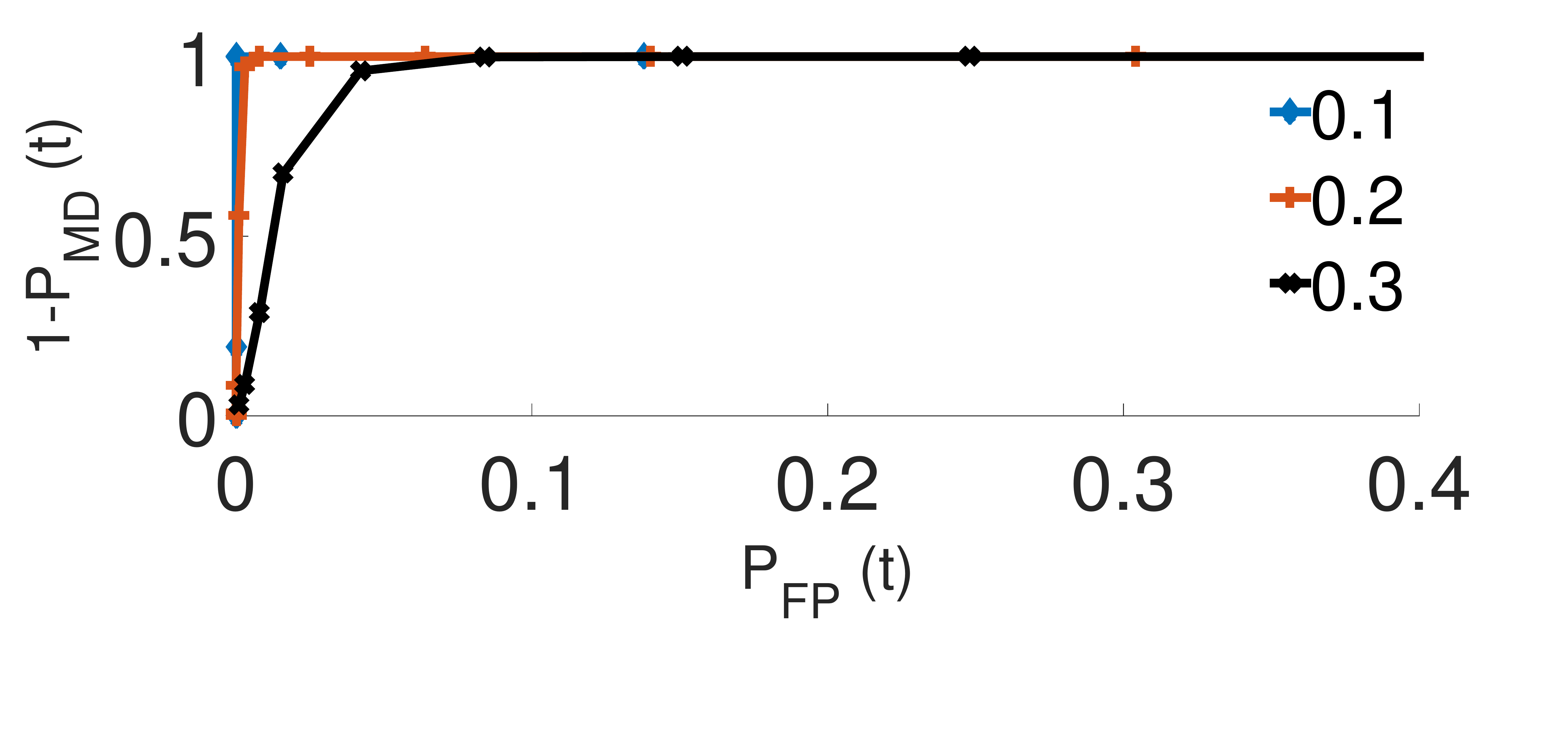}
\caption{}
\label{fig:ROC_static_FP_MD_BDD_Random}
\end{subfigure}
~
\begin{subfigure}{0.48\textwidth}
\includegraphics[width=1\textwidth,trim={0 5cm 0 0}]{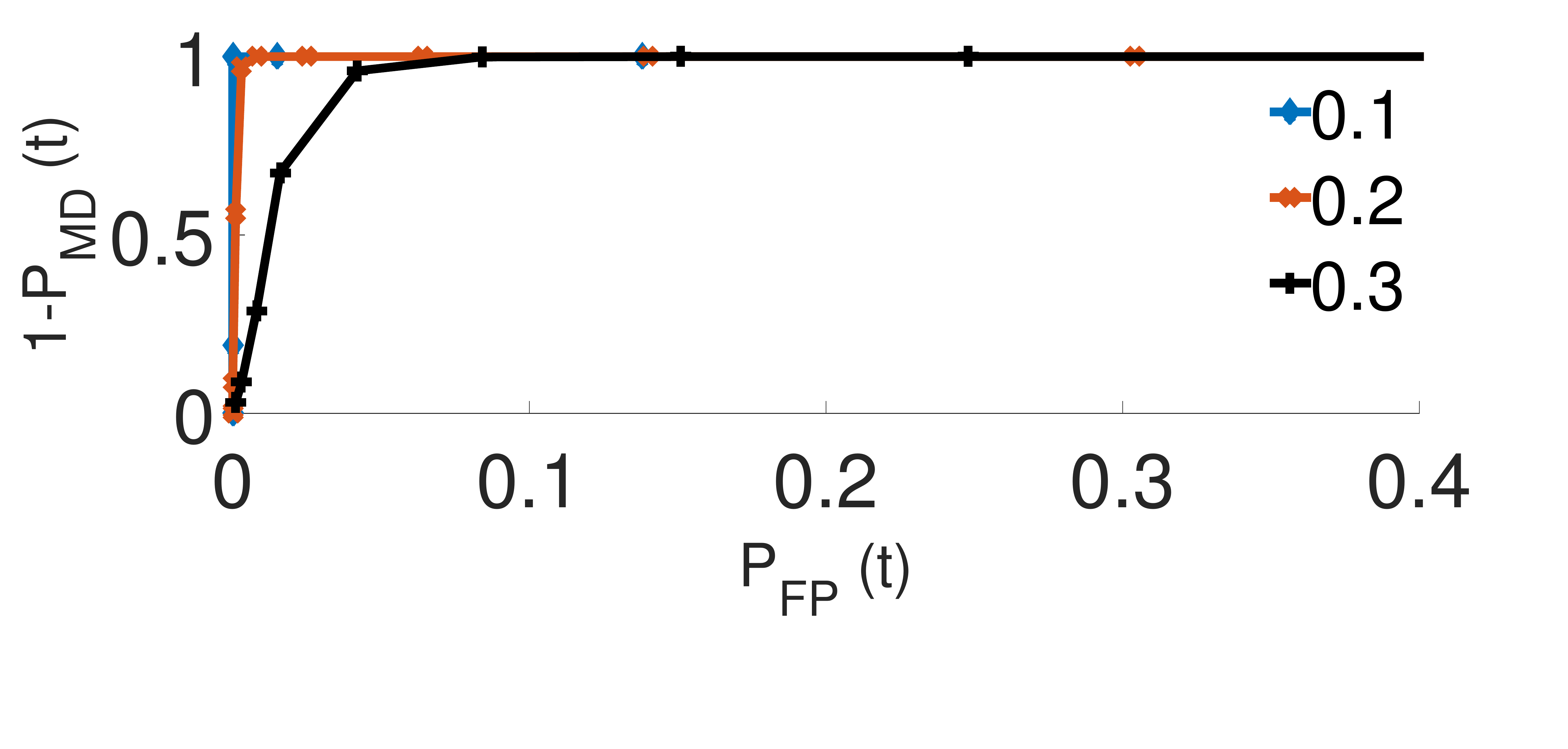}
\caption{}
\label{fig:ROC_static_FP_MD_GAD_Random}
\end{subfigure}
~
\begin{subfigure}{0.48\textwidth}
\centering
\includegraphics[width=1\textwidth,trim={0 5cm 0 0}]{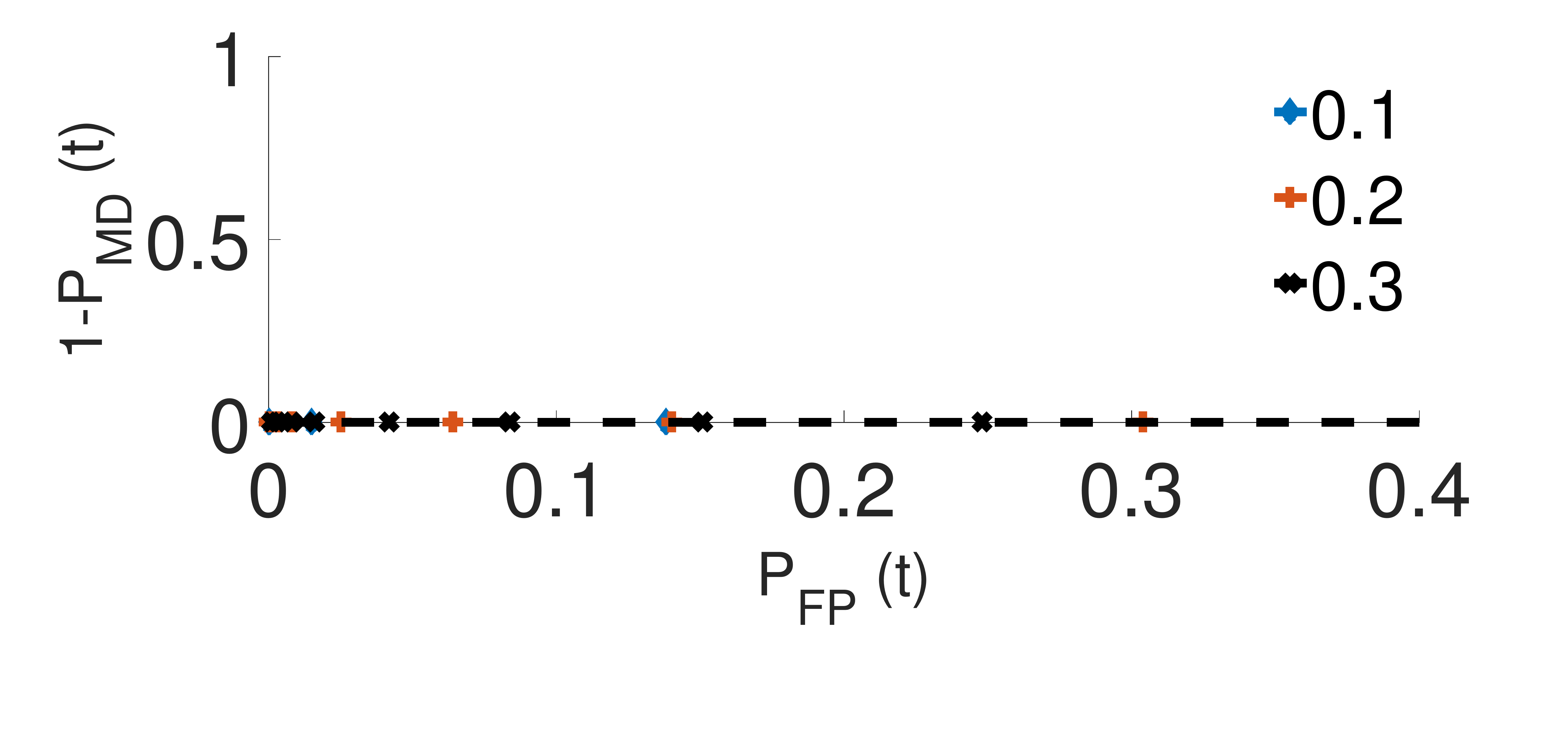}
\caption{}
\label{fig:ROC_static_FP_MD_BDD_Bypass}
\end{subfigure}
~
\begin{subfigure}{0.48\textwidth}
\centering
\includegraphics[width=1\textwidth,trim={0 5cm 0 0}]{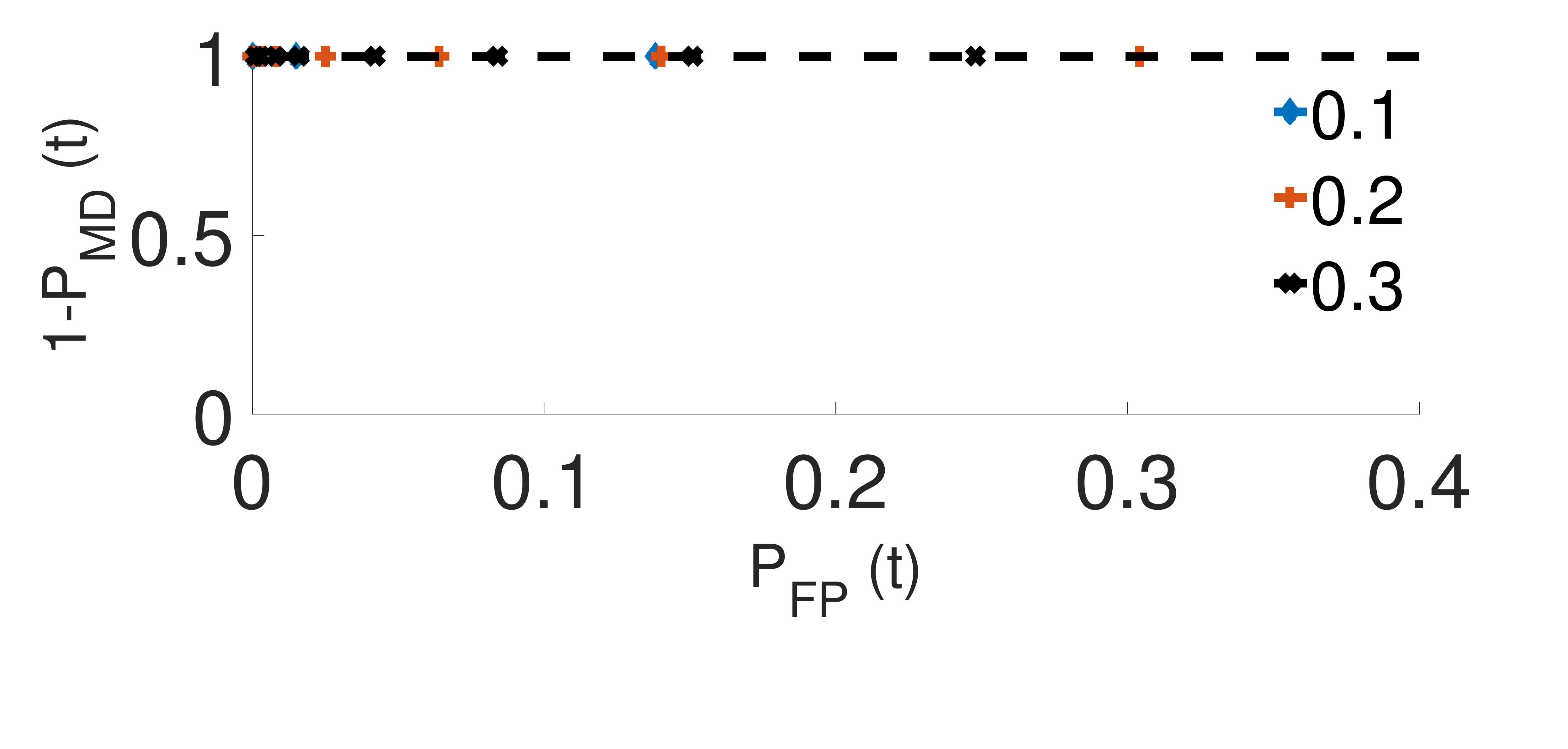}
\caption{}
\label{fig:ROC_static_FP_MD_GAD_Bypass}
\end{subfigure}
\caption{(a) ROC curve of the BDD with random attack. (b) ROC curve of the GAD with random attack. (c) ROC curve of the BDD with BDD-stealthy attack. (d)  ROC curve of the GAD with BDD-stealthy attack.}
\label{fig:ROC}
\end{figure}

\begin{figure}[!t]
\centering
\begin{subfigure}{0.48\textwidth}
\centering
\includegraphics[width=1\textwidth,trim={0 7cm 0 0}]{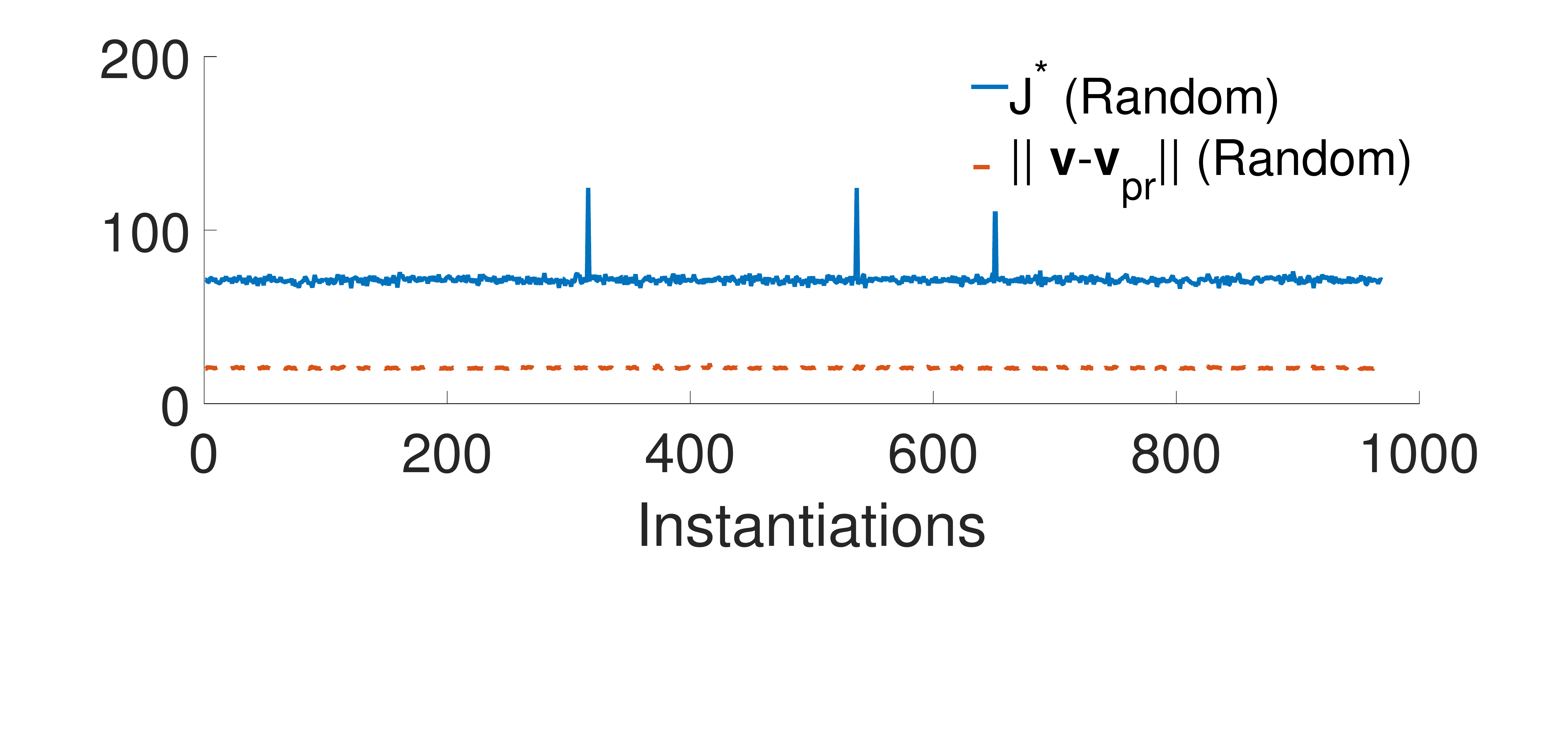}
\end{subfigure}
~
\begin{subfigure}{0.48\textwidth}
\centering
\includegraphics[width=1\textwidth,trim={0 7cm 0 0}]{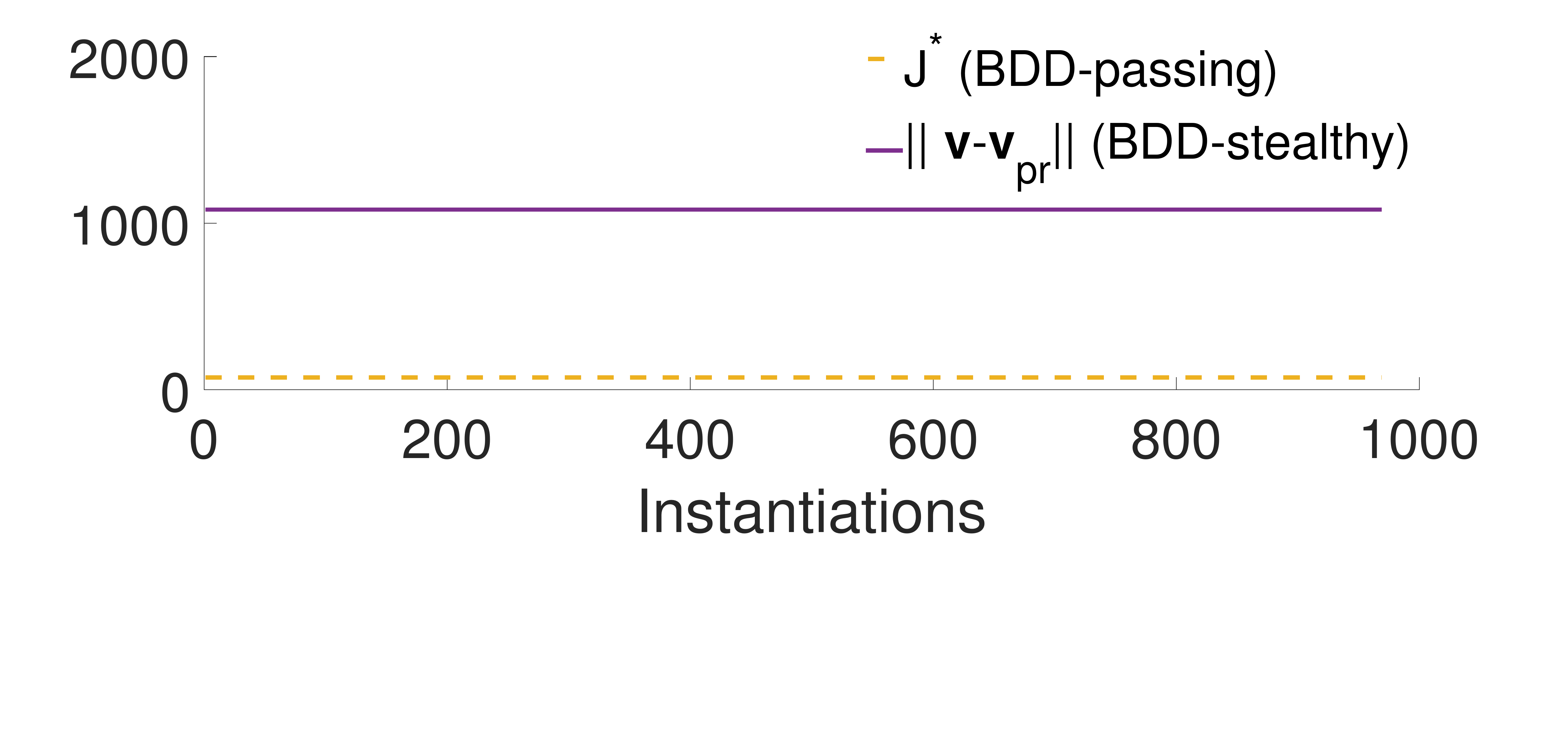}
\end{subfigure}
\caption{Analysis of SAD under random and BDD-stealthy attacks.}
\label{fig:static_MD_SAD_Random}
\end{figure}

The ROC curves under the two attacks exhibit different characteristics, which can be explained 
as follows. As evident from Fig.~\ref{fig:ROC_static_FP_MD_BDD_Random} and Fig.~\ref{fig:ROC_static_FP_MD_BDD_Bypass}, the BDD is effective 
in detecting random attacks but ineffective in detecting the BDD-passing attacks (in fact, the detection rate of BDD-stealthy attacks is $0$). This behavior can be explained by
the nature of BDD's design. Further, by comparing Fig.~\ref{fig:ROC_static_FP_MD_BDD_Random} and Fig.~\ref{fig:ROC_static_FP_MD_GAD_Random}, we can conclude that the SAD only marginally improves the
detection rate of random attacks compared with the stand-alone BDD detector.
However, when we compare Fig.~\ref{fig:ROC_static_FP_MD_BDD_Bypass} and Fig.~\ref{fig:ROC_static_FP_MD_GAD_Bypass}, we observe that for the BDD-stealthy attacks, the 
presence of the SAD significantly improves the detection rate. Specifically, the GAD detection rate is $1$ (no MDs).
This shows the effectiveness of the SAD in detecting BDD-passing attacks.

To understand the performance of SAD in the two cases of random attack and BDD-passing attack, we plot
the values of $J^*$ and $||\vv-\vv_{\text{pr}}||$ in Fig.~\ref{fig:static_MD_SAD_Random}. It can be seen that for random attacks that have been missed by the BDD, the value of $||\vv-\vv_{\text{pr}}||$ is consistently lower
than $J^*$ for all the noise instantiations, which results in MDs. In contrast, for BDD-stealthy attacks, the value of  $||\vv-\vv_{\text{pr}}||$ is greater than $J^*.$ 
This is because in the case of random attack, the attacker only manipulates the measurements from a few sensors (and thus the difference $||\vv-\vv_{\text{pr}}||$ is not high).
However, for the BDD-stealthy attacks, the attacker must manipulate the system measurements in a coordinated manner. In particular, for the system considered in the above simulations, we observe
that the attacker must manipulate the current and voltage measurements of all the nodes. Consequently,
the difference $||\vv-\vv_{\text{pr}}||$ is high.

The FP and MD rates in the above examples are illustrated for a fixed  TPS topology and parameters.
The above discussions give basic understanding on the impact of random measurement noises on the performance of the attack detectors.
However, as the trains change their positions and the status of motion, the TPS parameters change and consequently the FP and MD rates may vary. For instance, the practical GAD detector can have a high FP rate when one or more trains in the TPS changes its status of motion. In order to ensure that the proposed detectors have acceptable performance in these scenarios, in Section~\ref{sec:SimRes_Noise}, we present an adaptation mechanism for the GAD detector, which we call GAD with attack detection window (GAD-W). Extensive simulation results show that the 
GAD-W detector yields consistently low FP and MD rates for the varying TPS configurations. 

We end this section by explaining the introduction of the scaling parameter $\alpha$ in Step~4 of Algorithm~1. Note that in the absence of measurement noise, for a TPS system under attack, the value of $J^*$ is equal to $||\vv-\vv_{\text{pr}}||$ whenever there are only two discrete solutions that satisfy the BDD-passing condition (since the attacker must choose the solution that is different from the true measurements as his attack vector). In such a scenario, the fluctuations in the value of $J^*$ due to sensor measurement noises can often drive its value to greater than $||\vv-\vv_{\text{pr}}||$, leading to a high MD rate. In order to avoid this, we scale down the value of $J^*$ by a parameter $\alpha \in [0,1].$ We note that this problem is unique to the case when there are only two discrete solutions that satisfy the BDD-passing condition, and hence no scaling of $J^*$ is needed in the other cases.

\section{Simulations}
\label{sec:SimRes}

Our analyses in the previous sections address a particular time instant only. In this section, we conduct time-domain simulations with realistic running profiles of trains to illustrate the impact of FDI attacks. We also show the effectiveness of the
BDD in reducing the impact of the attacks, and that of the SAD in detecting those attacks that are BDD-stealthy.

\subsection{Simulation Settings and Methodology}
\label{sec:sim-method}

\begin{figure}[!t]
	\centering
	\begin{subfigure}{0.48\textwidth}
	\includegraphics[width=1\textwidth]{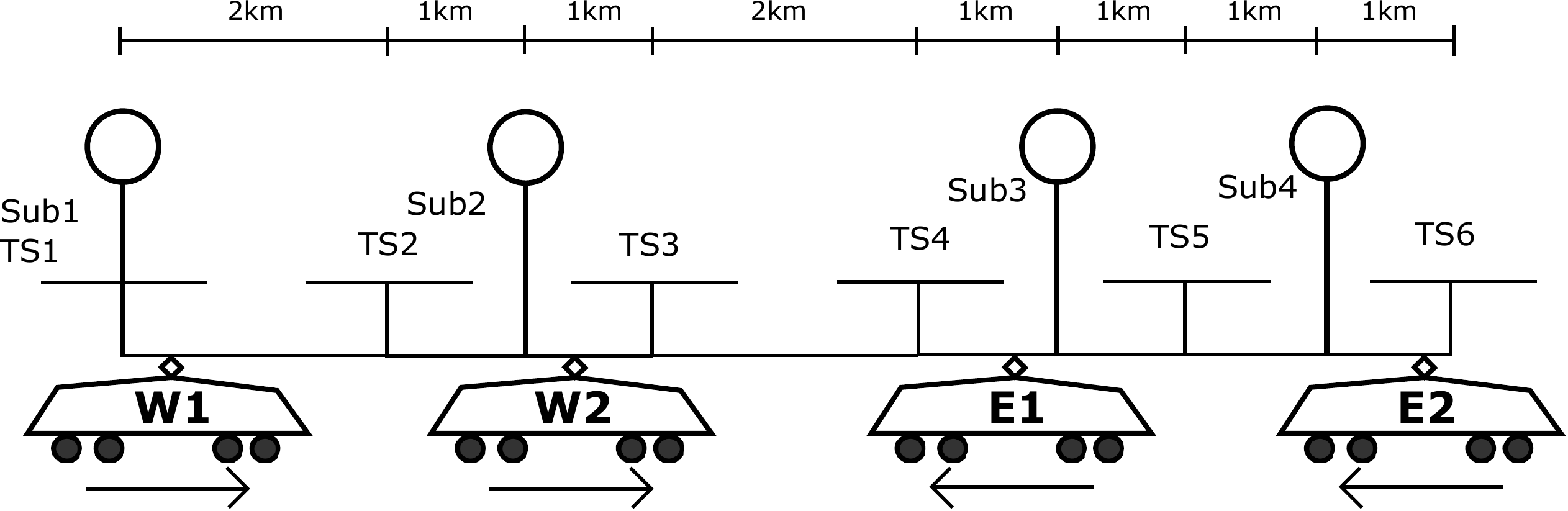}
	\caption{}
	\label{fig:setup}
\end{subfigure}
~
\begin{subfigure}{0.48\textwidth}
\centering
\includegraphics[width=1\textwidth,trim={0 0 0 0}]{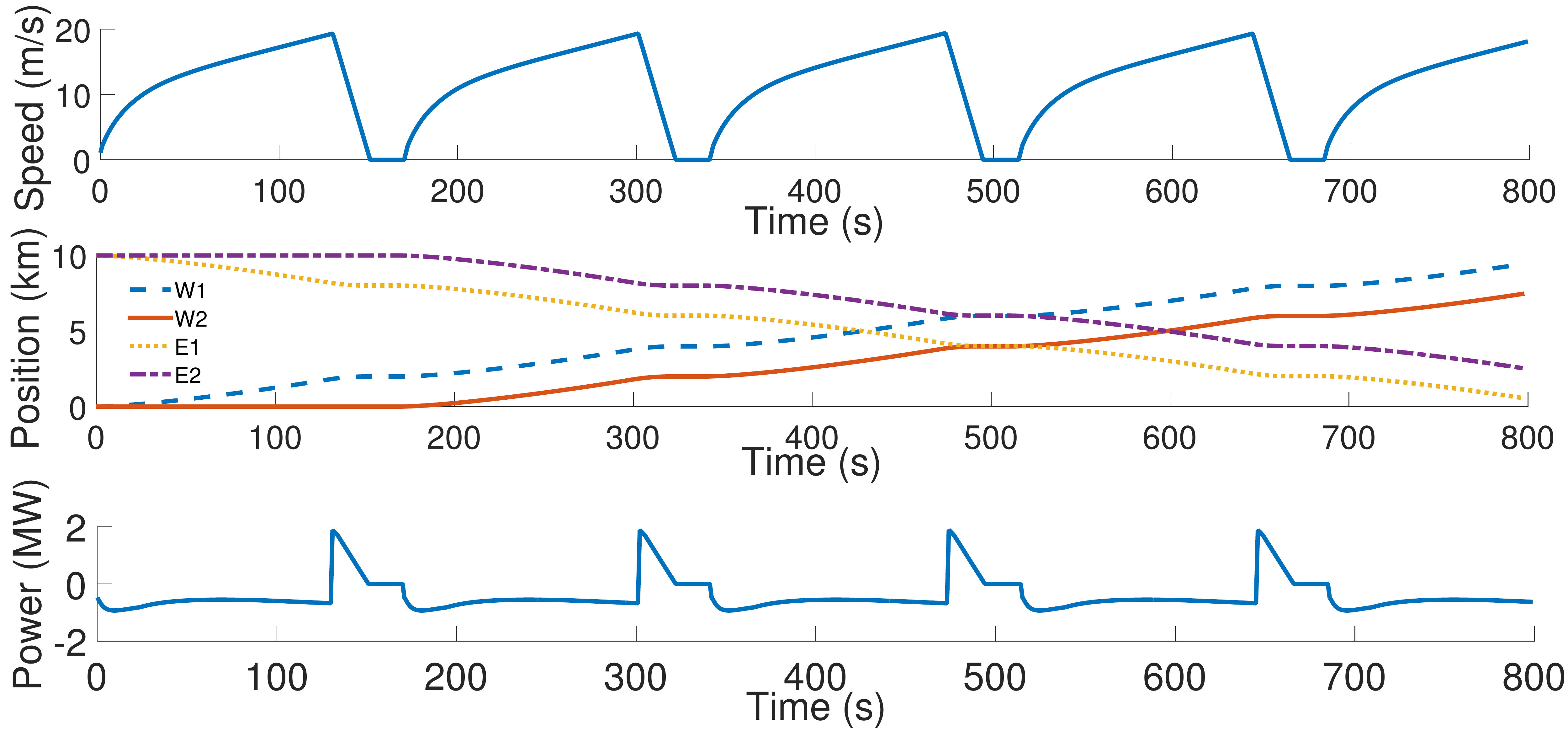}
\caption{}
\label{fig:pos_profile} \label{fig:sp_profile}
\end{subfigure}
\caption{(a) System set-up for simulations. Sub - Substations, TS - Train stations, W - Trains departing from the west, E - Trains departing from the east. (b) Train speed (top plot), position (middle plot), power demand, and regeneration capacity (bottom plot) over time. Power demand is negative and regeneration capacity is positive.}
\end{figure}

As Fig.~\ref{fig:setup} illustrates, we simulate a TPS consisting of
four trains (labeled W1, W2, E1, and E2), four substations (labeled Sub1 to Sub4), and six train
stations (labeled TS1 to TS6).
The parameters of the TPS are identical to those in Table~\ref{tbl:BFS_parameters}.
The positions of the substations and the train stations are shown in Fig.~\ref{fig:setup}.  
The trains W1 and W2 start their journeys from TS1 and
travel from west to east, whereas the trains E1 and E2 start their journeys from TS6 and
travel from east to west. The trains W1 and E1 depart at time zero and the trains W2 and E2 depart at the 170th second.
At each of the train stations, the trains stop for a duration of 20 seconds. 
Each train follows the same speed profile as shown in the top part of Fig.~\ref{fig:sp_profile}. The second plot of Fig.~\ref{fig:sp_profile} shows the trains' positions over time. 
Each train switches between traction and braking modes during the simulation, and its power demand and regeneration capacity over time are shown in the bottom plot of Fig.~\ref{fig:sp_profile}. This plot is derived based on mechanical energy consumption of the train under the specified running profile, and with an efficiency ratio of 70\% for the traction mode \cite{SuTangRoberts2015} and 40\% for the braking mode \cite{acikbas2007parameters}
of converting kinetic energy into electrical energy.
We simulate the TPS for 800 seconds at a time granularity of one second. 

To simulate attacks, the attacker injects an attack vector computed using the methods given in Sections \ref{sec:Cyberattacks} and \ref{sec:BDD} every second. In the absence of BDD, the attacker compromises the voltage and current measurements of all the train nodes. 
In the presence of BDD, the attacker tampers with the voltage and current measurements of all the train and substations nodes as well as the position information of the train nodes. The position information of substations cannot be compromised since their locations are fixed and known {\em a priori}. The maximum errors that the attacker can introduce to the voltage, current, and position measurements, as described in \eqref{eqn:feas_org1d_nse}, \eqref{eqn:feas_org1e_nse}, and \eqref{eqn:pos_attack}, are set as $\Delta V_i = 50\,\text{V}$, $\Delta I_i = 200\,\text{A}$, for $i \in \mathcal{N}_a$, unless otherwise specified. 
The choice of these parameters is made taking into account two practical considerations: (i) the measurement noise level (whose standard deviation is considered to be $\approx 0.5 \%$  of the full-scale voltage and current sensor readings \cite{sensor_accuracy}) (ii) the change in voltage and current measurements of the TPS between any two successive simulation instants (which be observed to be in the considered range based on extensive simulations). Note that if the variation of voltage and current is within this range, they pass the data-quality checks.

The simulations are carried out in MATLAB. The constrained optimization problems are solved using the \texttt{fmincon} function of MATLAB with the \texttt{MultiStart} algorithm. In the absence of attack, to compute the system state, we use the \texttt{fmincon} with a constant objective function and the electrical models and trains' local control laws presented in Section~\ref{sec:Tract_Power} as the constraints. We also use the function to compute the safety attack vectors under the heuristic approach and the optimal efficiency attack vectors.
If at any time instant, the \texttt{fmincon} function
returns an attack vector that is the same as the true system state, the attacker does not launch an attack, since the attack will not have
any impact.
Step 2 of the SAD algorithm is also implemented using the \texttt{fmincon} function.

Although our analysis in this paper is general and applicable to a TPS network of arbitrary size and topology, for simulations we consider a small-scale TPS in Fig.~\ref{fig:setup}. The rationale is two fold. First, the attacker may find it difficult to coordinate his attacks on a large number of geographically distributed trains. 
Computing resources may present another barrier for large-scale attacks. 
A more credible scenario is for the attacker to focus on one or a few trains in a TPS section. Second, since real-world TPS networks are mostly radial \cite{AbrahamssonThesis2012}, the impact of a focused and localized attack will not propagate over long distances. In view of these factors, we use the small-scale TPS to represent well a TPS section in a large system.

Moreover, to simplify our simulations, we do not consider overcurrent control. Specifically, we set the triggering threshold $V^{\text{Tr}}_{i,\min}$ to a low value, so that overcurrent control will not be activated.
As a result, the trains' speed profiles will not change because the trains need not curtail their power consumption.
At any time instant, therefore, a train's power consumption
is equal to its power demand during acceleration.
Because of this simplification, we do not simulate attacks on tractioning trains, which would alter the tractioning trains' power consumption and change their running profiles. Although we can simulate overcurrent control and attacks on tractioning trains by extending our simulator to admit changeable running profiles, the simulations reported in this paper already provide interesting understanding and insights into the impact of attacks and the effectiveness of countermeasures.

\subsection{Simulation Results}

\subsubsection{Efficiency Attacks}
The first set of simulations evaluates the impact of efficiency attacks on the TPS without BDD. Fig.~\ref{fig:dynamic_efficiency_noSE} shows the power absorbed/injected by the train E1 in the presence and absence of attacks. We can see that the efficiency attacks cause the regeneration trains to inject less power into the power network (please see the encircled regions, e.g., from 302th to 315th second for the train E1).
To calculate the loss in system efficiency, we ignore the time instants when all the trains are in traction mode, since we do not simulate attacks on the tractioning trains as discussed in Section~\ref{sec:sim-method}. As a result, the efficiency attacks cause a reduction of $28.3\%$ in the total energy adsorbed by the substations compared with the case of no attacks, during the time periods when there is at least one regenerating train under attack.

The second set of simulations evaluates the impact of efficiency attacks on the TPS with BDD. Similar to Fig.~\ref{fig:dynamic_efficiency_noSE}, Fig.~\ref{fig:dynamic_efficiency_SE} shows the power absorbed/injected in the absence and presence of attacks. It can be seen that in Fig.~10~(b), the curve for the power absorbed/injected by trains in the presence of attacks follows that for the absence of attacks more closely, in comparison to the respective curves in Fig.~10~(a). (Please see the encircled parts of the two figures.) Thus, although the efficiency attack can still induce the regenerating trains to inject less power to the power network, it causes a reduction of $6.2\%$ only in the total energy adsorbed by the substations, during the time periods when there is at least one regenerating train under attack. This is in contrast to the $28.3\%$ for the TPS without BDD.

We also examine the effect of efficiency attacks on the TPS with BDD under different settings of $\Delta s_i$ and $\Delta V_i$
in Fig.~\ref{fig:delta_s} and Fig.~\ref{fig:delta_v}, respectively. 
From these figures, we can see that at smaller settings of $\Delta s_i$ and $\Delta V_i$, the efficiency loss caused by the FDI attack diminishes. For instance, the efficiency loss is as low as $1.37 \%$ when $\Delta s_i = 0.1\,\text{km}$.
In practice, the TPS monitor can estimate the present train position based on the train's speed and its position at the previous time instant when it was known that there were no attacks. The present position reading can be compared with the estimated position using \eqref{eqn:pos_attack}. The setting of $\Delta s_i$ should consider natural errors of train positioning systems and the estimation error.
Existing train positioning systems such as GPS and Balise can achieve an accuracy of five to ten meters \cite{GPS2012}, \cite{balise}.
Thus, it is reasonable to assume that the combined
effect of the train positioning system error and the estimation error is less than $0.1\,\text{km}$.
Our results show that by properly tuning the BDD's attack detection parameters (e.g., $\Delta s_i$ and $\Delta V_i$), the efficiency loss caused by FDI attacks can be significantly reduced. 

\begin{figure}[!t]
\begin{subfigure}{0.48\textwidth}
\centerline{\includegraphics[trim = 0mm 55mm 0mm 10mm,width=6.5cm,height=3cm,keepaspectratio]{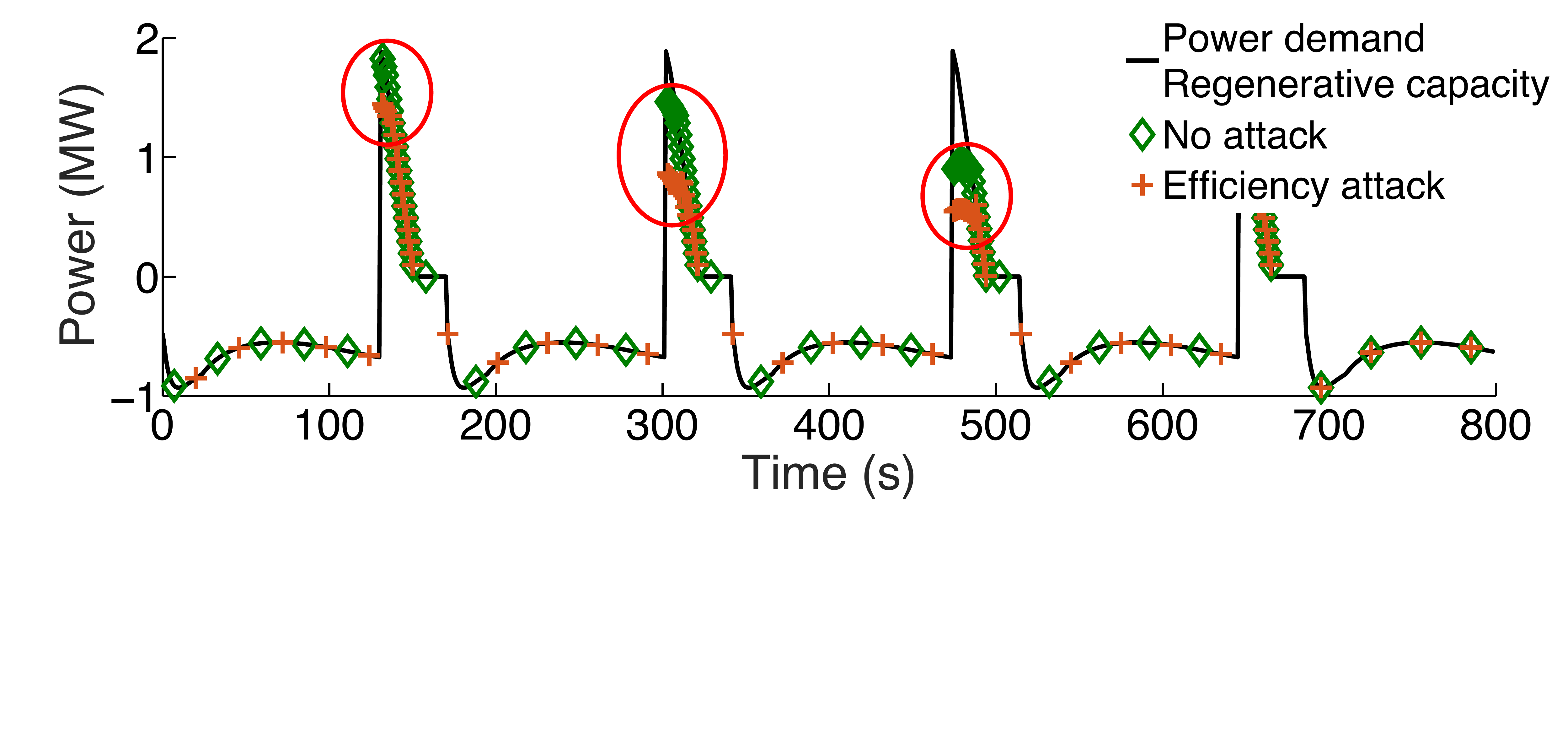}}
\caption{Effect of efficiency attacks on Train E1 in the absence of BDD.}
\label{fig:dynamic_efficiency_noSE}
\end{subfigure}
~
\begin{subfigure}{0.48\textwidth}
\centerline{\includegraphics[trim = 0mm 55mm 0mm 10mm,width=6.5cm,height=3cm,keepaspectratio]{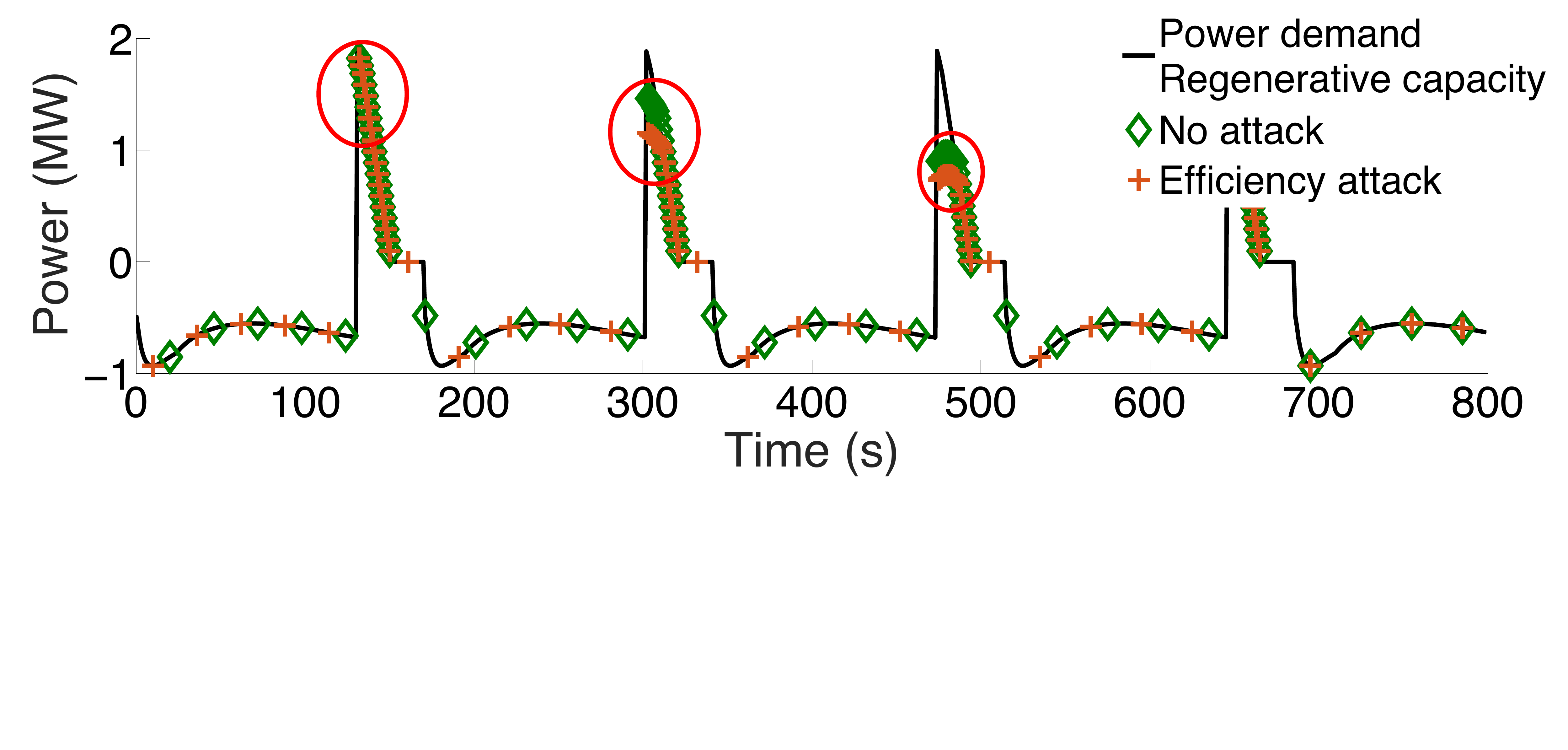}}
\caption{Effect of efficiency attacks on Train E1 in the presence of BDD.}
\label{fig:dynamic_efficiency_SE}
\end{subfigure}
\caption{Effect of efficiency attacks on Train E1. Circled regions highlight the time slots where the two curves (with and without attack) diverge. Note that the curve with attack follows the curve without attack more closely in the presence of BDD.}
\vspace{-1em}
\end{figure}

\begin{figure}[!t]
\centering
\begin{subfigure}{0.4\textwidth}
\centering
\includegraphics[width=0.8\textwidth,trim={0 1cm 4cm 0cm}]{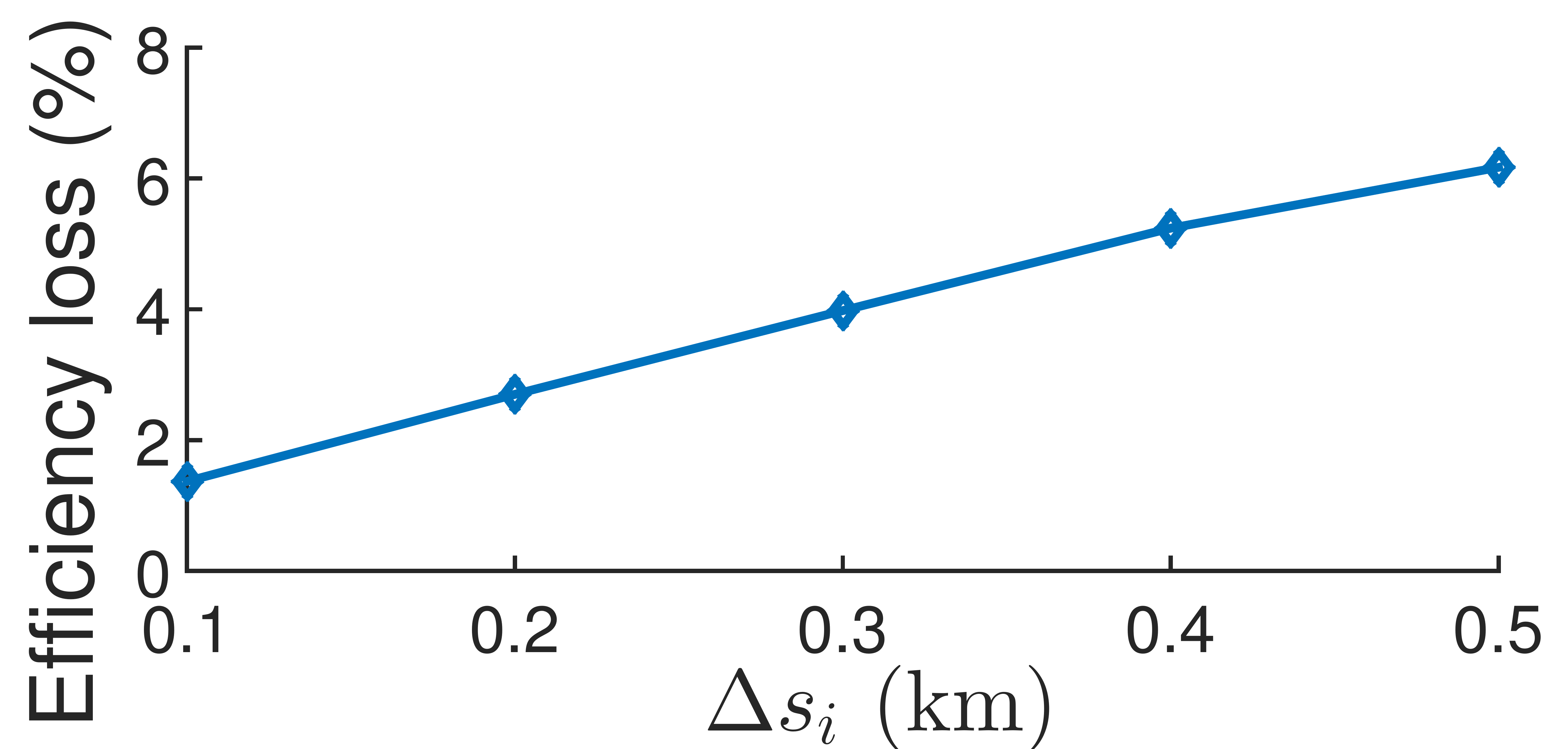}
\caption{Impact of $\Delta s_i$ ($\Delta V_i = 50\,\text{V}$, $\Delta I_i = 200\,\text{A}$)}
\label{fig:delta_s} 
\end{subfigure}
\hspace{1cm}
\begin{subfigure}{0.4\textwidth}
\centering
\includegraphics[width=0.8\textwidth,trim={0 1cm 4cm 0cm}]{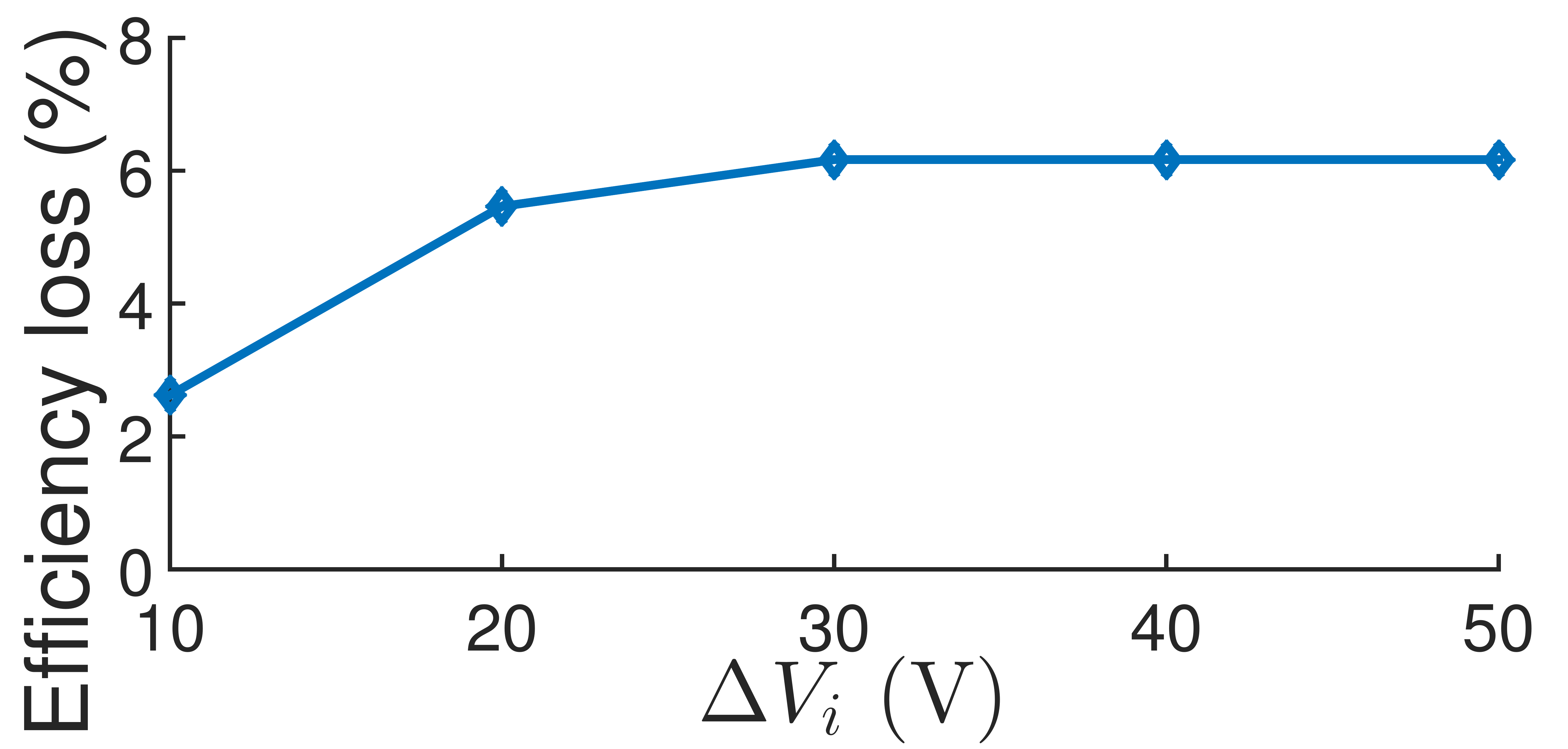}
\caption{Impact of $\Delta V_i$ ($\Delta s_i = 0.5\,\text{km}$, $\Delta I_i = 200\,\text{A}$)}
\label{fig:delta_v} 
\end{subfigure}
\caption{Effect of efficiency attacks on the TPS with BDD under different settings of $\Delta s_i$ and $\Delta V_i.$}
\end{figure}

\subsubsection{Safety Attacks}
We conduct two sets of simulations to evaluate the impact of safety attacks on the TPS: the first one without BDD and the second with BDD. Under safety attacks, the regenerating trains inject more power into the power network than that under no attacks, resulting in increased voltages. We say that the TPS experiences a safety breach when at least one node in the TPS experiences a safety breach.

Table~\ref{tbl:safety_cmp} summarizes the time durations of safety breaches under 
the two sets of simulations. We consider BDD with different settings of $\Delta s_i$. 
It can be observed that without BDD, the TPS experiences safety breaches for a total of eight seconds. The prolonged overvoltage may cause safety incidents. However, with BDD we see that, when $\Delta s_i$ is in the range of $0.1\,\text{km}$ to $0.3\,\text{km}$, the attack causes no safety breaches during the simulation. As discussed previously, the setting $\Delta s_i = 0.1\,\text{km}$ is appropriate in practice. Hence, this set of results shows that by appropriately setting the BDD parameters, safety breaches can be nearly eliminated.

\begin{table} [!t]
\caption{Time duration while the TPS experiences safety breaches under different settings of $\Delta s_i$ in the presence of BDD.}
\centering
\setlength\extrarowheight{2pt}
\begin{tabular}{|c|c|c|c|c|c|c|c|}
\hline
	$\Delta s_i$ (km) & No BDD & 0.5 & 0.4 & 0.3 &0.2& 0.1 \\ \hline 
	Time duration with safety breaches (second)& 8 & 4 & 1 & 0 & 0 & 0 \\ \hline
\end{tabular}
\label{tbl:safety_cmp}
\end{table}

\subsubsection{SAD Algorithm}
\label{sec:SAD}
The last set of simulations evaluates the effectiveness of SAD in detecting attacks that have bypassed the BDD. In this set of simulations, we set $\alpha = 1$ (since the scaling is not necessary in the absence of measurement noise). Furthermore, we use $p = 2$ in our evaluations.\footnote{Simulation results conducted with $p = 2$ and $p = \infty$ yielded similar performance of the SAD algorithm (in terms of the attack detection rate).} For each time instant,  among the discrete solutions to the BDD bypass condition discussed in Section~\ref{subsec:SAD}, the attacker tactically chooses the one closest to the true system state in the sense of $p$-norm distance. We compare our {\em practical approach} where the $\vv_{\text{pr}}$ is the nodal voltage vector at the previous time instant (cf.~Algorithm~1), with an {\em oracle approach} where the $\vv_{\text{pr}}$ is the nodal voltage vector at the present time instant in the absence of attack.
For the oracle approach, we observe that the $||\tilde{\vv}-\vv_{\text{pr}}||_p$ is consistently higher than the $J^*$ for the entire simulation. This suggests that the oracle approach can detect the onset of a BDD-stealthy attack launched at any time instant.
For the practical approach, we observe that the $||\tilde{\vv}-\vv_{\text{pr}}||_p$ is higher than the $J^*$ for 96\% of the simulation time. For the remaining 4\% of simulation time, the practical approach will miss the attack onset because of a significant change of $\vv$ from the previous time instant to the present. This shows that the practical approach can detect the attack onset with a high detection probability. 

We note that as the size of the TPS increases (in terms of the number of trains and substations under consideration), the number of constraints for the SAD algorithm as well as the solutions to the BDD-passing constraints will increase. Implementing the SAD algorithm may become computationally complex. However, as we pointed out earlier, it is often sufficient to consider only a small section of TPS for security analysis. Thus, in practical application, the computational overhead
of the SAD algorithm will be acceptable.

\subsection{Simulation Results With Random Sensor Measurement Noises}
\label{sec:SimRes_Noise}

\begin{figure}[!t]
\centering
\begin{subfigure}{0.48\textwidth}
\centering
\includegraphics[width=1\textwidth,trim={0 7cm 0 0}]{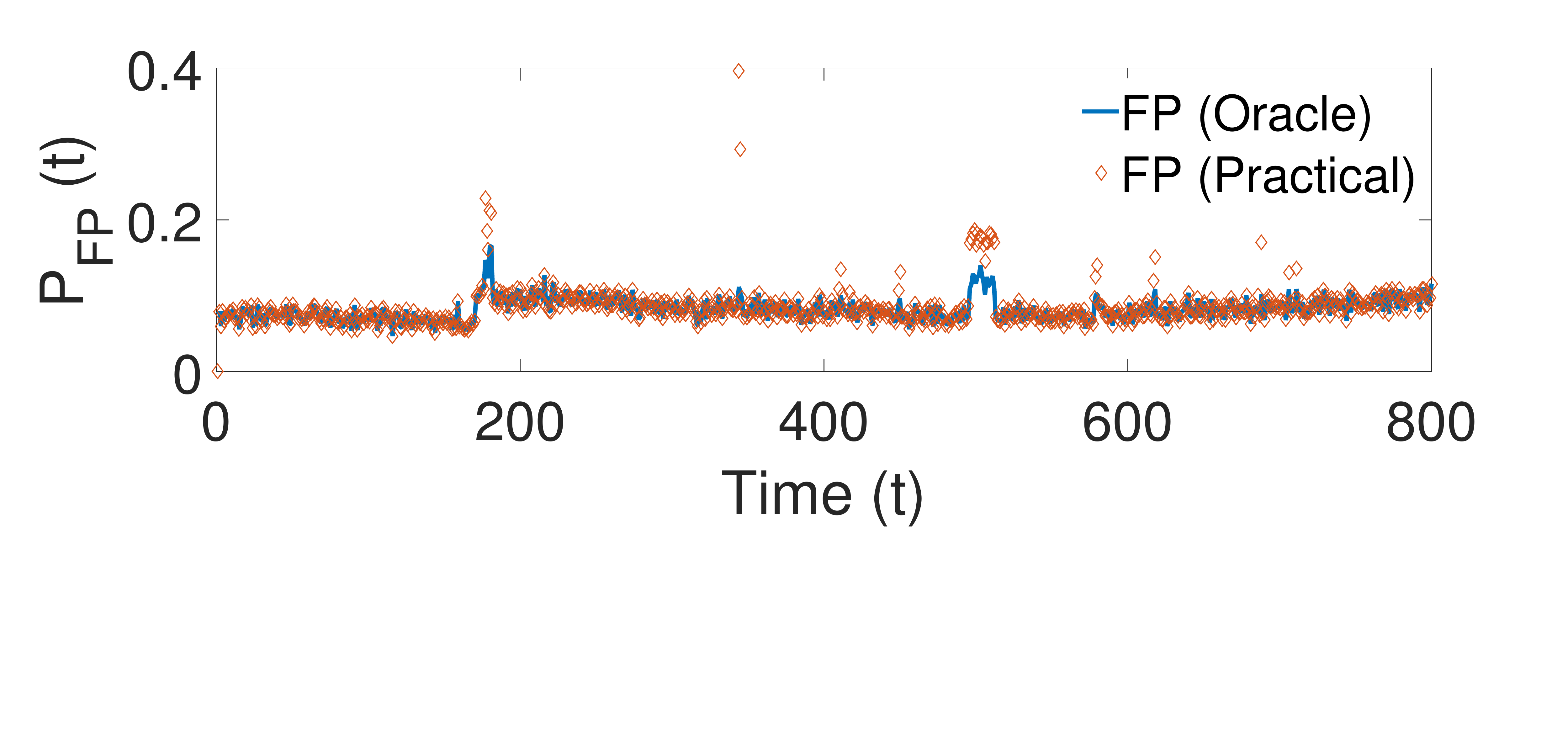}
\caption{ FP rates of the GAD as a function of time.}
\label{fig:dynamic_oracle_FP}
\end{subfigure}
~
\begin{subfigure}{0.48\textwidth}
\centering
\includegraphics[width=1\textwidth,trim={0 6cm 0 0}]{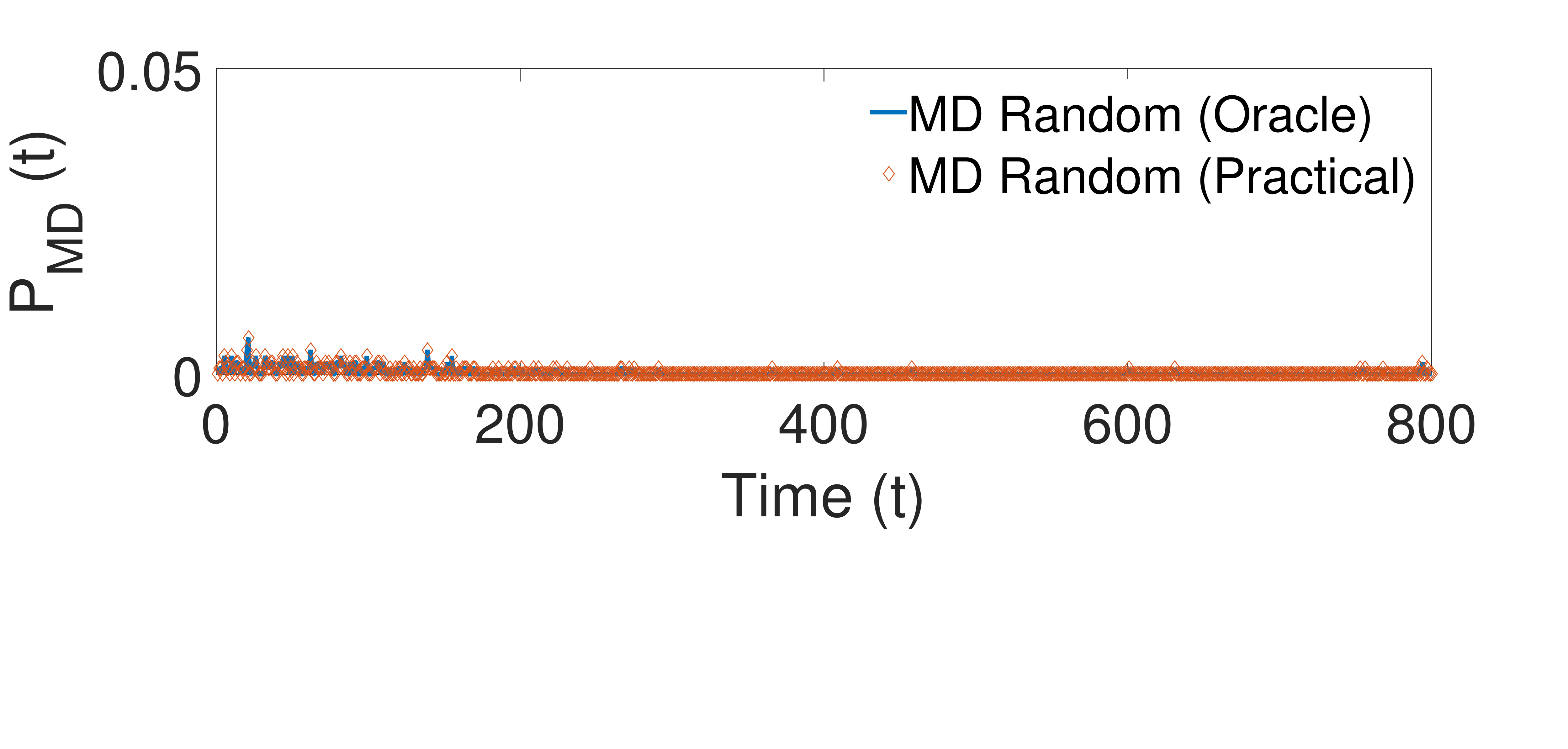}
\caption{MD rates of the GAD as a function of time for random attacks.}
\label{fig:dynamic_oracle_Random}
\end{subfigure}
~
\begin{subfigure}{0.48\textwidth}
\centering
\includegraphics[width=1\textwidth,trim={0 8cm 0 0}]{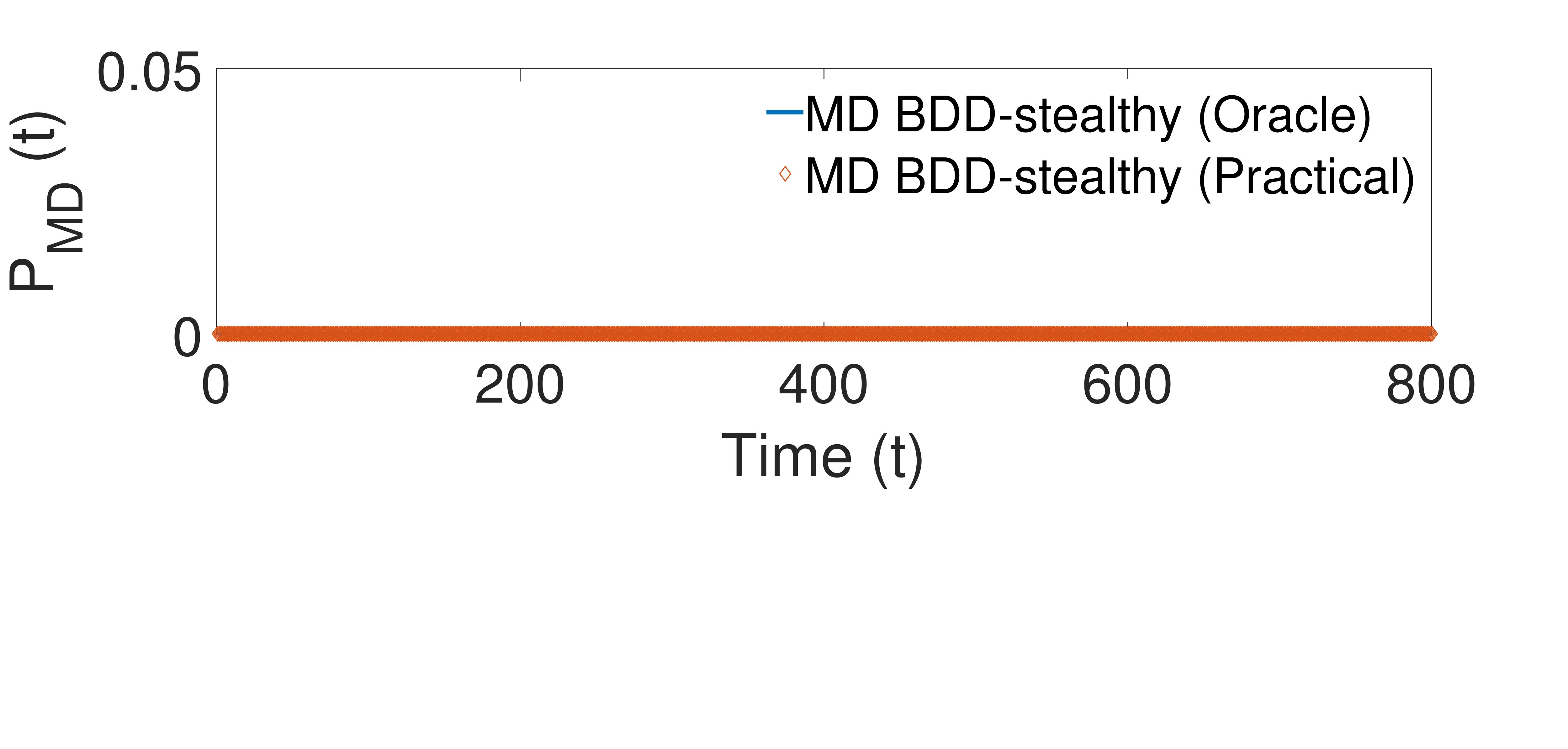}
\caption{MD rates of the GAD as a function of time for BDD-stealthy attacks.}
\label{fig:dynamic_oracle_Bypass}
\end{subfigure}
~
\begin{subfigure}{0.48\textwidth}
\centering
\includegraphics[width=1\textwidth,trim={0 8cm 0 0}]{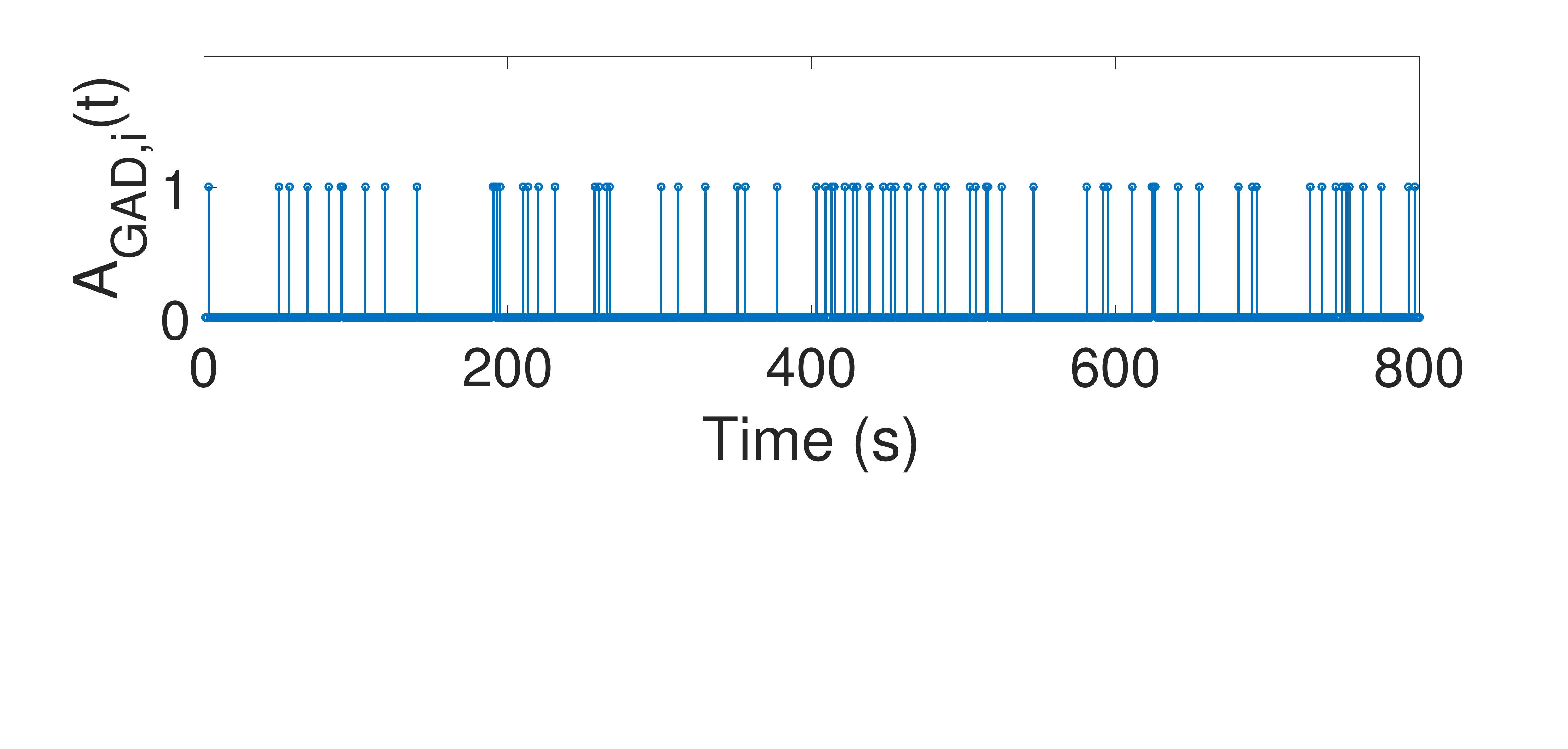}
\caption{GAD alarms over the simulation interval under no attacks.}
\label{fig:dynamic_FP_oracle_gad_occurrence}
\end{subfigure}
\caption{ FPs and MDs of the GAD. $\tau=16$ and $\alpha = 0.9$.}
\end{figure}

In this subsection, we examine the empirical FP and MD rates of the GAD 
at different time instants of the $800$ second simulation interval.
To compute these quantities, we run $N$ simulation runs. We let $Z_{\text{GAD}, i}(t)$ and $\Xi_{\text{GAD}, i}(t)$ denote the indicator variables representing FPs and MDs at a time instant $t \in \{1,\dots,800\}$ during the simulation run $i \in \{1,\dots,N\}.$ The empirical
FP and MD rates at time $t$ are then computed as
$P_{\text{FP}}(t)  = \frac{1}{N}\sum^{N}_{i = 1} Z_{\text{GAD}, i}(t),
P_{\text{MD}}(t)  = \frac{1}{N}\sum^{N}_{i = 1} \Xi_{\text{GAD}, i}(t).$
In our simulations, we set $N = 1000$ and the noise level to
 $0.3\%$ of the full-scale voltage and current sensor readings. 
The BDD detection threshold $\tau$ is set to $16,$ and $\alpha = 0.9$ for the SAD. The value of $\alpha$ was tuned numerically by observing observing the values of $||\tilde{\vv} - \vv_{\text{pr}}||_p$ and $J^*$ in the scenario when the BDD-passing constraint has only two solutions. The chosen value of $\alpha$ is sufficient to eliminate MDs.

Fig.~\ref{fig:dynamic_oracle_FP} shows the FP rate of the GAD, and Figs.~\ref{fig:dynamic_oracle_Random} and \ref{fig:dynamic_oracle_Bypass} show the MD rates of the GAD for random and BDD-stealthy attacks. 
For random attacks,  we inject an additive attack of 20 V to the voltage measurement of the leftmost train (in Fig.~\ref{fig:setup}). We make the following observations. First, we observe that the FP and MD rates fluctuate over time, since the TPS topology and parameters change. (Recall that the TPS topology and parameters depend on the position and the power drawn/injected by the trains.) 
Second, we observe that under the considered settings, both the oracle and practical GAD detectors yield very low MD rates at all time instants. Thus we conclude that by appropriately tuning the parameters of 
the BDD and SAD detectors ($\tau$ and $\alpha$), the MD rate of the GAD can be reduced to a very low value. Third, we observe that while the FP rate is low for most of the simulation interval, 
there are a few time instants at which the FP rate is relatively high, particularly for
the practical GAD detector (e.g., from $t = 497$ to $t = 511$, the FP rate $\approx$ 0.2). Furthermore, we observe that these time instants correspond to when one or more trains change their motion status from tractioning to braking mode, thus resulting in a drastic change in the system state. Recall that an accelerating train draws power from the network resulting in a voltage drop whereas as a braking train injects power resulting in a voltage raise. In these cases, the difference $||\vv-\vv_{\text{pr}}||$ can be high for the practical GAD detector since $\vv_{\text{pr}}$ is estimated only based on the historical values.

However, in practice, an extremely low FP rate is desired, since otherwise the system operator would have to frequently initiate unnecessary mitigation that may be disruptive. Thus, in what follows, we propose an adaptive version of the GAD, which we call GAD with attack detection window (GAD-W). GAD-W will give an extremely low FP rate in the presence of sensor measurement noises.

\subsubsection{GAD with Attack Detection Window}
The GAD-W detector applies an AND rule to fuse the detection results in an \emph{attack detection window}, i.e., instead of declaring the presence of an attack based on a single alarm, GAD-W waits for consecutive alarms over several time slots before declaring it. 
In the following, we first formally state the GAD-W detector and then provide the intuition behind its design. Denote by $A_{\text{GAD}}(t) \in \{0,1\}$ the detection result of the GAD at time $t$ and by $W \in \mathbb{N}$ the window size. The GAD-W detector raises an alarm only if there is an alarm at all the time instants within the attack window, i.e.,
\begin{align}
A_{\text{GAD-W}}(t) = A_{\text{GAD}}(t) \land A_{\text{GAD}}(t+1) \land \dots \land A_{\text{GAD}}(t+W-1).
\end{align}
The rationale is that in the absence of attacks, the occurrence of GAD alarms can be due to two factors: (i) the fluctuations of 
BDD residual induced by the measurement noise, or (ii) a drastic change in the system state between consecutive time slots. In the above two cases, the BDD and SAD will raise an alarm, respectively. The first case is a randomly occurring event (due to noise) and the second is a sparsely occurring event. Thus, the probability of having consecutive GAD alarms over a time window is low. Fig.~\ref{fig:dynamic_FP_oracle_gad_occurrence} confirms this hypothesis, in which we plot the GAD alarms for one instantiation of the $800$ second simulation interval in the absence attacks. It can be seen that the occurrence of alarms is sparse. Thus, the AND fusion rule in an attack detection is effective.

A larger window size $W$ can lower the probability of consecutive alarms within the detection window, resulting in a lower FP rate. However increasing the window size may lead to higher MD rates when an attack is present. Moreover, it also introduces longer delay in detecting the attacks. Thus, the setting of the optimal window size should balance between the FP and MD rates. In what follows, we present simulation results to show the variations of FP and MD rates for different window sizes, which will guide the setting of the window size.

\begin{figure}[!t]
\centering
\begin{subfigure}{0.48\textwidth}
\centering
\includegraphics[width=1\textwidth,trim={0 4cm 0 0}]{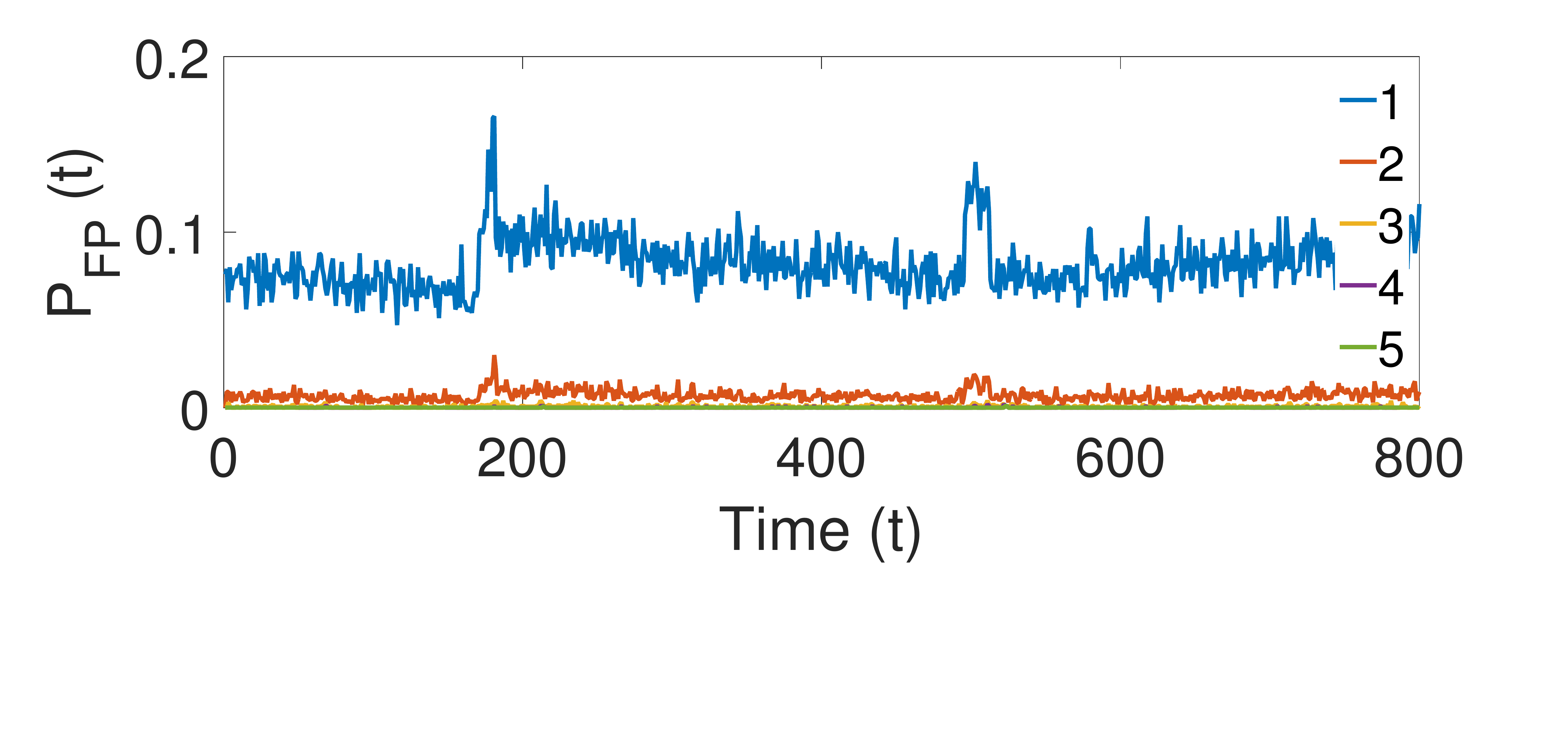}
\caption{FP rate over time for different window sizes.}
\end{subfigure}
~
\begin{subfigure}{0.48\textwidth}
\centering
\includegraphics[width=1\textwidth,trim={0 6cm 0 0}]{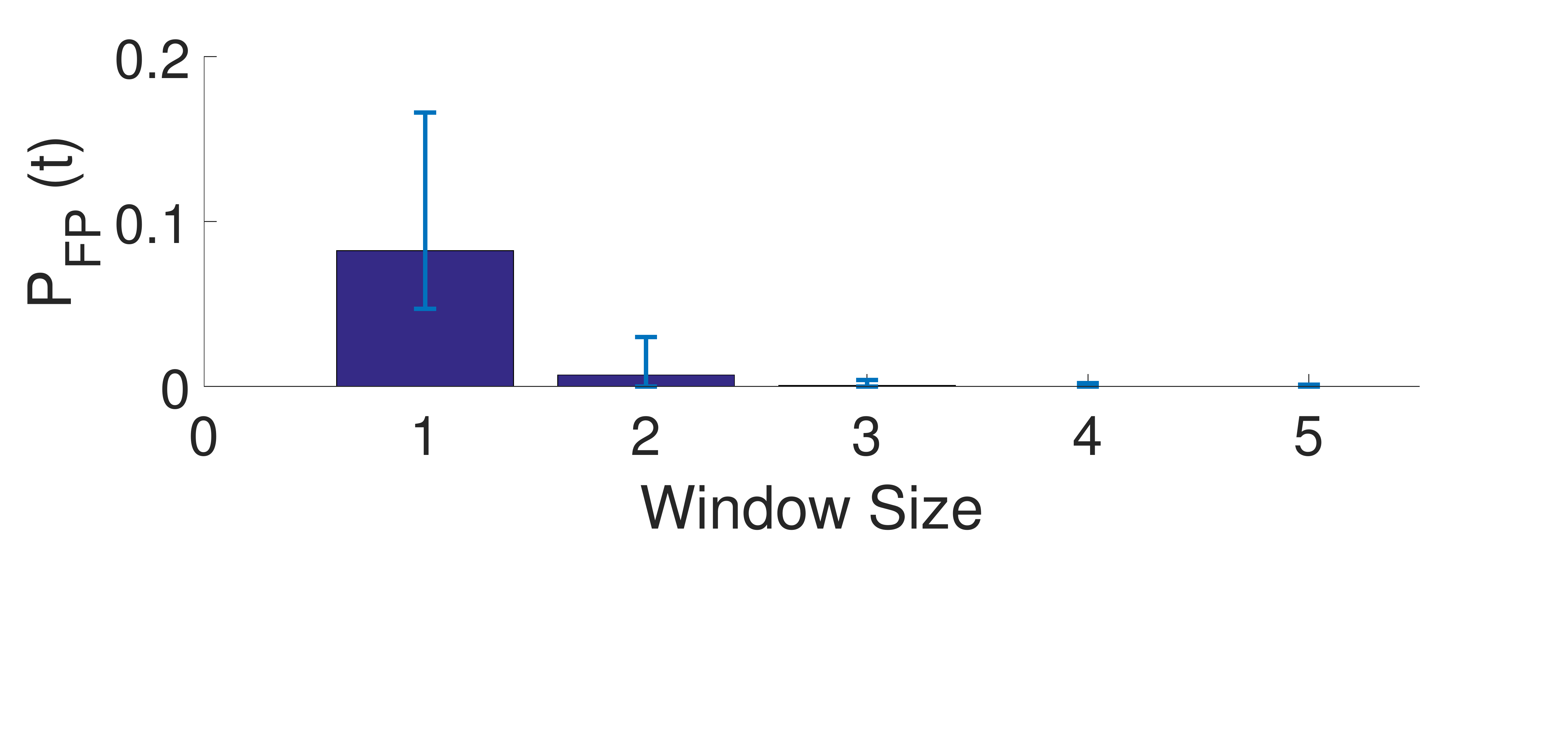}
\caption{FP rate under various attack detection window sizes. Error bars represent maximum and minimum values.}
\label{fig:oracle_previous_last_FP}
\end{subfigure}
~
\begin{subfigure}{0.48\textwidth}
\centering
\includegraphics[width=1\textwidth,trim={0 4cm 0 0}]{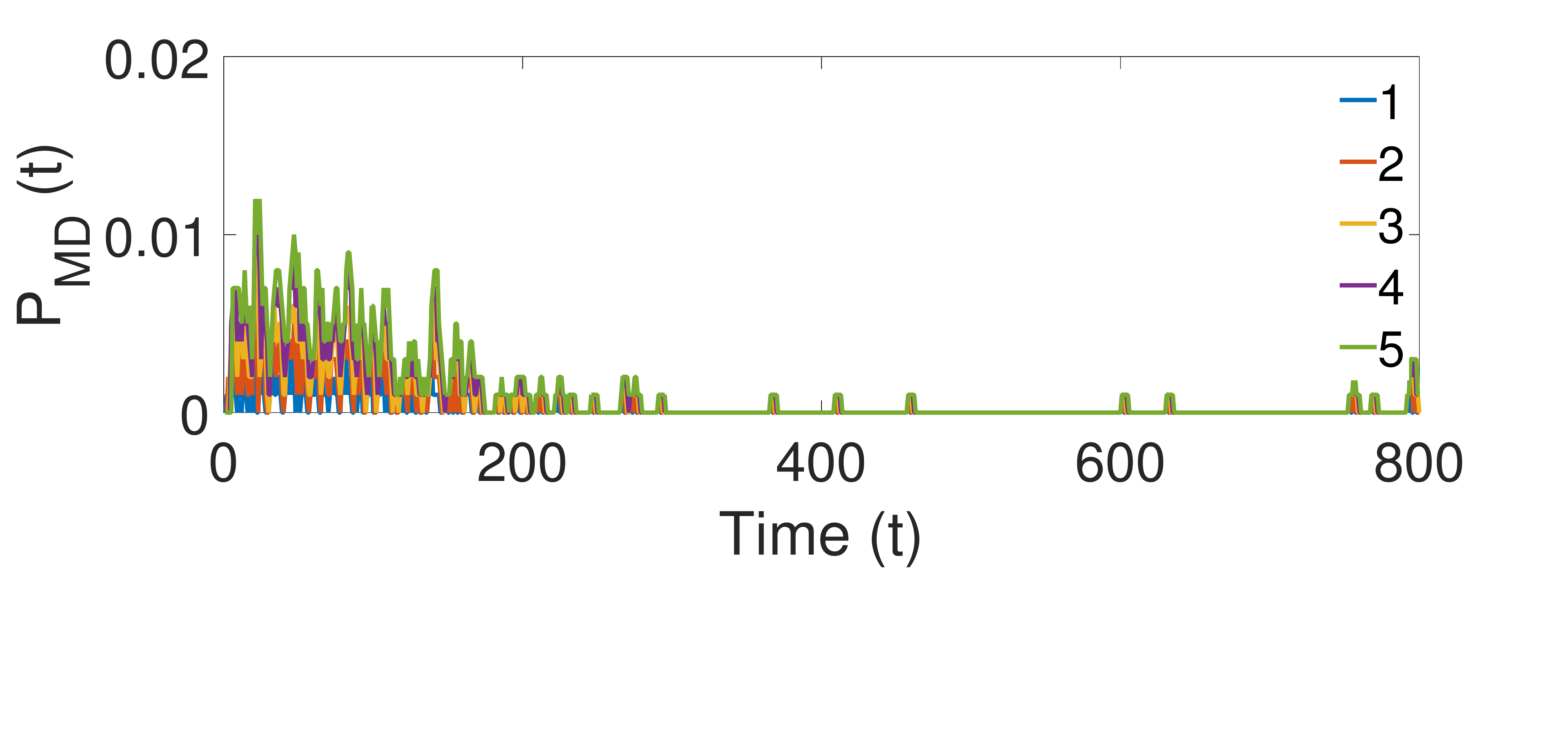}
\caption{MD rate over time for different window sizes under random attack.}
\end{subfigure}
~
\begin{subfigure}{0.48\textwidth}
\centering
\includegraphics[width=1\textwidth,trim={0 4cm 0 0}]{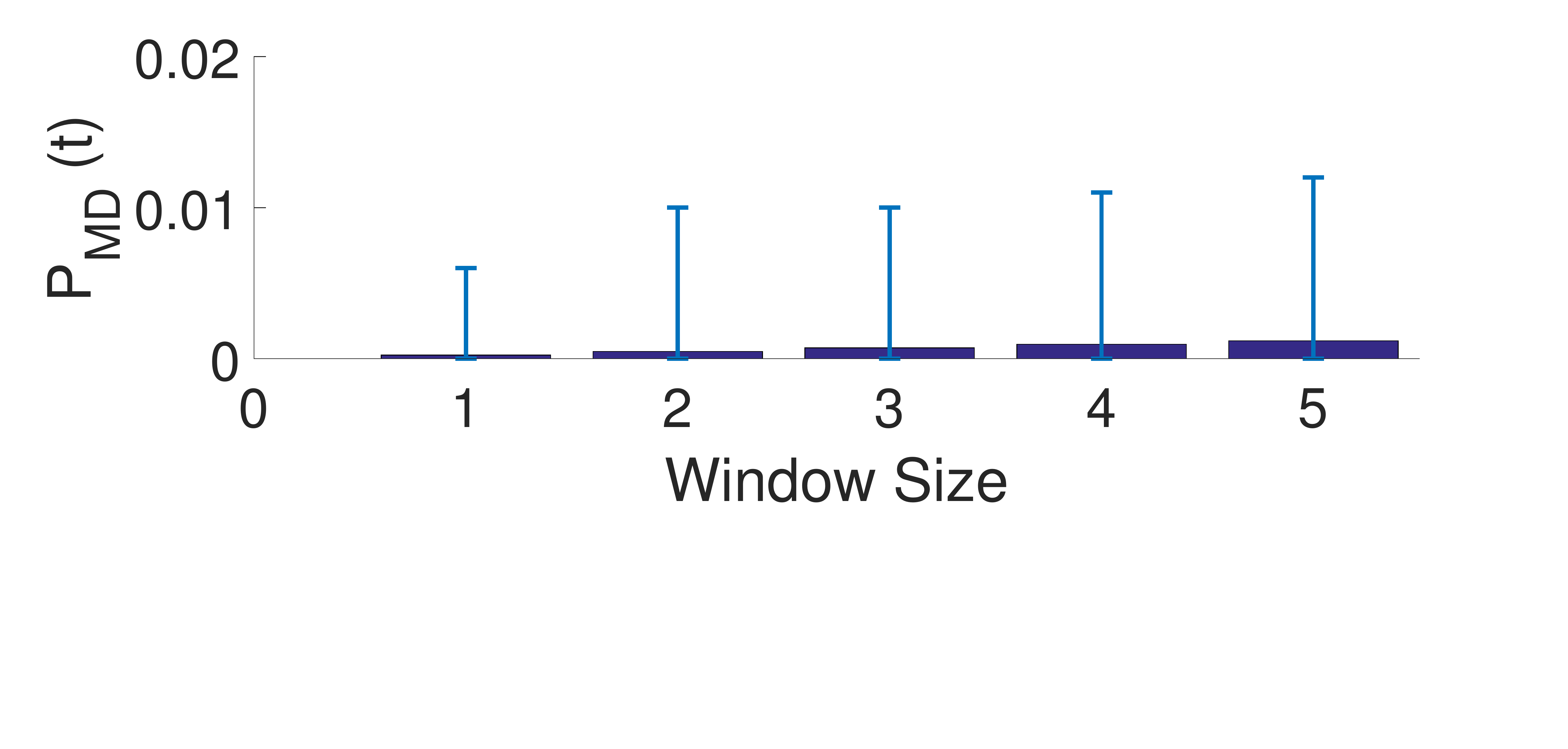}
\caption{MD rate under various attack detection window sizes under random attack. Error bars represent maximum and minimum values.}
\label{fig:oracle_previous_last_MD_Random}
\end{subfigure}
~
\begin{subfigure}{0.48\textwidth}
\centering
\includegraphics[width=1\textwidth,trim={0 8cm 0 0}]{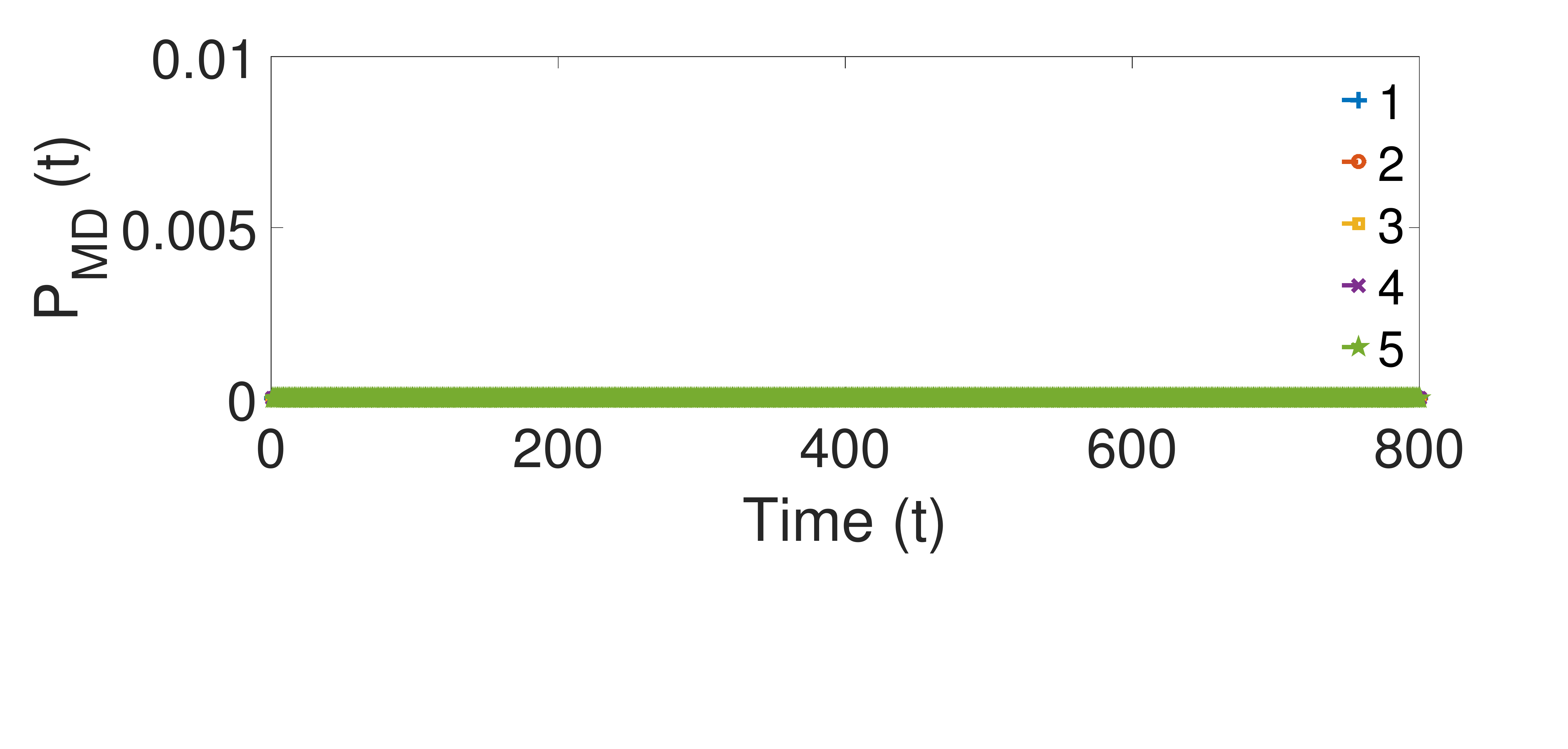}
\caption{MD rate over time for different window sizes under BDD-stealthy attack.}
\end{subfigure}
~
\begin{subfigure}{0.48\textwidth}
\centering
\includegraphics[width=1\textwidth,trim={0 8cm 0 0}]{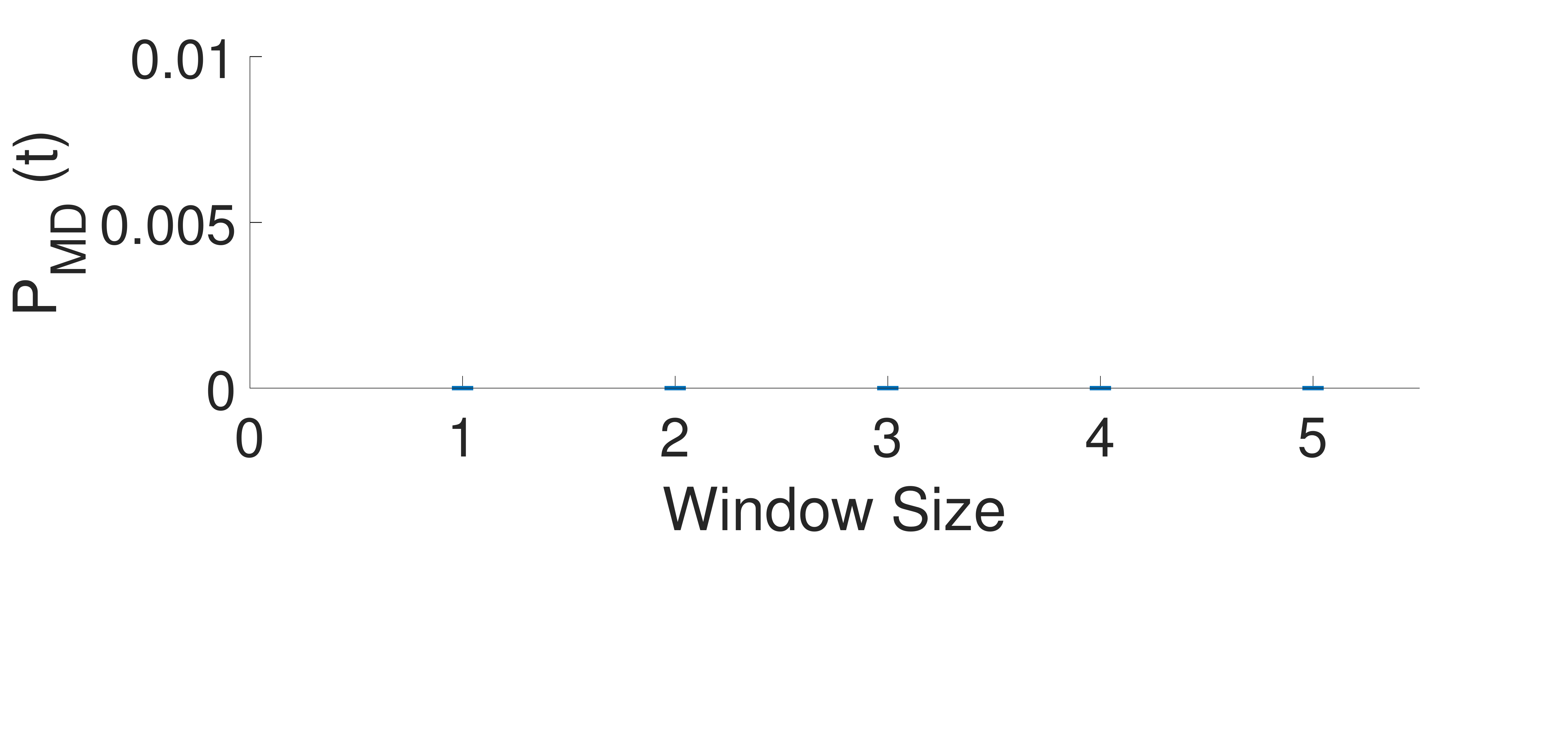}
\caption{MD rates under various attack detection window sizes under BDD-stealthy attack. Error bars represent maximum and minimum values.}
\end{subfigure}
\caption{FP and MD rates for oracle GAD detector under random attacks. }
\label{fig:Oracle_window}
\end{figure}

\begin{figure}[!t]
\centering
\begin{subfigure}{0.48\textwidth}
\centering
\includegraphics[width=1\textwidth,trim={0 5cm 0 0}]{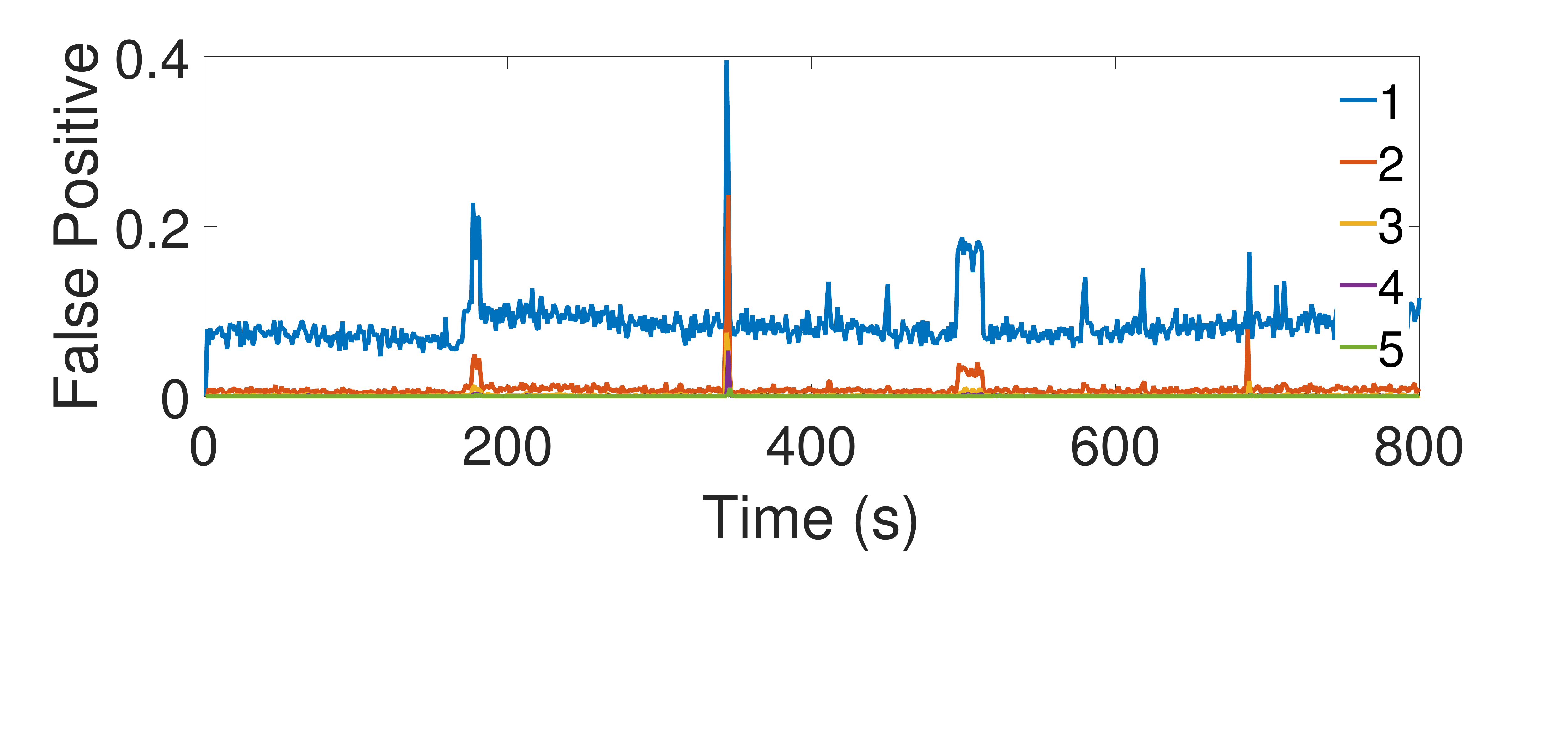}
\caption{FP rate over time for different window sizes.}
\end{subfigure}
~
\begin{subfigure}{0.48\textwidth}
\centering
\includegraphics[width=1\textwidth,trim={0 6cm 0 0}]{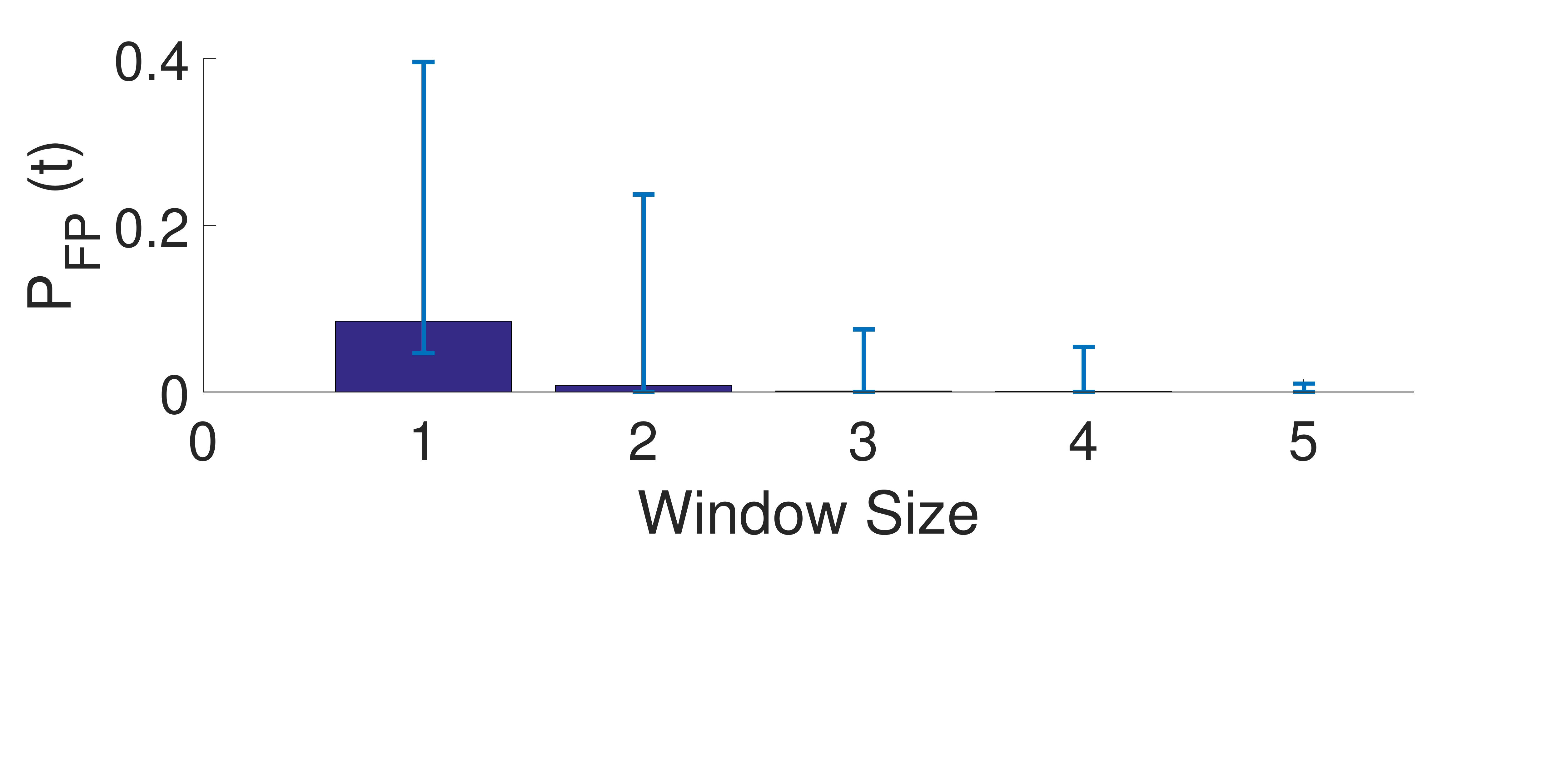}
\caption{FP rate under various attack detection window sizes. Error bars represent maximum and minimum values.}
\label{fig:practical_previous_last_FP}
\end{subfigure}
~
\begin{subfigure}{0.48\textwidth}
\centering
\includegraphics[width=1\textwidth,trim={0 5cm 0 0}]{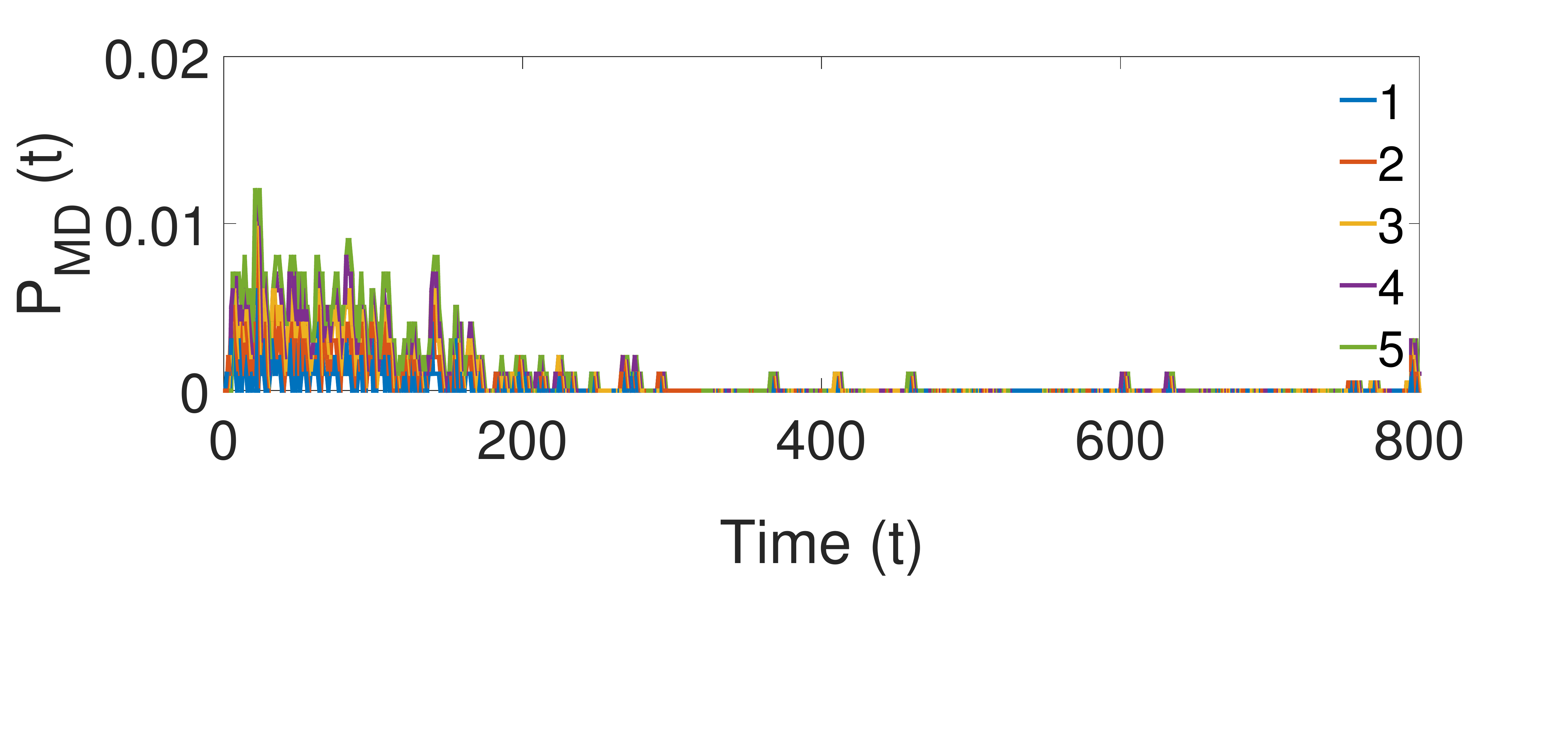}
\caption{MD rate over time for different window sizes under random attack.}
\end{subfigure}
~
\begin{subfigure}{0.48\textwidth}
\centering
\includegraphics[width=1\textwidth,trim={0 6cm 0 0}]{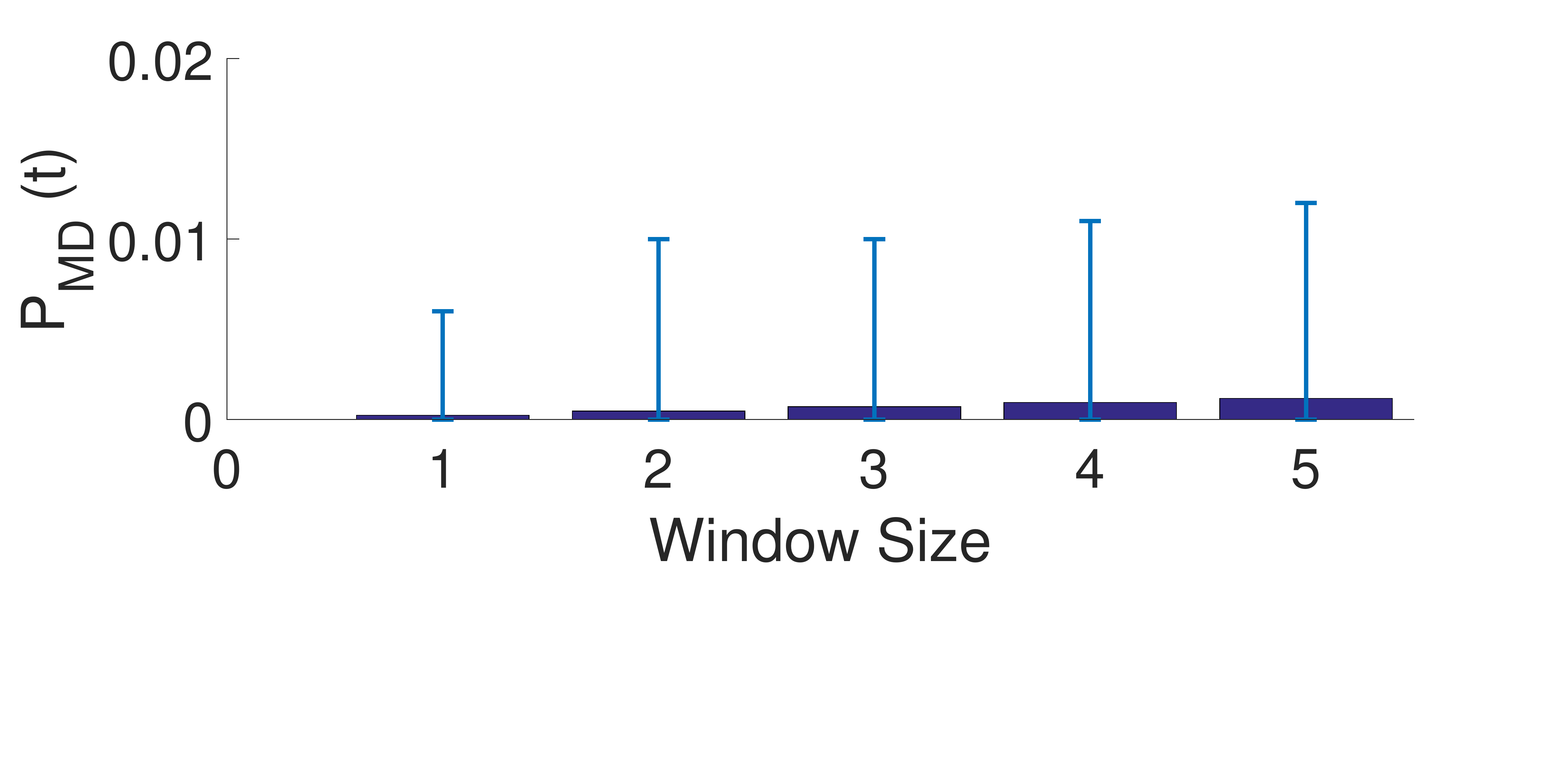}
\caption{MD rate under various attack detection window sizes under random attack. Error bars represent maximum and minimum values.}
\label{fig:practical_previous_last_MD_Random}
\end{subfigure}
~
\begin{subfigure}{0.48\textwidth}
\centering
\includegraphics[width=1\textwidth,trim={0 6cm 0 0}]{figure/dynamic_practical_last_window_MD}
\caption{MD rate over time for different window sizes under BDD-stealthy attack.}
\end{subfigure}
~
\begin{subfigure}{0.48\textwidth}
\centering
\includegraphics[width=1\textwidth,trim={0 6cm 0 0}]{figure/dynamic_practical_last_window_MD_errorbar}
\caption{MD rates under various attack detection window sizes under BDD-stealthy attacks. Error bars represent maximum and minimum values.}
\end{subfigure}
\caption{FP and MD rates for practical GAD detector with BDD-stealthy attacks.}
\label{fig:Practical_window}
\end{figure}

Fig.~\ref{fig:Oracle_window} and  Fig.~\ref{fig:Practical_window} show the FP and MD rates for GAD-W detector under both random and BDD-stealthy attacks. We observe that as the window size increases, the FP rate decreases, whereas the MD rate increases, for the random attacks. We observe that for a window size of $3,$ the average FP rate is $9 \times 10^{-4}$.
The average MD rate for the random attack is $7 \times 10^{-4}$.
Such extremely low of FP and MD rates are acceptable under practical scenarios. Finally, we observe that the MD rates for the BDD-stealthy attacks are very low both under the oracle and practical GAD detectors. This is because the SAD detector is specifically designed to detect BDD-stealthy attacks.

\section{Conclusions}
\label{sec:Conclusion}
In this paper, we studied FDI attacks on train-borne sensor measurements used in railway TPSes. To the best of our 
knowledge, ours is the first effort that has studied TPSes from a cybersecurity perspective.
To account for the safety-criticality of TPS, we adopted the Kerckhoffs's principle and addressed two fundamental problems of importance, namely, characterization of the impact of 
FDI attacks on TPSes, and development of detection techniques for these attacks.
We formulated and analyzed the efficiency and safety attacks that aim to minimize the system energy efficiency and
breach system safety conditions, respectively. 
To detect these attacks, we proposed a global detection system that serializes the proposed BDD and SAD algorithms, both of which may be implemented at a central TPS monitor. 
Furthermore, we proposed an adaptive GAD-W detector that achieves a very low FP rate in the presence of noisy sensor measurements. Our simulation results verified the susceptibility of the TPS setup to the FDI attacks, but these attacks can be detected effectively by the proposed global detection system.

\bibliographystyle{ACM-Reference-Format-Journals}
\bibliography{bibliography}


\begin{thebibliography}{00}


\ifx \showCODEN    \undefined \def \showCODEN     #1{\unskip}     \fi
\ifx \showDOI      \undefined \def \showDOI       #1{{\tt DOI:}\penalty0{#1}\ }
  \fi
\ifx \showISBNx    \undefined \def \showISBNx     #1{\unskip}     \fi
\ifx \showISBNxiii \undefined \def \showISBNxiii  #1{\unskip}     \fi
\ifx \showISSN     \undefined \def \showISSN      #1{\unskip}     \fi
\ifx \showLCCN     \undefined \def \showLCCN      #1{\unskip}     \fi
\ifx \shownote     \undefined \def \shownote      #1{#1}          \fi
\ifx \showarticletitle \undefined \def \showarticletitle #1{#1}   \fi
\ifx \showURL      \undefined \def \showURL       #1{#1}          \fi

\bibitem[\protect\citeauthoryear{??}{Osi}{2015}]%
        {Osiris2015}
 2015.
\newblock \showarticletitle{Osiris \& Urban Rail - {C}omprehensive Approach to
  Making the Save}.
\newblock {\em Mobility - The European Collective Transport Magazine\/} (2015).
\newblock
\newblock
\shownote{\url{http://bit.ly/2pryv7E}.}


\bibitem[\protect\citeauthoryear{??}{Rai}{2016}]%
        {RailTechnical}
 2016.
\newblock \showarticletitle{ELECTRIC TRACTION POWER SUPPLIES}.
\newblock {\em Railway Technical Web Pages\/} (2016).
\newblock
\newblock
\shownote{\url{http://www.railway-technical.com/etracp.shtml}.}


\bibitem[\protect\citeauthoryear{Abrahamsson}{Abrahamsson}{2012}]%
        {AbrahamssonThesis2012}
{L. Abrahamsson}. 2012.
\newblock {\em Optimal Railroad Power Supply System Operation and Design}.
\newblock PhD Thesis, KTH – Sweden.
\newblock


\bibitem[\protect\citeauthoryear{Açikbaş and Söylemez}{Açikbaş and
  Söylemez}{2007}]%
        {acikbas2007parameters}
{S. Açikbaş} {and} {M.T. Söylemez}. 2007.
\newblock \showarticletitle{Parameters affecting braking energy recuperation
  rate in DC rail transit}. In {\em ASME/IEEE Joint Rail Conf. \& Internal
  Combustion Engine Division Spring Technical Conf.}
\newblock


\bibitem[\protect\citeauthoryear{Alstom}{Alstom}{2001}]%
        {ALSTOM}
{Alstom}. 2001.
\newblock \showarticletitle{{ERTMS/ETCS} On-Board {ALSTOM} Solution}.
\newblock  (2001).
\newblock
\newblock
\shownote{\url{https://bit.ly/1OOb38f}.}


\bibitem[\protect\citeauthoryear{Amin, Litrico, Sastry, and Bayen}{Amin
  et~al\mbox{.}}{2013}]%
        {amin2013cyber}
{S. Amin}, {X. Litrico}, {S. Sastry}, {and} {A.M. Bayen}. 2013.
\newblock \showarticletitle{Cyber security of water scada systems—part {I}:
  {A}nalysis and experimentation of stealthy deception attacks}.
\newblock {\em IEEE Trans. Control Syst. Technol.\/} {21}, 5 (Sept. 2013),
  1963--1970.
\newblock


\bibitem[\protect\citeauthoryear{Arboleya, Diaz, and Coto}{Arboleya
  et~al\mbox{.}}{2012}]%
        {ArboleyaCotoTVT2012}
{P. Arboleya}, {G. Diaz}, {and} {M. Coto}. 2012.
\newblock \showarticletitle{Unified {AC/DC} Power Flow for Traction Systems:
  {A} New Concept}.
\newblock {\em IEEE Trans. Veh. Technol\/} {61}, 6 (July 2012), 2421--2430.
\newblock


\bibitem[\protect\citeauthoryear{Arboleya, Mohamed, González-Morán, and
  El-Sayed}{Arboleya et~al\mbox{.}}{2016}]%
        {ArboleyaBFS2015}
{P. Arboleya}, {B. Mohamed}, {C. González-Morán}, {and} {I. El-Sayed}. 2016.
\newblock \showarticletitle{BFS Algorithm for Voltage-Constrained Meshed DC
  Traction Networks With Nonsmooth Voltage-Dependent Loads and Generators}.
\newblock {\em IEEE Trans. Power Syst.\/}  {31} (2016), 1526--1536.
\newblock


\bibitem[\protect\citeauthoryear{Cai, Irving, and Case}{Cai
  et~al\mbox{.}}{1995}]%
        {CaiIterative1995}
{Y. Cai}, {M.R. Irving}, {and} {S.H. Case}. 1995.
\newblock \showarticletitle{Iterative techniques for the solution of complex
  {DC}-rail-traction systems including regenerative braking}.
\newblock {\em IEE Proc. Generation, Transmission and Distribution\/} {142}, 5
  (1995).
\newblock


\bibitem[\protect\citeauthoryear{C{\'a}rdenas, Amin, Lin, Huang, Huang, and
  Sastry}{C{\'a}rdenas et~al\mbox{.}}{2011}]%
        {cardenas2011attacks}
{A.A. C{\'a}rdenas}, {S. Amin}, {Z. Lin}, {Y. Huang}, {C. Huang}, {and} {S.
  Sastry}. 2011.
\newblock \showarticletitle{Attacks against process control systems: {R}isk
  assessment, detection, and response}. In {\em Proc. ACM AsiaCCS}.
\newblock


\bibitem[\protect\citeauthoryear{David}{David}{2015}]%
        {BBC_Regen2015}
{K.G. David}. 2015.
\newblock \showarticletitle{The train that powers its station}.
\newblock  (2015).
\newblock
\newblock
\shownote{\url{http://bbc.in/1KRROZK}.}


\bibitem[\protect\citeauthoryear{Depuru, Wang, and Devabhaktuni}{Depuru
  et~al\mbox{.}}{2011}]%
        {Theft2011}
{S.S.S.R. Depuru}, {L. Wang}, {and} {V. Devabhaktuni}. 2011.
\newblock \showarticletitle{Electricity theft: {O}verview, issues, prevention
  and a smart meter based approach to control theft}.
\newblock {\em Energy Policy\/} {39}, 2 (Feb. 2011), 1007--1015.
\newblock


\bibitem[\protect\citeauthoryear{Fletcher}{Fletcher}{1991}]%
        {Fletcher1991}
{R.G. Fletcher}. 1991.
\newblock \showarticletitle{Regenerative equipment for railway rolling stock}.
\newblock {\em Power Engineering Journal\/} {5}, 3 (May 1991), 105--114.
\newblock
\showISSN{0950-3366}


\bibitem[\protect\citeauthoryear{Gabrielle}{Gabrielle}{2014}]%
        {moscowtimes}
{T.F. Gabrielle}. 2014.
\newblock \showarticletitle{Deadly Derailment in {Moscow} Metro}.
\newblock  (2014).
\newblock
\newblock
\shownote{\url{http://bit.ly/2d5D7dy}.}


\bibitem[\protect\citeauthoryear{González-Gil, Palacin, Batty, and
  Powell}{González-Gil et~al\mbox{.}}{2014}]%
        {GonzálezGil2014509}
{A. González-Gil}, {R. Palacin}, {P. Batty}, {and} {J.P. Powell}. 2014.
\newblock \showarticletitle{A systems approach to reduce urban rail energy
  consumption}.
\newblock {\em Energy Conversion and Management\/}  {80} (2014), 509 -- 524.
\newblock


\bibitem[\protect\citeauthoryear{Hartwig, Grimm, Meyer~zu H{\"o}rste, and
  Lemmer}{Hartwig et~al\mbox{.}}{2006}]%
        {balise}
{K. Hartwig}, {M. Grimm}, {M. Meyer~zu H{\"o}rste}, {and} {K. Lemmer}. 2006.
\newblock \showarticletitle{Requirements for safety relevant positioning
  applications in rail traffic - {A} demonstrator for a train borne navigation
  platform called ``{D}emo{O}rt''}.
\newblock  (2006).
\newblock
\newblock
\shownote{\url{http://elib.dlr.de/21252/1/wcrr.pdf}.}


\bibitem[\protect\citeauthoryear{Jinsub and Lang}{Jinsub and Lang}{2013}]%
        {KimTongTopology2013}
{K. Jinsub} {and} {T. Lang}. 2013.
\newblock \showarticletitle{On Topology Attack of a Smart Grid: {U}ndetectable
  Attacks and Countermeasures}.
\newblock {\em IEEE J. Sel. Areas Commun.\/} {31}, 7 (July 2013), 1294--1305.
\newblock


\bibitem[\protect\citeauthoryear{Karnouskos}{Karnouskos}{2011}]%
        {karnouskos2011}
{S. Karnouskos}. 2011.
\newblock \showarticletitle{Stuxnet worm impact on industrial cyber-physical
  system security}. In {\em Conf. IEEE Industrial Electronics Society}.
\newblock


\bibitem[\protect\citeauthoryear{Karri, Rajendran, Rosenfeld, and
  Tehranipoor}{Karri et~al\mbox{.}}{2010}]%
        {HardwareTrojans2010}
{R. Karri}, {J. Rajendran}, {K. Rosenfeld}, {and} {M. Tehranipoor}. 2010.
\newblock \showarticletitle{Trustworthy Hardware: {I}dentifying and Classifying
  Hardware Trojans}.
\newblock {\em Computer\/} {43}, 10 (Oct 2010), 39--46.
\newblock


\bibitem[\protect\citeauthoryear{Kune, Backes, Clark, Kramer, Reynolds, Fu,
  Kim, and Xu}{Kune et~al\mbox{.}}{2013}]%
        {kune2013ghost}
{D.F. Kune}, {J. Backes}, {S.S. Clark}, {D. Kramer}, {M. Reynolds}, {K. Fu},
  {Y. Kim}, {and} {W. Xu}. 2013.
\newblock \showarticletitle{Ghost talk: {M}itigating {EMI} signal injection
  attacks against analog sensors}. In {\em IEEE Symp. Security and Privacy}.
\newblock


\bibitem[\protect\citeauthoryear{Liu, Ning, and Reiter}{Liu
  et~al\mbox{.}}{2009}]%
        {LiuNingReiter2009}
{Y. Liu}, {P. Ning}, {and} {M.K. Reiter}. 2009.
\newblock \showarticletitle{False Data Injection Attacks Against State
  Estimation in Electric Power Grids}. In {\em ACM CCS}.
\newblock


\bibitem[\protect\citeauthoryear{McDaniel and McLaughlin}{McDaniel and
  McLaughlin}{2009}]%
        {SmartMeterSecurity2009}
{P. McDaniel} {and} {S. McLaughlin}. 2009.
\newblock \showarticletitle{Security and Privacy Challenges in the Smart Grid}.
\newblock {\em IEEE Security Privacy\/} {7}, 3 (2009), 75--77.
\newblock


\bibitem[\protect\citeauthoryear{Mike}{Mike}{2009}]%
        {davis2009}
{D. Mike}. 2009.
\newblock \showarticletitle{Recoverable Advanced Metering Infrastructure}. In
  {\em Proc. Black Hat Technical Security Conference}.
\newblock


\bibitem[\protect\citeauthoryear{Miyatake and Ko}{Miyatake and Ko}{2010}]%
        {Miyatake2010}
{M. Miyatake} {and} {H. Ko}. 2010.
\newblock \showarticletitle{Optimization of Train Speed Profile for Minimum
  Energy Consumption}.
\newblock {\em IEEJ Transactions on Electrical and Electronic Engineering\/}
  {5}, 3 (2010), 263--269.
\newblock


\bibitem[\protect\citeauthoryear{Okada, Koseki, and Hisatomi}{Okada
  et~al\mbox{.}}{2004}]%
        {OkadaKoseki2004}
{Y. Okada}, {T. Koseki}, {and} {K. Hisatomi}. 2004.
\newblock \showarticletitle{Power management control in {DC}-electrified
  railways for the regenerative braking systems of electric trains}.
\newblock {\em Advances in Transport\/}  {15} (2004), 919--929.
\newblock


\bibitem[\protect\citeauthoryear{Pires, Nabeta, and Cardoso}{Pires
  et~al\mbox{.}}{2007}]%
        {PiresICCG2007}
{C.L. Pires}, {S.I. Nabeta}, {and} {J.R. Cardoso}. 2007.
\newblock \showarticletitle{{ICCG} method applied to solve {DC} traction load
  flow including earthing models}.
\newblock {\em IET Electric Power Applications\/} {1}, 2 (March 2007),
  193--198.
\newblock


\bibitem[\protect\citeauthoryear{Raghunathan, Wada, Ueda, and
  Takahashi}{Raghunathan et~al\mbox{.}}{2014}]%
        {RaghunathanCOMPRAIL2014}
{A.U. Raghunathan}, {T. Wada}, {K. Ueda}, {and} {S. Takahashi}. 2014.
\newblock \showarticletitle{Minimizing Energy Consumption in Railways by
  Voltage Control on Substations}. In {\em Proc. Intl. Conf. Railway
  Engineering Design and Optimization}.
\newblock


\bibitem[\protect\citeauthoryear{Rahman, Al-Shaer, and Kavasseri}{Rahman
  et~al\mbox{.}}{2014}]%
        {RahmanFormalModel2014}
{M.A. Rahman}, {E. Al-Shaer}, {and} {R.G. Kavasseri}. 2014.
\newblock \showarticletitle{A Formal Model for Verifying the Impact of Stealthy
  Attacks on Optimal Power Flow in Power Grids}. In {\em Proc. ACM/IEEE ICCPS}.
\newblock


\bibitem[\protect\citeauthoryear{Security}{Security}{2011}]%
        {insider2011}
{Homeland Security}. 2011.
\newblock \showarticletitle{{U.S. DHS.} Insider threat to utilities}.
\newblock  (2011).
\newblock
\newblock
\shownote{\url{https://bit.ly/1YPFoZH}.}


\bibitem[\protect\citeauthoryear{Shuai, Tao, and Roberts}{Shuai
  et~al\mbox{.}}{2015}]%
        {SuTangRoberts2015}
{S. Shuai}, {T. Tao}, {and} {C. Roberts}. 2015.
\newblock \showarticletitle{A Cooperative Train Control Model for Energy
  Saving}.
\newblock {\em IEEE Trans. Intell. Transp. Syst.\/} {16}, 2 (April 2015),
  622--631.
\newblock


\bibitem[\protect\citeauthoryear{Shuai, Tao, Xiang, and Ziyou}{Shuai
  et~al\mbox{.}}{2014}]%
        {SuTang2014}
{S. Shuai}, {T. Tao}, {L. Xiang}, {and} {G. Ziyou}. 2014.
\newblock \showarticletitle{Optimization of Multitrain Operations in a Subway
  System}.
\newblock {\em IEEE Trans. Intell. Transp. Syst.\/} {15}, 2 (April 2014),
  673--684.
\newblock


\bibitem[\protect\citeauthoryear{{S}mitt}{{S}mitt}{2016}]%
        {sensor_accuracy}
{Mors {S}mitt}. 2016.
\newblock \showarticletitle{Traction energy measuring solutions}.
\newblock  (2016).
\newblock
\newblock
\shownote{\url{http://bit.ly/2q5OUuZ}.}


\bibitem[\protect\citeauthoryear{{SMRT}}{{SMRT}}{2015}]%
        {SMRT15}
{{SMRT}}. 2015.
\newblock Press Release.
\newblock   (July 2015).
\newblock
\newblock
\shownote{\url{https://bit.ly/1RxGBSk}.}


\bibitem[\protect\citeauthoryear{Sottile}{Sottile}{2011}]%
        {frank2011}
{F. Sottile}. 2011.
\newblock {\em Real solutions to equations from geometry}. Vol.~57.
\newblock American Mathematical Society Providence, RI.
\newblock


\bibitem[\protect\citeauthoryear{Symantec}{Symantec}{2014}]%
        {dragonfly2014}
{Symantec}. 2014.
\newblock \showarticletitle{Dragonfly: Cyberespionage Attacks Against Energy
  Suppliers}.
\newblock  (2014).
\newblock
\newblock
\shownote{\url{http://symc.ly/2cowemc}.}


\bibitem[\protect\citeauthoryear{Talukdar and Koo}{Talukdar and Koo}{1977}]%
        {Talukdar1977}
{S.N. Talukdar} {and} {R.L. Koo}. 1977.
\newblock \showarticletitle{The analysis of electrified ground transportation
  networks}.
\newblock {\em IEEE Trans. Power App. Syst.\/} {96}, 1 (1977).
\newblock


\bibitem[\protect\citeauthoryear{Teixeira, Sandberg, Dan, and
  Johansson}{Teixeira et~al\mbox{.}}{2012}]%
        {OPFClosingLoop2012}
{A. Teixeira}, {H. Sandberg}, {G. Dan}, {and} {K.H. Johansson}. 2012.
\newblock \showarticletitle{Optimal power flow: {C}losing the loop over
  corrupted data}. In {\em Proc. ACC}.
\newblock


\bibitem[\protect\citeauthoryear{{The Economic Times -- Railways}}{{The
  Economic Times -- Railways}}{2012}]%
        {GPS2012}
{{The Economic Times -- Railways}}. 2012.
\newblock \showarticletitle{Indian {R}ailways to launch real-time train
  tracking via {G}oogle maps}.
\newblock  (2012).
\newblock
\newblock
\shownote{\url{https://bit.ly/1OIcMOe}.}


\bibitem[\protect\citeauthoryear{{Transport for London}}{{Transport for
  London}}{2008}]%
        {London_Metro2008}
{{Transport for London}}. 2008.
\newblock \showarticletitle{{LU} Carbon footprint report 2008}.
\newblock  (2008).
\newblock
\newblock
\shownote{\url{http://bit.ly/2pgb8xb}.}


\bibitem[\protect\citeauthoryear{Wood and Wollenberg}{Wood and
  Wollenberg}{1996}]%
        {wood1996power}
{A.J. Wood} {and} {B.F. Wollenberg}. 1996.
\newblock {\em {Power Generation, Operation, and Control}}.
\newblock A Wiley-Interscience.
\newblock


\bibitem[\protect\citeauthoryear{Yadav}{Yadav}{2013}]%
        {Overhead}
{Anil Yadav}. 2013.
\newblock \showarticletitle{Traction choices: {O}verhead ac vs third rail dc}.
\newblock  (2013).
\newblock
\newblock
\shownote{\url{http://bit.ly/2orprPW}.}


\bibitem[\protect\citeauthoryear{Yanling, Zuyi, and Kui}{Yanling
  et~al\mbox{.}}{2011}]%
        {RenLoadRedis2011}
{Y. Yanling}, {L. Zuyi}, {and} {R. Kui}. 2011.
\newblock \showarticletitle{Modeling Load Redistribution Attacks in Power
  Systems}.
\newblock {\em IEEE Trans. Smart Grid\/} {2}, 2 (2011).
\newblock


\end{thebibliography}

\section*{Appendix~A: BDD Threshold}

In this Appendix, we present how to set the BDD threshold $\tau$ to ensure that the false positive rate is maintained at a certain level. 

Recall that the expression for BDD residual is given by $r = ||{\zv} - {\Hm} \hat{{\vv}}||,$ where ${\zv} = {\Hm} \vv + \nv ,$ $\hat{{\vv}} = ({\Hm}^T \Sigmam \Hm)^{-1} {\Hm}^T \Sigmam {\zv}.$ Substituting the expression of $\hat{{\vv}},$ we obtain:
\begin{align}
r & = ||\zv  - \Hm ({\Hm}^T \Sigmam \Hm)^{-1} {\Hm}^T \Sigmam \zv|| \nonumber \\ 
& = ||\Hm {{\vv}} + \nv  - \Hm ({\Hm}^T \Sigmam \Hm)^{-1} {\Hm}^T \Sigmam (\Hm {\vv}  + \nv )||  \nonumber \\
& = || (\Id - \Gammam) \nv ||, \label{eqn:BDD_res}
\end{align}
where $\Gammam = \Hm (\Hm^T \Wm \Hm)^{-1} {\Hm}^T \Sigmam.$ 
From \eqref{eqn:BDD_res}, $r$ follows a chi-square distribution, since
the noise $\nv$ is Gaussian. To maintain a certain FP rate $\alpha,$ the BDD threshold can be set by 
solving $\mathbb{P} (r \geq \tau) = \alpha.$

\end{document}